\documentclass[3p,,preprint,10pt]{elsarticle}
\makeatletter\if@twocolumn\PassOptionsToPackage{switch}{lineno}\else\fi\makeatother

\usepackage{booktabs} 

\usepackage{tabulary,xcolor}
\usepackage{amsfonts,amsmath,amssymb}
\usepackage[T1]{fontenc}
\makeatletter
\let\save@ps@pprintTitle\ps@pprintTitle
\def\ps@pprintTitle{\save@ps@pprintTitle\gdef\@oddfoot{\footnotesize\itshape \null\hfill\today}}
\def\hlinewd#1{%
  \noalign{\ifnum0=`}\fi\hrule \@height #1%
  \futurelet\reserved@a\@xhline}

\AtBeginDocument{\ifNAT@numbers \biboptions{sort&compress}\fi}
\makeatother

  \renewenvironment{abstract}{\global\setbox\absbox=\vbox\bgroup
    \hsize=\textwidth%
  \noindent\unskip\textbf{}
   \par\medskip\noindent\unskip\ignorespaces}
   {\egroup}
  
\usepackage{ifluatex}
\ifluatex
\usepackage{fontspec}
\defaultfontfeatures{Ligatures=TeX}
\usepackage[]{unicode-math}
\unimathsetup{math-style=TeX}
\else 
\usepackage[utf8]{inputenc}
\fi 
\ifluatex\else\usepackage{stmaryrd}\fi

\usepackage{url,multirow,morefloats,floatflt,cancel,tfrupee}
\makeatletter

\AtBeginDocument{\@ifpackageloaded{textcomp}{}{\usepackage{textcomp}}}
\makeatother
\usepackage{colortbl}
\usepackage{xcolor}
\usepackage{pifont}
\usepackage[nointegrals]{wasysym}
\makeatletter

\def\mcWidth#1{\csname TY@F#1\endcsname+\tabcolsep}

\def\cAlignHack{\rightskip\@flushglue\leftskip\@flushglue\parindent\z@\parfillskip\z@skip}
\def\rAlignHack{\rightskip\z@skip\leftskip\@flushglue \parindent\z@\parfillskip\z@skip}

\@ifundefined{etal}{}{}

\usepackage{ifxetex}
\ifxetex\else\if@twocolumn\@ifpackageloaded{stfloats}{}{\usepackage{dblfloatfix}}\fi\fi

\AtBeginDocument{
\expandafter\ifx\csname eqalign\endcsname\relax
\def\eqalign#1{\null\vcenter{\def\\{\cr}\openup\jot\m@th
  \ialign{\strut$\displaystyle{##}$\hfil&$\displaystyle{{}##}$\hfil
      \crcr#1\crcr}}\,}
\fi
}

\AtBeginDocument{%
  \@ifpackageloaded{endfloat}%
   {\renewcommand\efloat@iwrite[1]{\immediate\expandafter\protected@write\csname efloat@post#1\endcsname{}}}{\newif\ifefloat@tables}%
}%

\def\BreakURLText#1{\@tfor\brk@tempa:=#1\do{\brk@tempa\hskip0pt}}
\let\lt=<
\let\gt=>
\def\processVert{\ifmmode|\else\textbar\fi}

\@ifundefined{subparagraph}{
\def\subparagraph{\@startsection{paragraph}{5}{2\parindent}{0ex plus 0.1ex minus 0.1ex}%
{0ex}{\normalfont\small\itshape}}%
}{}

\newcommand\role[1]{\unskip}
\newcommand\aucollab[1]{\unskip}
  
\@ifundefined{tsGraphicsScaleX}{\gdef\tsGraphicsScaleX{1}}{}
\@ifundefined{tsGraphicsScaleY}{\gdef\tsGraphicsScaleY{.9}}{}
\def\checkGraphicsWidth{\ifdim\Gin@nat@width>\linewidth
	\tsGraphicsScaleX\linewidth\else\Gin@nat@width\fi}

\def\checkGraphicsHeight{\ifdim\Gin@nat@height>.9\textheight
	\tsGraphicsScaleY\textheight\else\Gin@nat@height\fi}

\def\fixFloatSize#1{}%\@ifundefined{processdelayedfloats}{\setbox0=\hbox{\includegraphics{#1}}\ifnum\wd0<\columnwidth\relax\renewenvironment{figure*}{\begin{figure}}{\end{figure}}\fi}{}}
\let\ts@includegraphics\includegraphics

\def\inlinegraphic[#1]#2{{\edef\@tempa{#1}\edef\baseline@shift{\ifx\@tempa\@empty0\else#1\fi}\edef\tempZ{\the\numexpr(\numexpr(\baseline@shift*\f@size/100))}\protect\raisebox{\tempZ pt}{\ts@includegraphics{#2}}}}

\AtBeginDocument{\def\includegraphics{\@ifnextchar[{\ts@includegraphics}{\ts@includegraphics[width=\checkGraphicsWidth,height=\checkGraphicsHeight,keepaspectratio]}}}

\DeclareMathAlphabet{\mathpzc}{OT1}{pzc}{m}{it}

\def\URL#1#2{\@ifundefined{href}{#2}{\href{#1}{#2}}}

\def\UrlOrds{\do\*\do\-\do\~\do\'\do\"\do\-}%
\g@addto@macro{\UrlBreaks}{\UrlOrds}

\edef\fntEncoding{\f@encoding}

\makeatother

\newif\ifmultipleabstract\multipleabstractfalse%
\usepackage{subcaption}
\usepackage{adjustbox}
\usepackage{graphicx}
\usepackage{booktabs}
\usepackage{multirow}
\usepackage{float} % Need if position of table needs float
\usepackage[flushleft]{threeparttable}

\usepackage{amssymb}
\usepackage{mathtools}
\usepackage{amsmath} % per teoremi
\usepackage{amsthm}

\newtheoremstyle{break}
{\topsep}{\topsep}%
{\itshape}{}%
{\bfseries}{}%
{\newline}{}%
\theoremstyle{break}
\newtheorem{theorem}{Theorem}

\newtheorem{proposition}[]{Proposition}
\newtheorem{lemma}[theorem]{Remark}

\newtheorem{definition}[]{Definition}
\newtheorem{assumption}[]{A.}
\AtBeginEnvironment{proof}{\footnotesize}

\usepackage{dutchcal} % new mathcal for capital and normal font
\DeclareFontFamily{OT1}{pzc}{}
\DeclareFontShape{OT1}{pzc}{m}{}{<-> s * [1.10] pzcmi7t}{}
\DeclareMathAlphabet{\mathpzc}{OT1}{pzc}{m}{it}

\journal{Physica A}

\begin{document}

	\begin{frontmatter}
		
		
		
	%	\title{A latent-variable model for daily trip mobility networks applied to real-world origin-destination visitation flows}
		% 	\title{A diffusive transport evolution of visiting patterns in daily trip mobility networks}
		\title{A visit generation process for human mobility random graphs with location-specific latent-variables: from land use to travel demand}
		
		
		\author[inst1,inst2]{Fabio Vanni}
		
		\affiliation[inst1]{organization={Department of Economics, University of Insubria},%Department and Organization
			addressline={Via Monte Generoso 71}, 
			city={Varese},
			postcode={21100}, 
			%state={State One},
			country={Italy}}
		
		%\author[inst2]{Author Two}
		%\author[inst1,inst2]{Author Three}
		
		\affiliation[inst2]{organization={Center for nonlinear Science, Department of Physics, University of North Texas},%Department and Organization
			addressline={Union Circle}, 
			city={Denton},
			postcode={1155}, 
			state={Texas},
			country={USA}}
		
%		\affiliation[inst3]{organization={Department of Economics, Sant'Anna School of Advanced Studies},%Department and Organization
%			addressline={Piazza Martiri della Liberta', 33}, 
%			city={Pisa},
%			postcode={56127}, 
%			%state={State One},
%			country={Italy}}            
		
		\begin{abstract}
This research introduces a  mathematical framework to comprehending human mobility patterns, integrating mathematical modeling and economic analysis. The study focuses on latent-variable networks, investigating the dynamics of human mobility using stochastic models. By examining actual origin-destination data, the research reveals scaling relations and uncovers the economic implications of mobility patterns, such as the income elasticity of travel demand. The mathematical analysis commences with the development of a stochastic model based on inhomogeneous random graphs to construct a visitation model with multipurpose drivers for travel demand. A directed multigraph with weighted edges is considered, incorporating trip costs and labels to represent factors like distance traveled and travel time. 	The study gains insights into the structural properties and dynamic correlations of human mobility networks, to derive analytical and computational solutions for key network metrics, including scale-free behavior of the strength and degree distribution, together with the estimation of assortativity and clustering coefficient. Additionally, the model's validity is assessed through a real-world case study of the New York metropolitan area. The analysis of this data exposes clear scaling relations in commuting patterns, confirming theoretical predictions and validating the efficacy of the mathematical model. The model further explains a series of scaling behaviors in origin-destination flows among areas of a region, successfully reproducing statistical regularities observed in real-world cases using extensive human mobility datasets. In particular, the model's application to estimating income elasticity of travel demand bears significant implications for urban and transport economics. 
		\end{abstract}
		
%		\begin{graphicalabstract}
%			\includegraphics{img/City_diffusion}
%		\end{graphicalabstract}
		
%		\begin{highlights}
%			\item Presents a data-driven model for human mobility networks based on origin-destination structure,
%			\item  Mathematical formalism of compound renewal processes for human mobility dynamics using latent-variable networks,
%			\item Identifies scaling laws in visit frequencies and correlations between visit frequencies and trip costs
%			\item Clear scaling relations emerge between latent variables and travel demand from land-use to income elasticity,
%			\item Derives key metrics for human mobility networks validated through New York case study, with policy implications for urban transportation.
%		\end{highlights}
		
		
	\end{frontmatter}
	
	
	\section{Introduction} \label{sec:intro}
	How people move from one place to another represents a crucial key to depicting human relations and social interactions in complex societies. Data-driven structured studies on human mobility are valuable to identify the factors that drive the movement of individuals and goods, pointing out how they affects economic outcomes in field as labor market dynamics, economic growth, transportation planning and consumer behavior \cite{bettencourt2013origins,bettencourt2021introduction}. 
	Human mobility encompasses a wide range of spatial and temporal scales, from daily commuting within a city to long-term migration across countries and its dynamics  are influenced by factors like economic development, technological advancements, political stability and natural features of the territory. Human mobility and transport dynamics has been studied through different mathematical approaches ranging from individual path-based models to collective population based ones  starting from single individual's behavior up to examining collective mobility trends \cite{barbosa2018human,solmaz2019survey,arcaute2022recent}.  
	The study presented in this work aims to provide a stochastic model for complex network of human mobility based on an origin-destination structure which represents the backbone of  mobility  visitation flows with individuals moving from origin locations to destinations. Then, inferring latent-variables in a specific case study, the graph statistical patterns observed in origin-destination network flow can be recovered to have the same expected properties computed trough the latent-variable estimations. 
	Specifically, the model is grounded in the class of inhomogeneous random graphs \cite{van2016random,bollobas2007phase} also known  as latent-variable networks \cite{kim2018review,rastelli2016properties,hartle2021dynamic,balogh2019generalised,caldarelli2002scale}. Moreover, the mobility network is a directed multigraph with weighted edges, since each travel is labeled with weights (or trip jumps) that represent trip costs such as distance traveled, travel time or emissions per trip. The use of a graph representation of human mobility helps to investigate the dynamic behavior of the system starting from the properties of single constituents, and to estimate how much one part of the network influences another by using graph metrics like degree distribution, centrality, assortativity and transitivity of trip mobility network.  
	%The model has been derived starting from empirical observation of experimental data.
	Following the approach of dynamical formation of latent-variable graphs \cite{vanni2021incremental,boguna2003class,hoppe2013mutual}, the network  has been, then, formalized by using a master equation for the probability density function which encodes  the distribution of the number of visits in destination locations weighted by the costs of trips to get there. In particular, the evolution of visit distribution over time is described trough an integro-differential equation with an explicit asymptotic	solution which highlights the relation of visiting generation patterns to intrinsic and environmental features of locations as specified by latent-variable properties. In stochastic theory, the analytical interpretation achieved is related to the literature of continuous-time Markov processes with the use of a Kolmogorov-Feller partial differential equation  \cite{feller1991introduction,denisov2019exact}. The real-valued solution of the visit ditribution of the trip mobility network has been as compound distribution of traveling packet of normal like conditional probabilities for independent visitation process for each type of destinations. Such visit generation process of the temporal random graph can be is then proofed to be equivalently formalized in terms of a mixture of compound Poisson counting processes commonly used in financial and actuarial science literature \cite{jang2021review,privault2013stochastic}, for example in models of option pricing, risk and insurance analysis, high-frequency trading and market microstructure. 
     In the case of interest, the process is driven by hidden graph structures represented by latent variables associated to each location,  that will be identified with an attractiveness attribute and a productiveness one. The first variable encodes an intrinsic capability of a location to attract visits due to some specific properties of the destination, as extensively discussed in the paper. The second variable encodes, complementary, the ability to produce new travelers, and it depends on some intrinsic feature of the origin area. Moreover the way how each area produces or attracts new trips is achieved according to arrival and departure rates that will be express by some function of the latent variables.  In summary, the mobility graph process is fully characterized by the latent-variable statistical characteristics, and different configurations of the latent-variable patterns  will determine the  mobility network  scaling properties with a particular attention on scale-free behavior of the degree distribution and the spectral trend of graph correlations.    
    In a economic modeling framework, those latent variables allows the characterization of travel demand for different areas in the region, and they can be inferred through the analysis of important demographic, geographic and economic indicators.  As discussed in literature \cite{van2007effects,greenwald2006relationship,lee2015relating},  many different factors are involved for a place to be considered as attractive, such as trip purposes, job or leisure opportunities, infrastructure facilities, geographical characteristics, urban zoning planing. Generally, all those physical and human factors can be captured by land use and travel behavior analysis  which are at the basis of two primary approaches in transportation economics and engineering literature i.e. trip-based models and activity-based models. %Those models are often used  in practice by public planning agencies 
     In such literature, in fact, travel is considered to be a derived demand, that it is generated in response to people satisfying personal needs and desires \cite{geurs2004accessibility,acheampong2015land}. Moreover, the costs and efforts to realize the necessary trips are then influenced and determined by geographical, social and economic factors such as the distance from means of transport, house affordability, economic status, social equality,  the proximity to nature and many others.    
	The research direction is then focused on the use of a real mobility network in order to validate the model presented. In this work, New York metropolitan area has been used as case study by using detailed origin-destination tables where many important information such as the number of people that travel between different locations, along with the information on costs of the visits as well as demographic and economic  properties of the travelers. The analysis of such data  reveals clear scaling relations in commuting patterns of movement flow in agreement with empirical studies and consistent with theoretical predictions as in  \cite{song2010modelling,Fabioetal2013,universalschlapfer2021,nie2021understanding}.
	
	The first contribution of the paper is to provide a theoretical network framework able to describe  human mobility flow based on the existence of latent variables which characterize the travel demand features in urban and transport dynamics. Specifically, analytical, numerical and computational solutions are provided for the strength distribution, assortativity and clustering coefficient. The second contribution consists on the analysis of real origin destination network, so revealing  the existence of scale-free behaviors in the frequency of visits of locations and scaling relations between visits and  trip costs. Consequently, a more pragmatic original contribution of the study is to determine what are the latent variable statistical features able to reproduce the observed patterns of the trip mobility network  in terms of distribution and correlations. 
	%Such empirical scaling evidences will be rephrased in terms of decision making processes of policy makers and market dynamics which have brought the urban areas to be shaped in the way they are observed. 
	The study emphasizes a reciprocal interplay between human mobility and urban structure, whi is an important topic widely discussed in the literature \cite{le2012urban,wang2020review,bassolas2019hierarchical,greenwald2006relationship,lee2015relating,akiba2022correlation}.	Finally, as economical application, I will show the relation between the attractiveness scaling exponents and the value of income elasticity passing through the information about allocation and utilization of land resources for various economic and social activities.  The estimated income elasticity of demand is often used to predict future changes in consumption in response to changes in income. \cite{fouquet2012trends,le2014does,litman2021understanding,ghoddusi2021income}. 

The paper is organized in the following way:  the section \ref{sec_model} is devoted to the definition and mathematical modeling of the latent variable mobility network and, then,  analytical estimates of  scaling laws and graph correlations of visitation patterns. The section \ref{sec_results} will be devoted to network measurements in the case study of New York metropolitan area with the analysis of origin destination tables  obtained from Safegraph dataset \cite{kang2020multiscale}. Visit distribution and network correlations will be measured. In Section \ref{sec_empirical}, the network model is implemented and validated on the basis of latent variable specification. Finally, the latent variable formalism will be applied to an economic problem of the estimation of income elasticity of travel demand.
An appendix section and supplementary material are provided to enhance the discussion on the stochastic intepretation of the visit process, further statistical analysis of the data, and a detailed description of data and interpretation of  the latent variables in travel demand modeling.
% demand I will highlight the empirical features of a real world networks and then such findings will be interpreted in the light of the proposed mathematical framework. In section I will show the consistency of the scale-free visit network model in interpreting real world flow patterns and regularities between the statistical properties of the latent variables and the properties of the number of visits and their costs.  
%Data for human mobility studies are provided at different geographic scales \cite{tan2021mobility}. 
\begin{figure}[!ht]
		\centering
		\begin{subfigure}[l]{0.475\textwidth}
			\centering
			\includegraphics[width=0.75\linewidth]{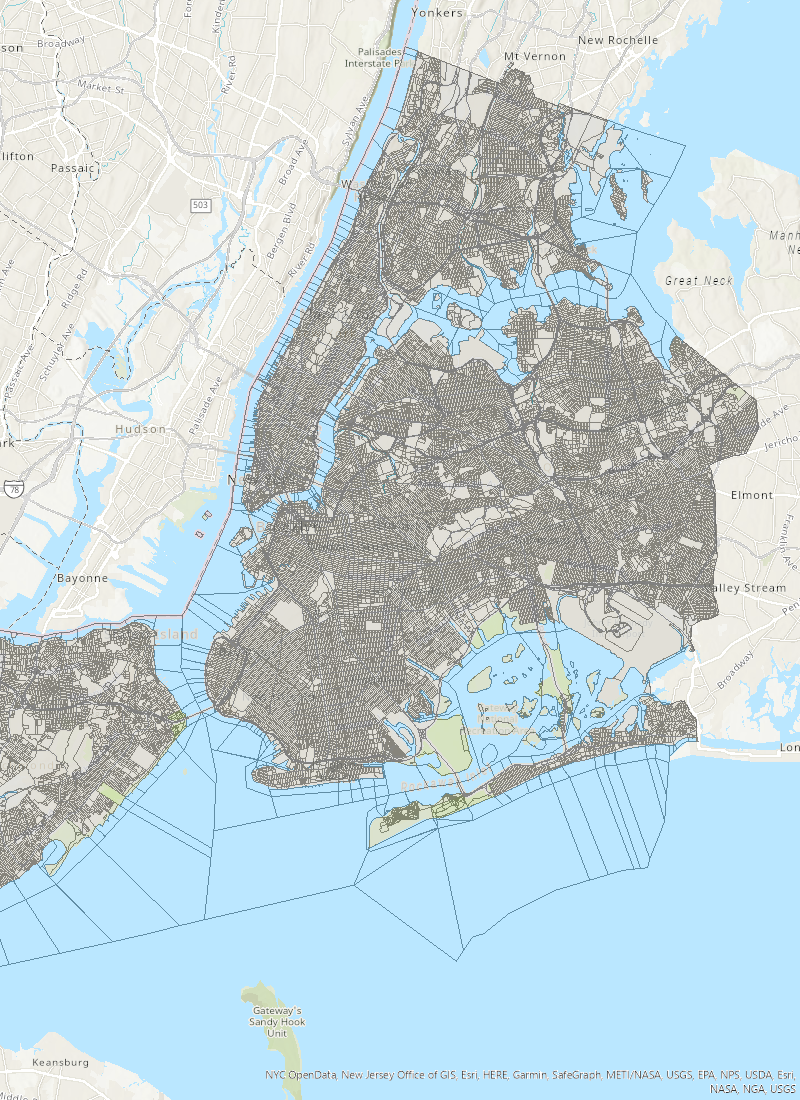}
			\caption{}
		\end{subfigure}
		\centering
		\begin{subfigure}[l]{0.475\textwidth}
			\centering
			\includegraphics[width=0.95\linewidth]{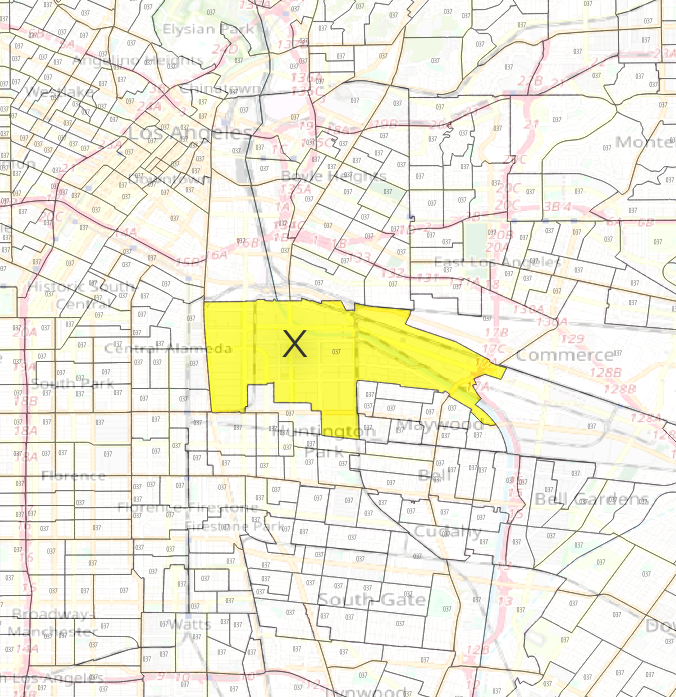}
			\caption{}
		\end{subfigure}
		\caption{Example of a tessellation of a large urban area (i.e. New York City) in  (Census) blocks (a) as in \cite{Tiger}. In particular, a block is characterized by an intrinsic attractiveness $x$ as a latent variable depending on many features of location (b). }\label{fig_tesselNY}
	\end{figure}
%	\begin{figure}[!ht]
%		\centering
%		%		\begin{subfigure}[l]{0.35\textwidth}
%		\centering
%		\includegraphics[width=0.85\linewidth]{img/City_diffusion.pdf}
%		\caption{Graphical representation of number of visits at different periods. Darker areas received more visits than lighter ones. Map has been created by using a map of London built with QGIS and Uber mobile data \cite{ubertaxi} for dynamic insights about traffic and mobility.}
%		\label{fig_citydiffuse}
%		%	\end{subfigure}
%	\end{figure}	
	\section{Model} \label{sec_model}	
A trip mobility network can be build upon an origin-destination rationale and it will be represented as a directed graph where the nodes are administrative units of an urban or regional area (a city, a county, a state etc...) and the directed edges represent number of visits from an origin to a destination area.  
Let us observe, that one latent-variable is the attractiveness of an area and it can be thought as an ancestral property of the block which captures ideally all the variables that can determine it, such as the number residential apartments, job positions, retail stores, geographical features, geopolitical conditions,  social and economical factors, transportation systems and facilities including also the average distance of a block with attractiveness from the rest of the urban area (i.e. the physical average accessibility). 
Essentially those latent-variables, from an economic perspective, can be seen as drivers of travel demand and trip behavior in general.		
		Each node is characterized by an intrinsic attractiveness $x$ representing the travel demand for a destination area. Such attribute is related to various properties of the location such as job opportunities, number of retail stores, geographic features, infrastructure, facilities, school districts etc.
		Similarly, each node is also characterized by another property $y$ that characterizes the location as an origin of trips. That attribute captures the possible users which can depart from the area, so capturing the potential of the are to produce trips, see the Supplementary Materials for a detailed discussion.
		%The decisions which are at the basis of travel demand models are based on factors as the choice and the reason to travel, the destination to travel to, the mode by which to travel and the route on which to travel. 
The model presented describes the occurrence dynamics of trips that occur between origins with a given population $y$  and destination blocks with a given attractiveness label $x$. In this paper, I will, primarily, focus on the case of exchangeable origin-destination blocks, so that the trip generating process can be represented by two independent process: the trip production and the trip attraction process. The first accounts for the number of trips (departures) originating from a block and the latter accounts for the number of trips (visits) ending in each destination block. 
	
First, the visitation model is defined as a time-varying network made up of three ingredients: 1) the location structure of the graph in terms of latent variables, 2) the visit arrival process to destinations, and 3) the effect of trips as sizes assigned to each visit. 
Second, properties of dynamic mobility network will be presented in terms of the its strength degree distribution and degree-correlations, where  the strength of a location is intended as number of visits weighted by their correspondent trip size.	
	\subsection{Definition}
A geographical region of interest, $R$, is specified as the portion of territory for which we are interested in generating the flows. Over the region of interest, a set of geographical tiles called tessellation, $\mathcal{T}$, made up of locations  $l_i$ so that $ {T}= \{ l_i: i = 1 \ldots n \}$  so that the locations are non-overlapping, $l_i \cap l_j = \emptyset, \forall i \neq j$, and  the union of all locations completely covers the region of interest, $\cup_{i=1}^{n}l_i=R$. The tessellation for real geographical regions can be obtained in many ways according to the scope. In the case of interest, location tiles are the census areas defined by national authorities for  administrative and demographic purposes. In particular e census blocks are the smallest geographic unit used by the United States Census Bureau for tabulation of data collected from all houses. An example of tessellation is given in Fig.\ref{fig_tesselNY} for Ney York city, where tiles are the census block division of the urban area.
We denote a graph by $\mathcal{G} = (T, E)$, where $T$ is the set of $n$ locations (nodes) and $E$ is the set of trips (edges). The graph is directed (trip direction) and it allows self-edges and multiple edges (many travelers from and to the same origin and destination).  Furthermore, the network is a labeled graph since the nodes will be specified by intrinsic location attributes.
Finally, the trip mobility network will be defined on top of three fundamental assumptions which will be explained in the next paragraph. The first assumption defines the latent variable backbone of the mobility graph.
	\begin{assumption}[Latent variable graph]\label{a_1}
		Each location $\ell_i$ is labeled with a node type by the latent variable $x_i$, which represents the  attractiveness of destination where a trip ends.  It will be interpreted  as the driver for the travel demand to each destination.
		Let $(x_i)_{i\in[n]}$ be a sequence of latent variables with values in a node-type space $\Omega_x \subseteq \mathbb{R}$ such that the empirical distribution of $(x_i)_{i\in[n]}$ approximates a probability measure $\mu_x$ as $n \to \infty$ \cite{bollobas2007phase,van2016random}.
		%That is, we assume that, for each μ-continuous Borel set A ⊆ S , as n → ∞,
		Consequently, the set of attractiveness variables $\{x_1, \ldots, x_n \}$ assocated to a location are considered realizations of independent and identically distributed latent random variables with an empirical distribution that converges almost surely to the cumulative distribution function $F(x)$.
	\end{assumption}
Let us observe that, in the paper, the latent-variable probability measure is defined on $(\mathbb{R},\mathcal{B}(\mathbb{R}))$ {where $\mathcal{B}(\mathbb{R})$ is the Borel $\sigma$-algebra generated from the real line} and the probability  will be assumed absolutely-continuous respect to the Lebesgue measure, so that the probability density function can be defined as $\rho(x)=F'(x)$. Such condition will allow for various mathematical and statistical manipulations in the paper for practical applications.
Similarly, the locations, other than being destinations, are, at the same time, labeled as origins. The correspondent  intrinsic feature variable will represent the resident active population which stays in the area from which a trip originates\footnote{Despite of not being a proper hidden variable, the people which can actively considered travelers undergoes to many factors (such as residence, age, employment status, and many others).}. The latent variable in this case is named as productiveness of the location. Consequently, the set of productiveness variables $\{y_1, \ldots, y_n \}$   in node-type space $\Omega_y\subseteq \mathbb{R}$ consists of realizations of independent and identically distributed latent random variables with probabilty measure $\mu_y$ where a probability density function is defined $\phi(y)$, which can be different from the attractiveness distribution. %Only in the case where all the social and economical factors can be considered to be known and homogeneous among all the locations, then the population $y$ can be directly measured from data so that the observed population labels $\{y_1,\ldots,y_N\}$ undergoing to an empirical distribution with probability density as $\phi( y )$.  
From a dynamical perspective, the model is presented as random graph process that is a stochastic process that describes a random graph evolving in time \cite{janson2011random,bollobas1998random}.
  The processes studied here will have a fixed vertex set, and they will 	start without any edges and grow by adding edges according 	to linking rule, without deleting any, since the total number of visits up to a certain time is 
 studied. Once a regional tessellation has been defined in terms of a graph as in A.\ref{a_1}, the graph processes for directed multiedges can be defined \cite{ben2020semi} as random graph $\mathcal{G}_t$ evolving so that $\forall t$ a new edge is added  $\mathcal{G}_{t+1}=\mathcal{G}_t\cup \ell $ where $\ell \in 2\binom{V}{2}\setminus E{(\mathcal{G}_t)}$ is chosen randomly with replacement with probability proportional to a kernel function $\mathcal{K}(x,y)$ defined as $\Omega_x\times\Omega_y\to [0,\infty)$. Such kernel is rules out the chances that a location of attractiveness $x$ will be a destination of a visit whose trip originated in a location of feature $y$.  Knowing that  $\mu_y(dy)$ represents the measure assigned to the infinitesimal interval of productiveness for origin locations. 
	\begin{assumption}[Trip generation]\label{a_2} 
	 The trip arrival process is defined through the infinitesimal arrival intensity  $d\mathcal{V}_x=\mathcal{K}(x,y)\mu_y(dy)$ which drives the dynamics of new travels landing in destinations of attractiveness $x$ conditional to trips that departed from locations of infinitesimal productiveness $y$. 
	\end{assumption}
Such rate can be interpreted as the propensity (or intensity) for a traveler to move towards a destination of a given attractiveness $x$.
So, the graph evolves  when a new visit from an origin to a destination is completed (a new link in the network) according to the attraction and the production rates in the degree-space of the mobility graph. Similarly, as dual problem, a trip departure process can be defined by the infinitesimal departure intensity $d\mathcal{V}_y=\mathcal{K}(x,y)\mu_y(dy)$.
%The overall trip generation process can be described by both using the arrival and departure intensity.	

		%The random graph process is a family of $\{\mathcal{G}^{(t)}\}$ of random graphs where the parameter $t$ is interpreted as time .  
		The evolving mobility graph $\{\mathcal{G}^{(t)}\}$ is fully described by means of a time-varying adjacency matrix $A^{(t)}$ which represents the origin-destination table of the mobility problem at time $t$ and the sum along the columns represent the in-degree of the nodes or equivalently the number $k$ of visits  received up to time $t$. %  and $A_{i,j}^{(t)}$ are the entries of the adjacency matrix of graph at sampling time $t$.  
Finally, an extension of the model is considered where typical lengths of single trips is associated to the corresponded visits. In the case of weighted network where each links is associated to a weight  as a random variable independent from visitation process and the graph is expressed in terms of a weighted directed multigraph adjacency matrix $\tilde{A}^{(t)}=C^{(t)}\circ A^{(t)} $, which indicates a element-wise matrix multiplication and $C$ is considered either a coefficient matrix or a random one according to the real-world phenomena under study alongside the data that can be used.
%Such adjacency matrix represents the trip distribution matrix from origin to destination areas. 
Specifically:
	\begin{assumption}[Trip weights]\label{a_3}
 All the arrivals of new trips are weighted in terms of some trip features. The weights $\mathcal{r}$  are independent and identically distributed  random variables, with a specified probability density function $\varrho_x(\mathcal{r})$. %Those weights are considered independent of the counting process of visits $\{\mathcal{C}_x(t)\}_{t\in \mathbb{R}^+}$. ]
\end{assumption}
 Those weights have the meaning of trip 'size' as for example the distance traveled from an origin to the selected destination, or the emission impact of each trip, or any type of travel cost or visit benefit. The weight probability $\varrho_x(\mathcal{r})$ may depend or may not on the attractiveness of the destination node  according to the particular type of weight considered in the specific situation under investigation.
 Let us notice that  the in-strength (i.e. weighted in-degree) of a destination node $\kappa$ is defined as the sum of  weights of all the visits received. In the case that weights are all equal to $1$ the strength $\kappa$ is equivalent to the degree $k$.
At this point, finally, the model can be defined:	
    \begin{definition}\label{definition_net}
    	    	The trip mobility network  is a temporal inhomogeneous random graph model that describes a visit generation process satisfying the assumptions A.\ref{a_1}, A.\ref{a_2}, A.\ref{a_3}.
    \end{definition}

%	The expected value of a given 	quantity $O$ is obtained by averaging over all possible configurations as:
%	\begin{equation*}
%\langle O(\mathcal{A})\rangle=	\sum_{\{\mathcal{A}\}}\mathsf{W}(\mathcal{A})O(\mathcal{A})
%	\end{equation*}
%	
The weighted mobility multigraph process $\mathcal{G}_t$  can be studied in terms of the adjacency matrix $\tilde{A}^{(t)}$ representation of the graph.
A general framework which  describes ensembles of dynamic networks can be assessed by making a Markov assumption on the evolution of the network by studying the probability of realization of a member of configuration ensemble  graph over time \cite{hanneke2010discrete,zhang2017random}. Specifically, the temporal evolution of finding the trip mobility network in the configuration $\mathcal{A}$ at time $t$ after $L$ steps (trips)  can be written as:
%\begin{align}
$$\mathbb{P}(\mathcal{A},t) =\prod_{l=0}^{L}\mathbb{P}(\tilde{A}^{(t_l)} | \tilde{A}^{(t_{l-1})},\boldsymbol{{\varLambda}})  $$
%\end{align}
where $\boldsymbol{\varLambda} $ is any transition rule built on latent variables, which describes the creation, at each time $t_l$, of a new trip from an origin towards a destination.	
From a modeling perspective,  the complexity of knowing the configurational distribution of the origin-destination mobility network can be reduced trough the study of the degree (or strength) distribution and its higher order statistics.	We actually model the degree distribution, that in the in-degree case results to be the visiting distribution defined as:
%\begin{equation}\label{eq_fulldegree}
$$P(\kappa,t)=\sum_{\{\mathcal{A}\}}\sum_{i}\delta(\kappa-\kappa_i)	\mathbb{P}(\mathcal{A},t)$$
%\end{equation}
where $\delta$ is the delta function, and $\mathbb{P}(\mathcal{A},t)$ is the probability to find our trip mobility network in the configuration $\mathcal{A}$ at time $t$, where each node has a strength degree $\kappa_i$. If the in-degree of a node is simply the number of arrivals of new travelers, the in-strength of a location node is defined as the number of arrivals of  new travelers who have faced a cost. Such weighted version of arrivals is named visits where each trip has an intensity.

	\subsection{Visit distribution}
In the case of this work, the derivation of the  visit distribution is assessed on the basis of the latent variable framework.
% starting from the dynamics of the same distribution for destinations of a given attractiveness $x$ according with the assumptions and the definition in the hidden variable framework. 
	The model is developed upon the asymptotic regime assumption of an infinite network where the number of locations in a tessellation is extremely large, where each location can generate and attract an unlimited number of edges (trips). Consequently, the continuous mean-field approach is used so that single locations can be studied as uncorrelated nodes in the same class of locations with the same attractiveness  \cite{boguna2003class,hoppe2013mutual,barrat2008dynamical}. So rather than studying the single location $l_i$, one analyzes those locations which belong to the same class of attractiveness level, i.e.  $l_x$ seen as a continuous random variable with probability density $\rho(x)$.
 Let us call conditional  visit distribution  $p(\kappa,t|x)$ the  strength distribution conditional to destinations of the attractiveness type $x$, and  each class  evolves independently one from any another, so that a superposition of co-evolving conditional degree distributions is possible. 
%The master equation above is valid in the period observation (i.e. a day) and its solution  represents the traveling attraction-packet for blocks with attractiveness $x$.	
The overall visit distribution can be derived according to the following proposition:
	\begin{proposition}[]\label{prop1}
		According the visitation model's assumptions as in the Definition \ref{definition_net}, the visit distribution of the mobility network is fully characterized by  the attraction rate  
		$\nu_x$ that defines the transition probability, per unit of time, that a destination of attractiveness $x$ increases its number of visits by one from any destination. The attraction rate is defined as the mean intensity of the trip arrival process $\nu_x=\int_{\Omega_y}d\mathcal{V}_x$.
\begin{itemize}
	\item The evolution of the conditional visit distribution can be described by a master equation for destinations with attractiveness $x$ as an integro-differential Kolmogorov-Feller equation:
	\begin{align}\label{eq_mastereq}
	% p_x (k,n+1) =&\; p_{x} (k,n) +\nu_x\, p_{x} (k-1,n) - \nu_x\,p_{x} (k,n) ,\\
		\frac{\partial }{\partial t} p_x (\kappa,t) =&\;\int_{0}^{\kappa} \nu_x \varrho_x(\mathcal{r})\, [ p_x(\kappa - \mathcal{r},t)-p_x(\kappa,t)]d\mathcal{r} 
	\end{align}
	with the initial condition  $p_x(\kappa, 0) = \delta(\kappa)$.
	  In the asymptotic regime,  the conditional probability $p(\kappa,t|x)$  can be written as:
\begin{align}\label{eq_truncatednormal}
p(\kappa,t|x)\sim&\; \epsilon_t\tfrac{1}{\sqrt{2\pi \nu_x \langle \mathcal{r}^2\rangle_x t}}e^{-\frac{(\kappa-\nu_x  \langle \mathcal{r} \rangle_x t)^2}{2\nu_x\langle \mathcal{r}^2\rangle_x t}} \qquad \text{, where } \; \epsilon_t=\tfrac{2}{1+\text{erf}\left[\frac{\nu_x  \langle \mathcal{r} \rangle_x t }{\sqrt{2\nu_x  \langle \mathcal{r}^2 \rangle _x t}}\right]} 
\end{align}
where the correction factor $\epsilon_x \to 1$ in the asymptotic limit $t\to \infty$, and 	$\langle \mathcal{r} \rangle _x$ and $\langle \mathcal{r}^2 \rangle _x$ are the first and the second moment of the trip-weight distribution $\varrho_x(\mathcal{r})$.
		%Then the following  condition is always verified :$$\int_{0}^{\infty}\nu_x\rho(x)dx=\tfrac{1}{B} $$ where $B$ is the total number of blocks i.e. the dimension of the origin-destination matrix.
		% \begin{equation}\label{eq:propaga_discr}
		%p_x(k,n)={{n}\choose{k}} \nu_x^k\,(1-\nu_x)^{n-k} \simeq \frac{(n\,  \nu_x)^k}{k!}e^{-n\, \nu_x} 
		%\end{equation}
		
			\item Finally, the temporal asymptotic expression of the visit distribution is a mixture distribution as: 
			\begin{equation}\label{eq_visitingdistr}
			% P_{\dashv}(k,n)  =& \int p_x(k,n)\rho(x) \,dx \nonumber \\
			%= & \; \frac{1}{J(k,n)}\; \int \delta \Big(x- x_0(k,n) \Big) \rho (x) dx\\
			% = & \; \frac{1}{J(k,n)} {\rho\big(x_0(k,n)\big)}\\
			P(\kappa,t)=\int_{\Omega_x} p(\kappa,t|x)\rho(x)dx \sim \; \sum_{i=1}^{m} \left | \tfrac{\partial z(x)}{\partial x} \right |_{x_{0,i}}^{-1}  {\rho\big(x_{0,i}\big)} 
			\end{equation}
			where $x_{0,i}=x_{0,i}(\kappa,t)$ are the zeros of the expression $z(x)=\kappa-\nu_x \langle \mathcal{r} \rangle_x\,t$, and $P(\kappa,t)$ the in-strength distribution of the overall trip mobility network. %once the explicit function of the attraction rate has been specified, %, and the factor $J(k,n)= \frac{\partial h(x,n)}{\partial x}\big|_{x=x_0(k,n)}$.  
			%\begin{align}\label{eq:master_eq_p}
			%p_x(k,n+1) - p_x(k,n) =&\; \nu_x\, p_x(k-1,n) - \nu_x\,p_x(k,n),
			%\end{align} 
		\end{itemize}
		
		%\begin{align*}
		%\Bar{k}(n) = \;n \int \nu_x\rho(x)dx ,
		%\end{align*}
	\end{proposition}
	\begin{proof}
		Let us notice that the attraction rate $\nu_x $ is 	transition rate of new visits per unit of time interval, as the chance of having a new arrival in a destination of type $x$ originated from any location, so that the attraction rate is the mean intensity  as in \cite{van2016random,bollobas2007phase,hoppe2013mutual,vanni2021incremental}:
		\begin{align}
		\nu_x=\int_{\Omega_y}d\mathcal{V}_x={\int_{\Omega_y} \mathcal{K}(x,y)d\mu_y(y)}=\int_{\Omega_y}  \mathcal{K}(x,y) \phi(y)dy
		\end{align}
	The master equation for the evolution of the conditional probabilities $p_x(\kappa,n)$ for locations with attractiveness $x$, where the step size is the correspondent  $\Delta \kappa= \mathcal{r}$ which represents the weight of each link i.e. the decision heuristic variable to move in the selected destination. It can be written as:
	\begin{align*}
		p (\kappa,t+\tau|x)=&\; (1 -\nu_x)p(\kappa,t|x)+   \int_{0}^{\kappa} \, \nu_x  p(\kappa-\mathcal{r} ,t|x) \varrho_x(\mathcal{r} ) d\mathcal{r} 
	\end{align*}
		%Let us start from the discrete master equation for the degree distribution written as:
		%\begin{align*}
		%p_x (k,n+1) -  p_{x} (k,n) =&\; \nu_x\, p_{x} (k-1,n) - \nu_x\,p_{x} (k,n) 
		%\end{align*}
		%with the initial condition $p_{x}(k,n=0)=\delta_{0,k}$,  as discussed in \cite{vanni2021incremental} where the solution has been provided by using the generating function approach and then various approximation methods have been discussed. 
		If we also assume that $p_x$ is slow, so that it changes only slightly during this time step $\tau=\Delta t$ and redefine $\nu_x:=\nu_x/\tau$, then we can write the continuum master equation like:
		\begin{align}
		\label{eq_masteeqrappendixy}
		\frac{\partial }{\partial t} p_x (\kappa,t) =\dot{p}_x (\kappa,t)=&\;\int_{0}^{\kappa} \nu_x \varrho_x(\mathcal{r})\, [ p_x(\kappa - \mathcal{r},t)-p_x(\kappa,t)]d\mathcal{r} %\\
		%	\frac{\partial }{\partial t} P (\kappa,t) =&\;\int_{\Omega_x}\int_{0}^{\kappa} \nu_x \varrho_x(\mathcal{r})\, [ p_x(\kappa - \mathcal{r},t)-p_x(\kappa,t)] d\mathcal{r}\, \rho(x) dx
		\end{align}
		and $ \varrho_x(\mathcal{r})$ is the distribution of the trip distances covered to reach destination blocks of attractiveness $x$.	
		At this point let us apply the Laplace transformation in the variable $\kappa$ so $\mathcal{L}\{p(k,t)\}\equiv\hat{p}(s,t)$ and the master equation transforms as:
		\begin{align*}
		\dot{\hat{p}}(s,t) = &\nu_x\hat{\varrho}_x(s)\hat{p}_x(s,t)-\nu_{x}\hat{p}_x(s,t) = \nu_{x}\left(1-\hat{\varrho}_x(s)\right)\hat{p}_x(s,t)
		\end{align*}
		where the convolution product has been used. The solution can be written as:	
		\begin{align}\label{eq_fullsolutionlaplace}
		\hat{p}(s,t)=&\,ce^{-\nu_x(1-\hat{\varrho}_x(s))t}
		\end{align}		
		with the initial condition $\hat{p}(s,0)=\mathcal{L}\{\delta(k)\}=e^{0}=1$ so that $c=1$.  One can notice that the characteristic function is equivalent to the laplace transform as expressed before. 
		It is possible to write the Laplace transform of the jump distribution $\varrho_x(\mathcal{r})$ in terms of its moments as:
		\begin{align}
		\hat{\varrho}_x(s)= & \sum_{n=0}^{\infty} (-1)^n\frac{s^n}{n!}\mathbb{E}(\kappa^n)
		\end{align}	
		If one assumes that $\varrho(\mathcal{r})$ is a peaked distribution, following the central limit theorem rationale, one can assume it is described by the first two (finite) moments so that the solution in Eq.\eqref{eq_fullsolutionlaplace} can be approximately: 
		\begin{align}
		\hat{p}(s,t)\approx &  e^{-\nu_x t\langle \mathcal{r}\rangle_x s + \frac{1}{2}\nu_x t\langle \mathcal{r}^2\rangle_x s^2}
		\end{align}			
		its inverse Laplace transform  can be considered \cite{berberan2007computation}  the case asymptotic case of $s\ll \frac{\langle \mathcal{r} \rangle_x }{\langle \mathcal{r}^2 \rangle_x}$ of the truncated normal distribution as stated in the eq.\eqref{eq_truncatednormal}
%		\begin{align}
%		p(\kappa,t|x)=&\; \epsilon_t\tfrac{1}{\sqrt{2\pi \nu_x \langle \mathcal{r}^2\rangle_x t}}e^{-\frac{(\kappa-\nu_x  %\langle \mathcal{r} \rangle_x t)^2}{2\nu_x\langle \mathcal{r}^2\rangle_x t}} \qquad \text{, where } \; %\epsilon_t=\tfrac{2}{1+\text{erf}\left[\frac{\nu_x t \langle \mathcal{r} \rangle_x }{\sqrt{2\nu_x t \langle \mathcal{r}^2 %\rangle_x}}\right]} 
%		\end{align}
%		where the correction factor $\epsilon_x \to 1$ in the asymptotic limit $t\to \infty$.
		
		For the second point, the visit (in-strength) distribution  of the mobiliy network is expressed as a  compound probability distribution that results from assuming that a random variable $\kappa$ is distributed according to some parametrized distribution  with the latent parameter $x$ distributed according to some attractiveness distribution \cite{everitt2013finite,sundt2009recursions}. So,  the (unconditional) visiting in-strength distribution results from marginalizing the conditional distribution $p(k,t|x)$ of the non-negative real-valued random variable $x$. So, the probability density function of the visiting distribution is given by the following the mixture density:	
		\begin{align}\label{eq_inmixingfull}
		P(\kappa,t)=&\mathbb{E}\left[ p(\kappa,t|x)\right]=\int_{\Omega_x} p(\kappa,t|x)d\mu_x(x)=\int_{\Omega_x} p(\kappa,t|x) \rho(x)dx
		\end{align}
		Here, $p(\kappa|x,t)$ is the distribution of $\kappa$ when we know $x$ at time $t$, in which the relation between $\kappa$ and $x$ can be seen as deterministic, i.e. $\kappa = F(x,t)=\mathbb{E}_t[\kappa|x]=\nu_{x}\langle  \mathcal{r} \rangle_x t$, defining the distribution by its expected degree value through the moment-generating function from the Laplace transform above. So, $\kappa$ can only be single value, whose distribution is represented by dirac-delta functions $\delta(\kappa-F(x,t))$. 
		%In this case $F(x,t)$ can be taken as the mean of the $p_x$ which is 
		%$$F(x,t)=\nu_{x}\langle  \mathcal{r} \rangle t \,\left(1+\sqrt{\tfrac{2}{\pi \nu_{x} \frac{\langle  \mathcal{r} \rangle}{\langle  \mathcal{r}^2 \rangle} t}} \right).$$ 
		Consequently, the empirical visiting density probability function can be written as:
		\begin{align}\label{eq_inmixingtail}
		P(\kappa,t) \sim &\int_{\Omega_x} \delta (\kappa-F(x,t))\rho(x)dx = \int_{\Omega_x} \delta (\kappa-{\nu_x}\langle \mathcal{r} \rangle_x t)\rho(x)dx 
		\end{align} 
		which consists in approximating the visiting probability density by means of a Dirac mixture \cite{tulsyan2016particle,handschin1969monte}, where $\rho(x)$ is the attractiveness probability density. Such procedure is equivalent to a change of variable respect to the deterministic one-to-one function in  the static model as in \cite{caldarelli2002scale}. At this point we use the property:
		$
		\delta(z(x))= \sum_{i=1}^{m}\tfrac{\delta(x-x_0^{(i)})}{|\partial z(x)/\partial x|}
		$
		where $x^{(i)}_0$ are the m-roots of $z(x)=0$ where in the transport model $z(x)=\kappa-\nu_x \langle \mathcal{r} \rangle_x t $ where $z$ is a continuously differentiable function with $z'$ nowhere zero. So:
		\begin{align}
		P(\kappa,t) \sim & \; \sum_{i=1}^{m} \left | \tfrac{\partial z(x)}{\partial x} \right |_{x_0^{(i)}(\kappa)}^{-1}\; \int \delta \Big(x- x_0^{(i)}(\kappa) \Big) \rho (x) dx \nonumber \sim  \; \sum_{i=1}^{m} \left | \tfrac{\partial z(x)}{\partial x} \right |_{x_0^{(i)}(\kappa)}^{-1}  {\rho\big(x_0^{(i)}(\kappa)\big)}
		\end{align} 
		which represent a general formula for the tail behavior of the degree distribution for a mobility network  with a generic attraction rate $\nu_x$. 
		
	\end{proof}
Let us observe that if the trip weight distribution $\varrho_x(\mathcal{r})$ is a dirac delta, then the strength distribution is equivalent to the degree distribution. Another particular case is when the trip weight distribution is identical over the attractiveness variable so that $\varrho_x(\mathcal{r})=\varrho(\mathcal{r})$ then the strength is proportional to the degree $\kappa=\langle \mathcal{r} \rangle k$. In the next remark, the same results can be expressed in terms of stochastic process terminology rather than in terms of an integral-differential  master equation. 
Let us observe that the type of processes described in Proposition \ref{prop1} can be reinterpreted in terms of a mixture of compound Poisson processes as described in the Appendix which represents to be very popular in financial mathematics to model stock prices, insurance claims, and other financial phenomena  \cite{privault2022introduction,tankov2003financial}. The visitation model of the mobility network can also be interpreted in combinatorial terms by using urn processes for solving balls in bins problems and finding the occupancy distributions as sketched in SM4. Such approach is used in world trade literature  as in \cite{armenter2014balls,riccaboni2013global,casiraghi2021configuration}. Despite different approaches, the one introduced in the paper has the advantage to provided direct, though asymptotic, solutions to the scaling relations in the mobility networks.

As a particular choice of the occupation probability $p$, for every couple of origin-destination locations, trips can be realized according to a  kernel function $\mathcal{K} (x , y )$ \cite{caldarelli2002scale,chung2002connected,chung2002average,servedio2004vertex}.
	As a crucial case in the context of scale-free networks, one can consider the attraction rate to be proportional to some power of the destinations' attractiveness:
	\begin{lemma}\label{coro1}
	Le us assume that the attraction rate is of the form $\nu_x= \nu_0 x^{\alpha_0}$, with $\alpha_0 >0$, and homogeneous trip-weight distribution $\varrho_x(\cdot)=\varrho(\cdot)$, the asymptotic trip-visit distribution can be written as:
		\begin{align}
		%P_{\dashv}(k,n)&=\frac{ 1}{\alpha_0\nu_0^{\frac{1}{\alpha_0}}}\, \frac{1}{n^{\frac{1}{\alpha_0}} \, k^{\frac{\alpha_0-1}{\alpha_0}}}\; \rho \left( \frac{k^{\frac{1}{\alpha_0}}}{ \nu_0^{\frac{1}{\alpha_0}} n^{\frac{1}{\alpha_0}} }\right) \\
		%	  P_{\dashv}(k,n)&\sim \frac{\rho (x_0)}{\alpha_0\nu_0nx_0^{\alpha_0-1}}
		P(\kappa,t)&\sim  t^{-\frac{1}{\alpha_0}}  \,\kappa^{\frac{1}{\alpha_0}-1} \, \rho (x_0)
		\end{align}
		where $x_0=x_0(\kappa,t)=\left(\frac{\kappa}{\nu_0t}\right)^{\frac{1}{\alpha_0}}$, and where $\rho$ is the attractiveness probability density function. For $\alpha_0=0$ the Erdos–Renyi random graph is recovered.
		
		In the particular case that the attractiveness distribution is $\rho(x)\sim \rho_0 x^{-\eta}$, the visiting in-degree distribution has the following asymptotic tail distribution  
		\begin{equation}\label{eq_visitingscalefree}
		P(\kappa,t) \sim t^{\frac{\eta-1}{\alpha_0}} \kappa^{-(1+\frac{\eta-1}{\alpha_0})}
		\end{equation}
		which shows the typical scale-free structure of an inverse power law distribution for the visiting degree of the mobility network.
	\end{lemma}

The analytical results in the previous remark is  confirmed by numerical integration of the compound distribution eq.\eqref{eq_visitingdistr} by using the  truncated normal conditional probability eq.\eqref{eq_truncatednormal}. Moreover, a  graph process is performed via monte carlo (MC) simulation of the network where occupation probability is expressed via a separable linking function $ \mathcal{K}(x,y)=g(x)h(y)$ for the evolution of the adjacency matrix. Such kernel gives arise to an attraction rate as $\nu_x= \nu_0 g(x)$ , where $\nu_0$ is a normalization constant as shown in details in the supplementary materials. So by choosing $g(x)=x^{\alpha_0}$  we are in the case as specified in  Remark \ref{coro1}. At this point, it is possible to compare the three approaches and confirm the consistency of results obtained. In Fig.\ref{fig_3frames}   it can be observed how the three approaches provide the same vist distribution for a particular choice of the parameters as discussed in the caption. 
	As expressed in different research works \cite{balogh2019generalised,masuda2004analysis,fujihara2010universal,di2022score,servedio2004vertex}, different combinations of the attraction rate and attractiveness distribution can generate the same trip-visit distribution. In particular a scale-free behavior can be recovered trough   exponential drivers where the attraction rate is in the  form $\nu_x= \nu_0 e^{\gamma x}$ and the attractiveness distribution is $\rho(x)\sim \rho_0 e^{-\lambda x}$. Also in this case the the asymptotic trip-visit distribution is scale free, in particular, via the transformation variable $x_0=x_0(\kappa,t)=\frac{1}{\gamma}\log \frac{\kappa}{\nu_0 t}$, the distribution becomes $P(\kappa,t) \sim t^{\frac{\lambda}{\gamma}} \kappa^{-(1+\frac{\lambda}{\gamma})}$. 
%\begin{figure}[!ht]
%		\centering
%		%		\begin{subfigure}[l]{0.35\textwidth}
%		\centering
%		\includegraphics[width=0.75\linewidth]{img/Attraction_analytical_distribution.pdf}
%		\caption{Visit distribution computed through three different approaches for trips with constant weights in the case where attractiveness distribution is $\rho(x)\sim\rho_0x^{-2}$  and kernel function $ \mathcal{K}(x,y)\propto x^{1.5}h(y)$ so that  $\nu_x=\nu_0 x^{1.5}$. The probability density $P(\kappa)$ has been estimate with three different approaches: (1) it is evaluated through the numerical integration of the compound  probability as in eq.\eqref{eq_visitingdistr}. (2) It is evaluated through montecarlo (MC) graph simulation of sequential adjacency matrices with $N=300$ locations and with a simulation time of $t=10^5$ time steps. The two approaches provide the same scale-free behavior of the visit distribution as $P(\kappa,t) \sim t^{\frac{2}{3}} \kappa^{-\frac{5}{3}}$ as expected by the analytical asymptotic estimate eq.\eqref{eq_visitingscalefree}. }
%		\label{fig_3frames}
%		%	\end{subfigure}
%\end{figure}
	\begin{figure}[!ht]
	\centering
\begin{subfigure}[c]{0.65\textwidth}
	\includegraphics[width=0.95\linewidth]{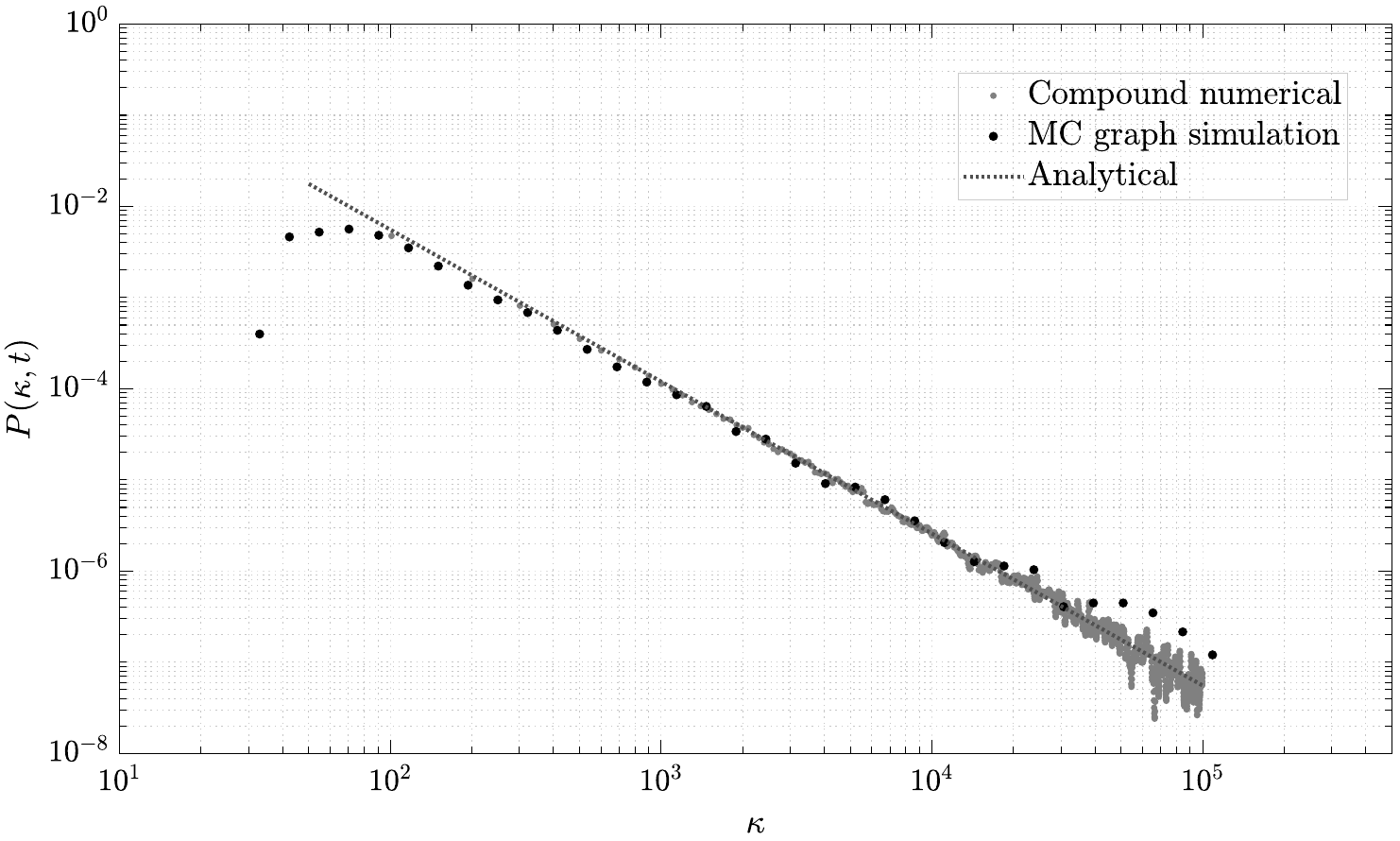}
\caption{ }
\label{fig_3appro}
\end{subfigure}
	%		\begin{subfigure}[l]{0.35\textwidth}
	\centering
\begin{subfigure}[c]{0.3\textwidth}
	\includegraphics[width=0.95\linewidth]{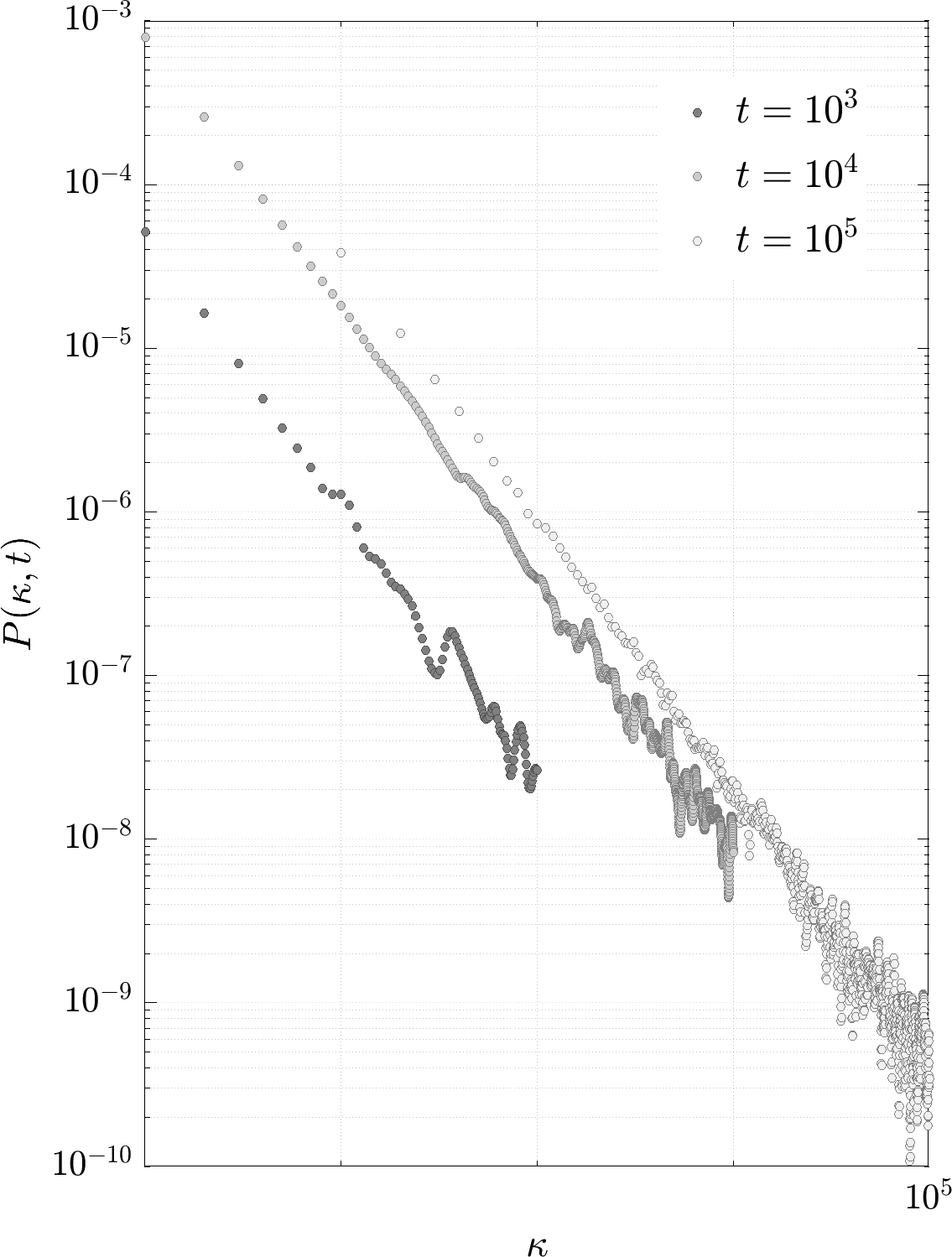}
	\caption{}
	\label{fig_citydiffuse2}
\end{subfigure}
\caption{Visit distribution computed  in the case of constant trip weights,  and the  attractiveness distribution is $\rho(x)\sim\rho_0x^{-2}$  and kernel function $ \mathcal{K}(x,y)\propto x^{1.5}h(y)$ so that  $\nu_x=\nu_0 x^{1.5}$. In (a), the probability density $P(\kappa)$ has been estimate with three different approaches: (1) it is evaluated through the numerical integration of the compound  probability as in eq.\eqref{eq_visitingdistr}. (2) It is evaluated through montecarlo (MC) graph simulation of sequential adjacency matrices with $N=300$ locations and with a simulation time of $t=10^5$ time steps. The three approaches provide the same scale-free behavior of the visit distribution as $P(\kappa,t) \sim t^{\frac{2}{3}} \kappa^{-\frac{5}{3}}$ as expected by the analytical asymptotic estimate eq.\eqref{eq_visitingscalefree}. In (b) the compound distribution approach is calculated at three different time snapshots.   }\label{fig_3frames}
\end{figure}	
%	\begin{lemma}\label{coro2}
%	Le us assume that the attraction rate is in the exponential form $\nu_x= \nu_0 e^{\gamma x}$, the asymptotic trip-visit distribution can be written as:
%	\begin{align}
%	%P_{\dashv}(k,n)&=\frac{ 1}{\alpha_0\nu_0^{\frac{1}{\alpha_0}}}\, \frac{1}{n^{\frac{1}{\alpha_0}} \, k^{\frac{\alpha_0-1}{\alpha_0}}}\; \rho \left( \frac{k^{\frac{1}{\alpha_0}}}{ \nu_0^{\frac{1}{\alpha_0}} n^{\frac{1}{\alpha_0}} }\right) \\
%	%	  P_{\dashv}(k,n)&\sim \frac{\rho (x_0)}{\alpha_0\nu_0nx_0^{\alpha_0-1}}
%	P(\kappa,t)&\sim \tfrac{1}{\alpha_0 \nu_0 }\kappa^{-1} \, \rho (x_0)
%	\end{align}
%	where $x_0=x_0(\kappa,t)=\frac{1}{\gamma}\log \frac{\kappa}{\nu_0 t}$, and where $\rho$ is the attractiveness probability density function. For $\gamma=0$ the Erdos–Renyi random graph is recovered.
%	
%	In the particular case an exponential attractiveness distribution $\rho(x)\sim \rho_0 e^{-\lambda x}$, the trip-visit in-strength distribution has the following asymptotic distribution: 
%	\begin{equation}\label{eq_visitingscalefreexp}
%	P(\kappa,t) \sim t^{\frac{\lambda}{\gamma}} \kappa^{-(1+\frac{\lambda}{\gamma})}
%	\end{equation}
%	which shows the typical scale-free structure of an inverse power law distribution for the visiting degree of the mobility network.
%\end{lemma}	
In particular, there is a model selection issue, since different choices of attractiveness features can generate the same effect on the strength-distribution, one could clarify the ambiguity  investigating higher order characterization of the degree distribution of the mobility network.

	\subsection{Visit correlations}
	A more detailed characterization concerns the exploration of the connectivity correlations in the origin-destination correspondences of trip mobility network.  Higher order statistics of a network in the degree space, can be obtained trough by the conditional probability  $P(\kappa^{(1)}, \kappa^{(2)}, \ldots, \kappa^{(k)}|\kappa',t)$ that a node with strength $\kappa'$ connects to nodes with strength $\kappa^{(1)}, k^{(2)},\ldots, \kappa^{(k)}$ at time $t$. 
	The simplest of these  degree correlations is the two-point correlation being described by the conditional probability $P(\kappa|\kappa',t)$ as  the probability that a trip departing from an origin location of out-strength (departure) $\kappa'$ reaches a destination node of in-strength (visit) $\kappa$. 
	The correlations between degrees of the nearest-neighbouring vertices  are described by the probability distribution:
	\begin{equation*}
	P(k,k',t)=\sum_{\{\mathcal{A}\}}\sum_{ij}\delta{(k-k_i)}	\mathbb{P}(\mathcal{A},t) \delta{(k_j-k')}
	\end{equation*}
	However, the empirical evaluation of such conditional probability in real networks is  cumbersome,  so the weighted degree-degree correlations are commonly accounted by average-nearest-neighbor's strength function $k_{nn}(\kappa,t)$ which makes use of a smoothed conditional probability \cite{latora2017complex} often used as a measure of degree homophily of the nodes. In the latent variable framework, as shown in Fig.\ref{fig_dual},  the conditional assortativity  $k_{nn} (x)$ measures how much a location with attractiveness $x$ tend to be a destination of an origin location of population $y$ as defined in \cite{boguna2003class,serrano2007correlations,murase2019sampling}. 
	In a similar way, the three-point correlations can be studied in terms of the clustering coefficient spectrum $c(\kappa,t)$ which indicates the probability that two neighbors of strength-$\kappa$ node are neighbors themselves. In the case of weighted and directed networks there many different ways to define the cluster coefficient \cite{fardet2021weighted,fagiolo2007clustering}. At the latent variable level, the conditional clustering coefficient  of a destination with attractiveness $x$ can be interpreted as the probability that two randomly chosen locations with  trips towards a destination with attractiveness $x$ are neighbors. Consequently, the Markovian property at the latent variable level \cite{boguna2002epidemic,serrano2007correlations,jiang2014topological} allows to calculate analytical expressions for the assortativity $ k_{nn}(\kappa)$, quantifying two vertices correlations, and clustering coefficient spectrum $c(\kappa)$, as a measure of three vertices correlation. A very important result is that the degree correlations of trip-visit distributions are completely determined by the attraction (and production) rate and by the origin-destination conditional probability $\chi(y|x)$.
			
	\begin{figure}[!ht]
	\centering
	\begin{subfigure}[l]{0.4\textwidth}
		\centering
		\includegraphics[width=0.65\linewidth]{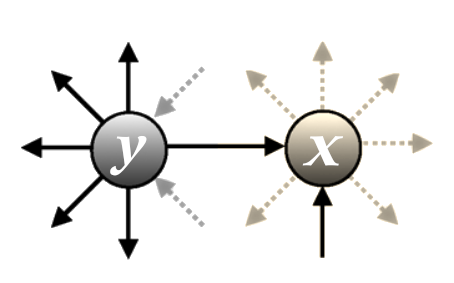}
		\caption{destination assortativity for the node $x$}
		\label{fig_dual}
	\end{subfigure}
	\begin{subfigure}[r]{0.4\textwidth}
		\centering
		\includegraphics[width=0.75\linewidth]{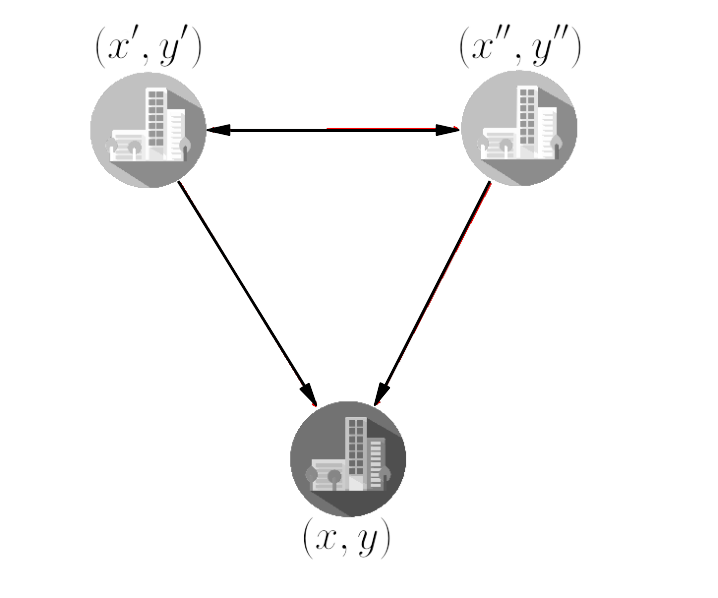}
		\caption{destination clustering for the node $(x,y)$}
		\label{fig_triads}
	\end{subfigure}
\caption{The two point and three point correlations of the trip mobility network can be calculated in terms of the hidden variables $x$ and $y$. In particular the conditional average-nearest-neighbor's strength at the latent variable level (a) is the in-out strength (origin-destination) assortativity coefficient,  and the cluster coefficient at the latent variable level (b) is the "In" clustering (or destination clustering) coefficient for weighted and directed networks as \cite{fardet2021weighted,fagiolo2007clustering}, i.e.  a triangle such that there are two trips coming into of the destination node $(x \leftarrow y', x \leftarrow y'', x' \leftarrow y'' \vee x'' \leftarrow y')$.
}
\end{figure}	
	\begin{proposition}\label{propANND}
		 In the visiting mobility network under the latent variable assumption, the origin-destination correlation is defined as the conditional probability that a visit in the destination of attractiveness $x$ has originated from a location of population $y$, and it is written as:
		\begin{equation}\label{eq_phip}
		\chi(y|x)= \frac{\partial }{\partial y}\log \mathcal{V}_x
		%\chi(y|x)= \frac{\partial}{\partial y} \log \mathcal{N}_x
		\end{equation}
	%	where $\mathcal{V}_x$ is the primitive function of the attraction rate $\nu_x$.
		%\item The origin-destination degree correlations:
		%\begin{equation}\label{eq_ANND}
		%\Bar{k}_{nn}(k,n)=  \frac{n}{P_{\dashv}(k,n)}\int  \int \nu_y\, \rho(x) p_x(k,n)  \chi(y|x)dxdy +1
		%\end{equation}
		As consequence, the following estimates of the two-point and three-point correlations hold:
	\begin{itemize}	
		\item the average out-strength of origin neighbors of  destinations with in-strength $\kappa$,  can be written as:
		\begin{align}
		 {k}_{nn} (\kappa,t) & \sim\frac{t}{P(\kappa,t)}\iint \nu_{y}p(\kappa,t|x) \chi(y|x) \rho(x)  dy dx
		\end{align}
		
		If destinations and origins are independent $\chi(y|x)=\chi(y)$ and $	\langle {k}_{nn}\rangle(\kappa,t)  = \text{const.}$
	 
  \item	  the clustering coefficient for destinations of in-strength $\kappa$ is:
		\begin{align}
		c_{}(\kappa,t)\sim  & \frac{1}{2\nu_0P(\kappa,t)}\iiint p(\kappa,t|x)\rho(x) \left(\nu_{y'}+\nu_{y^{\prime \prime}}\right)\,\chi(y^{\prime}|x)\,\chi(y^{\prime \prime}|x)dy'dy''dx
		\end{align}
\end{itemize}
	\end{proposition}
\begin{proof}
	 The conditional origin-destination probability is the conditional probability that a destination block of attractiveness $x$ is connected to an origin block of population $y$ is:
	\begin{equation}\label{eq_phip}
	\chi(y|x)=\frac{\phi(y) \mathcal{K}(x,y)}{\int \phi(y) \mathcal{K}(x,y) dy} = \frac{\partial \mathcal{V}_x(x,y)}{\partial y}\frac{1}{\nu_x(x,y)} =\frac{\partial}{\partial y} \log \mathcal{V}_x(x,y)
	%\chi(y|x)= \frac{\partial}{\partial y} \log \mathcal{N}_x
	\end{equation}
	where $\mathcal{V}_x$ is the primitive function of the attraction rate $\nu_x$. Io order to write the explicit expression of the origin-destination correlation it is necessary to know the pairing rule, for example trough the connection kernel $ \mathcal{K}(x,y)$ so that $\nu_x=\nu_0\int_{\Omega_y}  \mathcal{K}(x,y)\phi(y) dy$, or imagining a generic primitive function for the attraction rate. So the conditional origin-destination probability is an important indicator for correlations between origins and destinations in the degree-distribution of the visiting mobility network. 
	The conditional average-nearest-neighbor's in-degree for destinations of attractiveness $x$ can be written for directed multigraph in a continuous limit as in \cite{caldarelli2002scale,vanni2021incremental}:
\begin{align}
	k_{nn}(x) & =\int \mathbb{E}[\kappa|y]\, \chi(y|x) dy  
\end{align}
where in the mobility model the conditional expected strength is $E[\kappa|y] \propto \nu_{y} t$. Since, for the markovian degree property, the degree two point correlation can be fully determined by the conditional probability $P(\kappa'|\kappa)$, the average degree of neighbors of an in-degree $\kappa$ destination is known to be calculated as \cite{boguna2003class,murase2019sampling}:
\begin{align*}
	k_{nn}(\kappa,t) & =1+\frac{1}{P(\kappa,t)}\int p(\kappa,t|x)\rho(x) k_{nn}(x)  dx  =1+\frac{t}{P(\kappa,t)}\int \int \nu_y\chi(y|x)p(\kappa,t|x) \rho(x)  dy dx
\end{align*}
which is an in-out (origin-destination) assortativity measure, which is independent of $\kappa$ for $\chi(y|x)=\chi(y)$ as in the case of multiplicative separable linking function $ \mathcal{K}(x,y)$. Let us notice that since the network is directed, other than weighted, one could define, similarly, other three average nearest neighbor's degree functions: destination-origin, origin-origin and destination-destination.

	The clustering coefficient of a destination with attractiveness $x$ can be interpreted as the probability that two
	randomly chosen edges from $x$ are origin-neighbors. The clustering of a destination of degree one or zero is defined as zero.
	In the space of latent variables, consider a destination $i$ of attractiveness $x_i$ and population $y_i$, which is connected with with probability $p(y_j,y_k|x_i)$  through trips originated from two other locations $j$ and $k$ which have attractiveness $x_j$ and $x_k$ and population $y_j$ and $y_k$ respectively. Since the network is markovian at the latent variable level, $p(y_j,y_k|x_i)=p(y_j|x_i)p(y_k|x_i)$. Thus, similarly to the definition in \cite{boguna2003class,van2017local} together with modifications \cite{fardet2021weighted,clemente2018directed} for the directed and weighted case, the local origins-destination clustering for locations of attractiveness $x_i$ can be written as:
		\begin{align*}
c_{}(x_i) &=\sum_{j,k}p\big((x_j,y_j),(x_k,y_k)\big)\,p(y_j,y_k|x_i) =\frac{\sum_{j,k} \tfrac{1}{2}\big(\mathcal{K}(x_j,y_k)+\mathcal{K}(x_k,y_j)\big)\,\mathcal{K}(x_i,y_j)\,\mathcal{K}(x_i,y_k)}{\sum_{j,k} \mathcal{K}(x_i,y_j)\,\mathcal{K}(x_i,y_k)}
\end{align*}
where $p\big((x_j,y_j),(x_k,y_k)\big)$ is the probability that the two origin nodes are connected one to the other in both directions.
Now, in the asymptotic continuous regime  the clustering coefficient can be rewritten:
		\begin{align*}
c(x) &=M\iiiint \tfrac{1}{2} (\mathcal{K}(x',y'')+ \mathcal{K}(x'',y'))\,\mathcal{K}(x,y')\,\mathcal{K}(x,y'')\, \rho(x')\rho(x'')\phi(y')\phi(y'')dx' dx''dy'dy''
\end{align*}		
where $M=(\int \mathcal{K}(x,y')\phi(y')dy' \int \mathcal{K}(x,y'')\phi(y'')dy'')^{-1}$, so it is possible to write:
\begin{align*}
c(x)&=\iiint \chi(y'|x)\mathcal{K}(x',y'')\chi(y''|x)\rho(x')dx'dy'dy'' 
=\frac{1}{2\nu_0}\iint \left(\nu_{y'}+\nu_{y''}\right)\,\chi(y'|x)\,\chi(y''|x)dy'dy''
\end{align*}
knowing that $\nu_y=\nu_0\int_{\Omega_x}  \mathcal{K}(x,y)\rho(x) dx$ and the definition of $\chi(y|x)$.
%\begin{align*}
%c(x) &=\tfrac{M}{2}\iiint [f(x',y'')\rho(x')]\,[f(x,y')\phi(y')]\,[f(x,y'')\phi(y'')]dx'dy'dy'' +\\
%&+\tfrac{M}{2}\iiint [f(x'',y')\rho(x'')]\,[f(x,y')\phi(y')]\,[f(x,y'')\phi(y'')]dx''dy'dy'' \\
% %&=\iiint ([f(x',y'')\rho(x')]dx')\,\chi(y'|x)\,\chi(y''|x) dy'dy''\\
% &=\frac{1}{2\nu_0}\iint \left(\nu_{y'}+\nu_{y''}\right)\,\chi(y'|x)\,\chi(y''|x)dy'dy''
%% &= \tfrac{1}{2}\nu_0k_{nn}(x)\left( \int \chi(y''|x)dy'' + \int \chi(y'|x)dy'\right)
%\end{align*}
Let us notice that for independent origins and destinations then $c(x)= c_0=const$. Moreover in the case of clustering coefficient in a multigraph one can calculate the number of triangles repeated $\kappa$ times, which in a markovian graph can be approximated on average as $\overline{c}_{t}(x)=tc(x)$, since at each time step a possible link is considered as a bernoullian trail, so that the observed number of links in $t$ trials follows a binomial distribution and so the expected value is $t \mathcal{K}(x,y)$. Consequently, since the clustering coefficient has values in $[0,1]$, we normalize the adjacency matrix respect to $t$, so that  the average local clustering coefficient of a node with strength $\kappa$, denoted by $c(\kappa)$, \cite{serrano2007correlations,boguna2003class,stegehuis2017clustering,murase2019sampling} is given  by:
		% $c(\kappa)=\int \int P(\kappa',\kappa''|\kappa)p$,
		 \begin{align*}
		 c(\kappa,t)&=\frac{1}{tP(\kappa,t)}\int p(\kappa,t|x)\overline{c}_{t}(x)\rho(x)dx
		 =\frac{1}{P(\kappa,t)}\int p(\kappa,t|x)c(x)\rho(x)dx
		 \end{align*}
		 which represents the local in-clustering spectrum for destination locations and it is independent of $\kappa$ for $\chi(y|x)=\chi(x)$ as in the case of multiplicative separable linking function $ \mathcal{K}(x,y)$, where $P(\kappa,t)$ represents the in-strength distribution.	 
%		 Finally,  the global average in-clustering  coefficient in the mobility network is then given as:
%		 \begin{align*}
%		 \langle c(t) \rangle = \int c(\kappa,t) P(\kappa,t)d\kappa= \iint p(\kappa,t|x)c(x)\rho(x)dx d\kappa
%		 \end{align*}	 
Let us notice that in the case of the clustering coefficient the adjacency matrix and the latent variables are needed to be normalized in order to transform a multigraph in a weighted graph and from there a clustering coefficient no larger than 1 is guaranteed.
\end{proof}

	%\begin{proposition}[\underline{Bursty visitation model}]
	%The bursty visitation of locations in human mobility
	%\end{proposition}
	%Bursty visitiation patterns \cite{lv2021bursty}
The Markovian nature of this class of networks implies that all higher-order correlations can be expressed as a
function of the attraction and production rates $\nu_x,\nu_{y}$ and the conditional origin-destination probability $\chi(y|x)$, allowing an exact treatment of mobility models at the mean-field  level. 	 
Under the hypothesis that origins and destinations are independent, that is $\chi(y|x)=\chi(y)$, then the  average-nearest-neighbor's strength function  and the clustering coefficient are constant along $\kappa$ as shown, as an example, in the simulation shown in Fig.\ref{fig_ANNDsimul} and Fig.\ref{fig_clustersimul}. Consequently, for neutral networks the two and three point correlations can be obtained with three different approaches that will provide the same estimate by using the input approach of latent variables ('latent estimate'), by using the output approach trough the  adjacency matrix ('expected value') and, finally, the algorithm computation of assortativity and clustering coefficients for directed and weighted networks ('simulation approach'). 
	
	\begin{lemma}\label{coro2}
Under the hypothesis that visit production process and visit attraction process are independent 
 the average in-strength of nearest neighbor function is constant, and in the asymptotic limit:
 \begin{align}
 k_{nn}(\kappa,t)\sim \frac{t\mathbb{E}[h^2(y)]}{N\mathbb{E}[h(y)]^2}=\frac{\langle \kappa_{out}^2\rangle}{\langle \kappa_{out}\rangle}
 \end{align}  
 \medskip
 As regard with the clustering coefficient spectrum, under the same hypothesis:
 \begin{align}
 c(\kappa)\sim \frac{\mathbb{E}[g({x})]\mathbb{E}[h^2({y})]}{\mathbb{E}[h({y})]}=\frac{\langle \kappa_{in} \rangle}{tN}\left( \frac{\langle \kappa_{out}^2 \rangle -\langle \kappa_{out} \rangle}{\langle \kappa_{out} \rangle^2}\right)^2
 \end{align}
\end{lemma}	\label{eq_clusteruncorr}
\begin{proof}
The equations for the average in-strength of nearest neighbors and the in-clustering coefficient can be derived directly from Proposition \ref{propANND} after some algebraic manipulations considering the production rate $\nu_y$ and the conditional probability $\chi(y|x)$ are independent from $x$ as, in the case when the attraction and production rates are recovered from a linking function that is multiplicative separable, i.e. $ \mathcal{K}(x,y)=g(x)h(y)$. In fact, when origins and destinations are independent, that is $\chi(y|x)=\chi(y)$, then the  average-nearest-neighbor's strength function 
%\begin{align*}
$k_{nn}(\kappa,t) =1+ t\int \nu_y \chi(y)dy$
% \end{align*}
which is a constant over the strength degrees $\kappa$ and where ${E[g(x)]}$ and ${E[h(y)]}$ are the expectation values of the function $g(x)$ and $h(y)$ under the circumstances they have finite values, which occurs even for fat-tail distributions in a finite set for the latent variables \cite{vanni2021incremental}. In the case of infinite moments then the equation above represents a unreliable estimation but still shows the neutral assortativity of the graph. 
 As regard with the clustering coefficient, the result in \eqref{eq_clusteruncorr} is straightforward by using the multiplicative separable linking function as above, with the only difference that it has been normalized in order to provide a weighted matrix with link weights not larger that one.
Another approaches for the estimates of the assortativity and clustering coefficient can be derived in terms of the strength degrees of the nodes as provided in \cite{boguna2003class,serrano2007correlations,latora2017complex} with proper modifications for directed weighted graph, see \cite{boguna2005generalized,serrano2006correlations,vanni2021incremental}.
Their expected values for $k_{nn}(\kappa)$ and $c(\kappa)$ are analytically known in literature for neutral networks, i.e. no degree correlations, as derived in terms of degrees the expected average in-strength of nearest neighbors can be written:
\begin{align*}
k_{nn}(\kappa_{in})&=\sum_{\kappa_{out}}\kappa_{out}P(\kappa_{out}|\kappa_{in})	=\sum_{\kappa_{out}}\kappa_{out}\frac{\kappa_{out}P(\kappa_{out})}{\langle \kappa_{out}\rangle}=\tfrac{\langle \kappa_{out}\rangle^2}{\langle \kappa_{out}\rangle}=\langle k_{nn}\rangle= \mathbb{E}[k^{(u)}_{nn}]
\end{align*}
where in the absence of correlations $P(\kappa_{out}|\kappa_{in})=k_{out}P(\kappa_{out})/\langle \kappa_{in}\rangle$ has been used.
In the case of clustering coefficient, after normalization of the multigraph, the in-clustering coefficient can be written as in \cite{dorogovtsev2004clustering,serrano2007correlations}:
 \begin{align*}
  c(\kappa)=&\sum_{\kappa'_{out},\kappa''_{out}}\frac{(\kappa'_{out}-1)(\kappa''_{out}-1)}{tN\kappa''_{out}P(\kappa''_{out})}P(\kappa''_{out}|\kappa'_{out})P(\kappa''_{out}|\kappa_{in})P(\kappa'_{out}|\kappa_{in})  \\
 =&\frac{\langle \kappa_{in} \rangle^3 }{tN\kappa_{in}^2P^2(\kappa_{in})}\sum_{\kappa'_{out},\kappa''_{out}} \frac{(\kappa'_{out}-1)(\kappa''_{out}-1)P(\kappa''_{out},\kappa'_{out})P(\kappa''_{out},\kappa_{in})P(\kappa'_{out},\kappa_{in})}{\kappa'_{out}\kappa''_{out}P(\kappa'_{out})P(\kappa''_{out})}
  \end{align*}
That in the case of uncorrelated networks:
\begin{align*}
c(\kappa) =\frac{\langle \kappa_{in} \rangle^3}{tN\kappa_{in}^2P^2(\kappa_{in})} \sum_{\kappa'_{out},\kappa''_{out}}&\frac{(\kappa'_{out}-1)(\kappa''_{out}-1)}{\kappa'_{out}\kappa''_{out}P(\kappa'_{out})P(\kappa''_{out})}\cdot \\ &
 \cdot\frac{k''_{out}P(k''_{out})}{\langle k''_{out}\rangle}\frac{k'_{out}P(k'_{out})}{\langle k'_{out}\rangle}\cdot 
 \frac{k''_{out}P(k''_{out})}{\langle k''_{out}\rangle}\frac{k_{in}P(k_{in})}{\langle k_{in}\rangle}\cdot 
  \frac{k'_{out}P(k'_{out})}{\langle k'_{out}\rangle}\frac{k_{in}P(k_{in})}{\langle k_{in}\rangle}\\
   =\frac{\langle \kappa_{in} \rangle}{tN\langle \kappa_{out} \rangle^4 } \sum_{\kappa''_{out}}(\kappa''_{out}-1&)\kappa''_{out}P(\kappa''_{out})\sum_{\kappa'_{out}}(\kappa'_{out}-1)\kappa'_{out}P(\kappa'_{out})=\langle c\rangle =\mathbb{E}[c^{(u)}]
     \end{align*}
\end{proof}
	
\begin{figure}[!ht]\label{fig_nocorre}
	\centering
	\begin{subfigure}[l]{0.65\textwidth}
		\centering
		\includegraphics[width=0.95\linewidth]{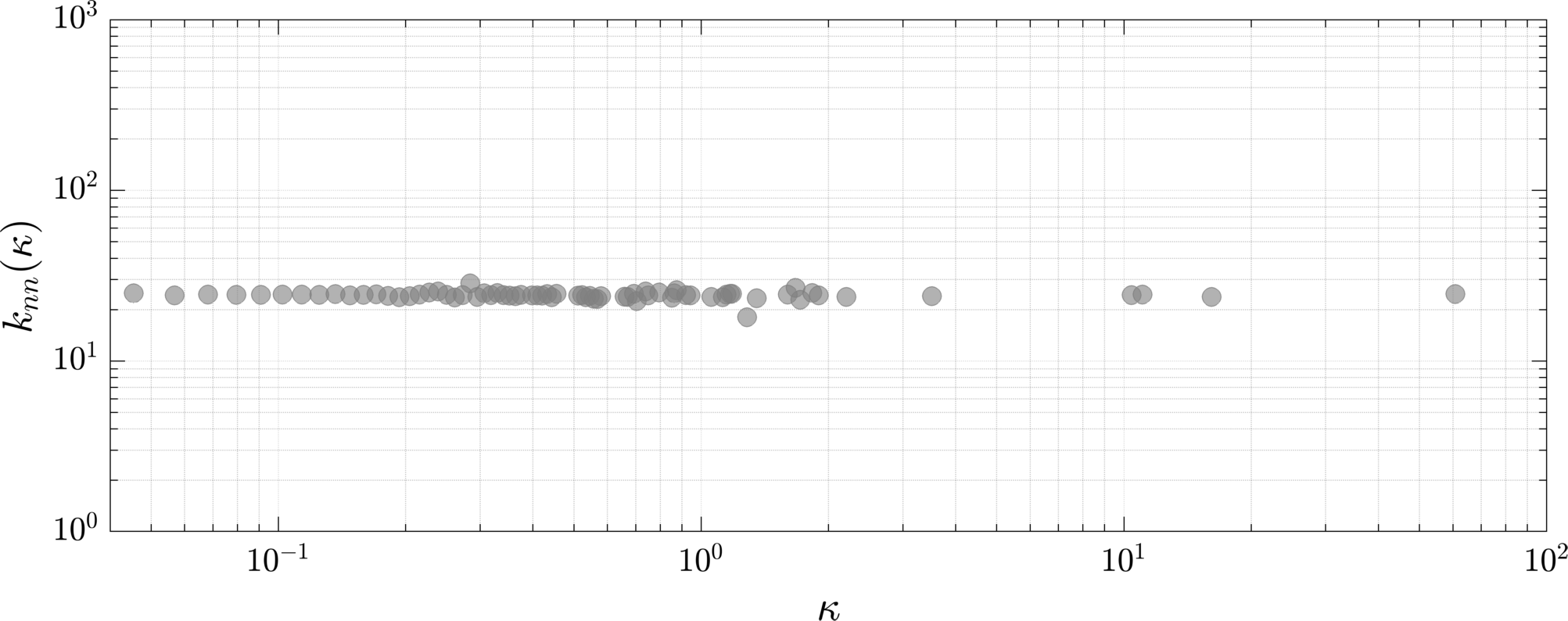}
		\caption{destination assortativity spectrum}
		\label{fig_ANNDsimul}
	\end{subfigure}
\\
\vspace{0.5cm}
	\begin{subfigure}[r]{0.65\textwidth}
		\centering
		\includegraphics[width=0.95\linewidth]{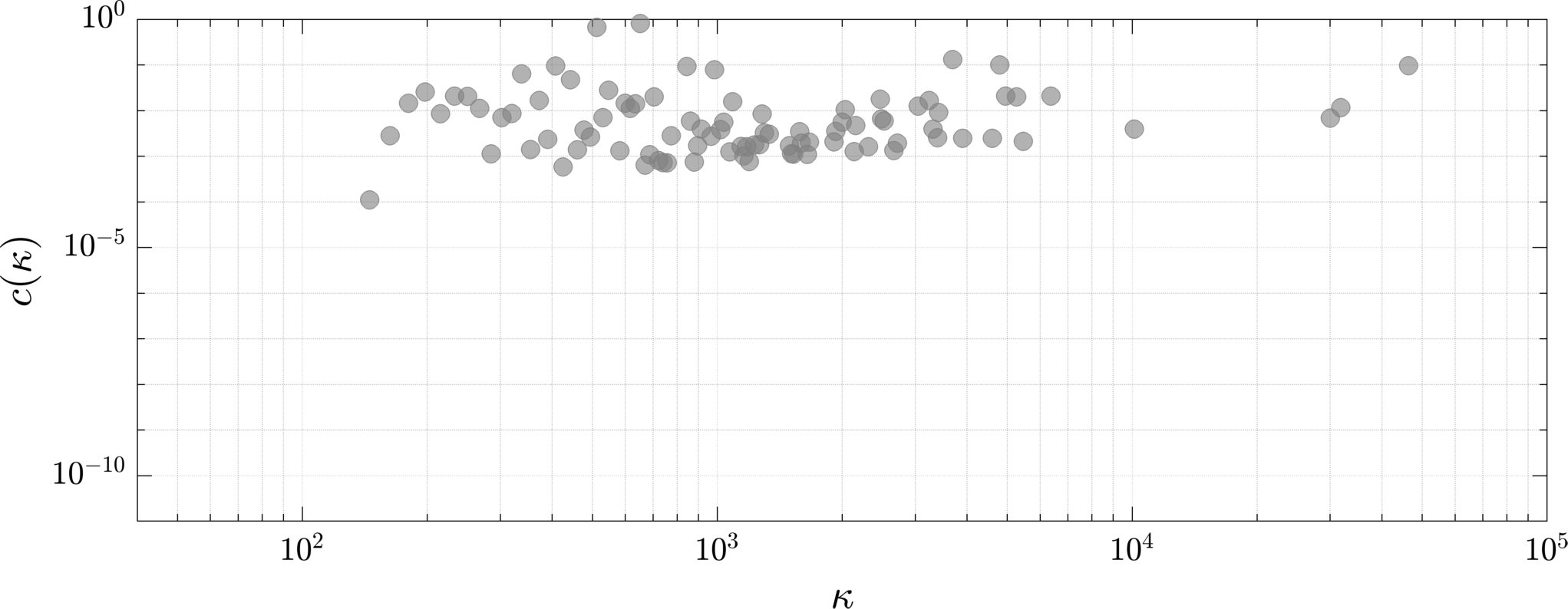}
		\caption{destination clustering spectrum}
		\label{fig_clustersimul}
	\end{subfigure}
	\caption{ The estimate of assortativity and transitivity of a latent variable network 	with the same structure as in Fig.\ref{fig_3frames} with the additional specification of $\phi(y)\sim \phi_0 y^{-2}$. The average-nearest-neighbor's strength function (a) which show a neutral assortativity in the network, the dashed line represents the assortativity mean value. The local clustering coefficient for different values of $\kappa$ in (b), so that the transitivity is constant, the dashed line represent the average global cluster coefficient.}
\end{figure}
 Let us notice the the clustering coefficient is normalized respect to time $t$ since $c(\kappa)\in [0,1]$ and so it results to be the generalization for directed multigraphs without correlations as for simple graphs in \cite{servedio2004vertex,boguna2003class,van2017local}.	Simulations for such results are shown in Fig.\ref{fig_correlationtime} where several computational simulations of a graph for different time length $t$ is presented alongside the prediction results for uncorrelated networks for assortativity and clustering coefficient.   It is worth noticing that the analytical prediction are asymptotically valid so that no isolated nodes or leafs exists since local neighborhood clustering is typically not defined if a node has one  or no neighbor and such situation  influences the estimation of the global clustering in sparse networks \cite{kaiser2008mean}. In the present work, the clustering coefficient algorithm removes all the local clustering of all the nodes with less than 2 neighbors, so the global clustering coefficient is over-estimated\footnote{Another choice would be to set to zero the local clustering coefficient for all nodes with less than two neighbors. Is such case the global clustering coefficient would be under-estimated}. 	
\begin{figure}[!ht]
	\centering
	\begin{subfigure}[l]{0.95\textwidth}
		\centering
		\includegraphics[width=0.75\linewidth]{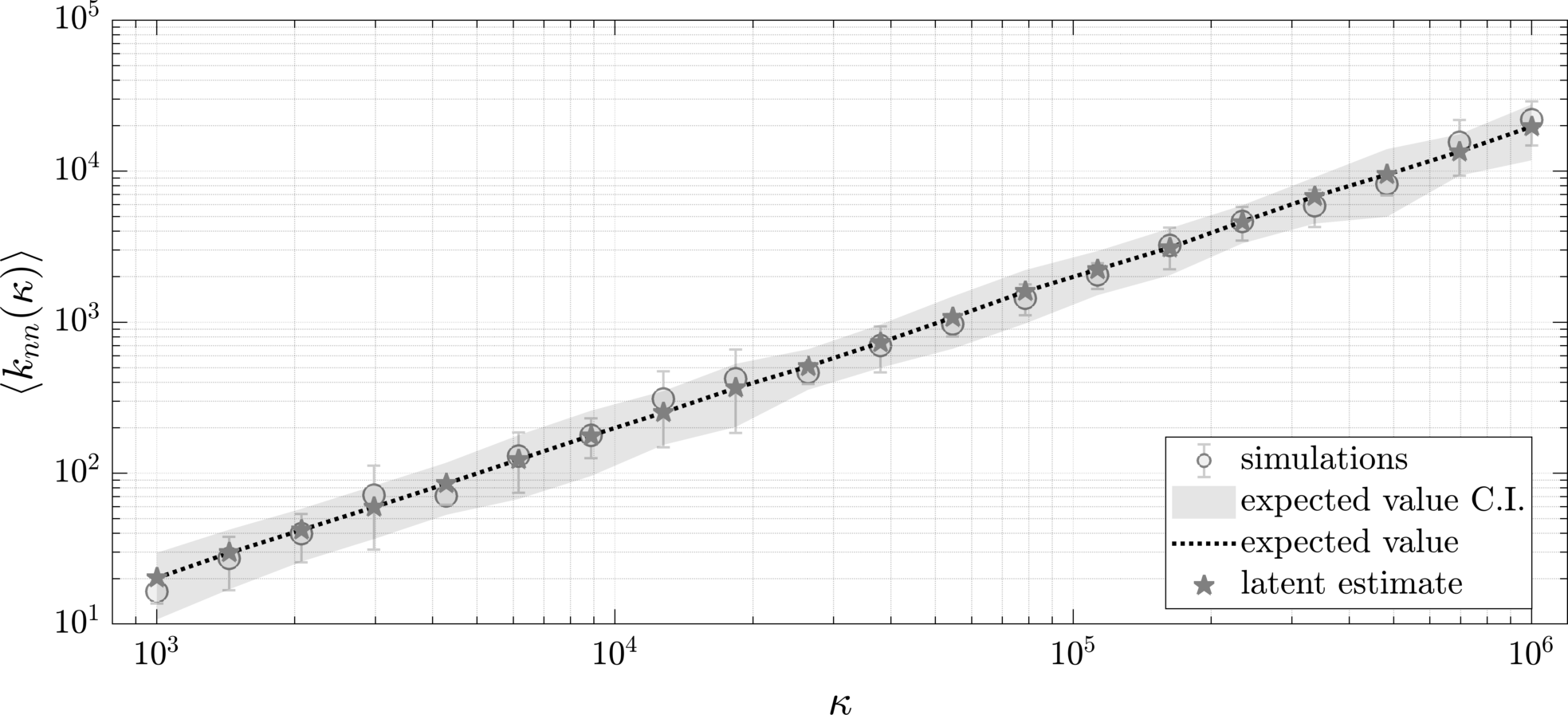}
		\caption{mean in-assortativity for different simulation times}
		\label{fig_ANNDsimultime}
	\end{subfigure}
	\\
	\vspace{0.5cm}
	\begin{subfigure}[r]{0.95\textwidth}
		\centering
		\includegraphics[width=0.75\linewidth]{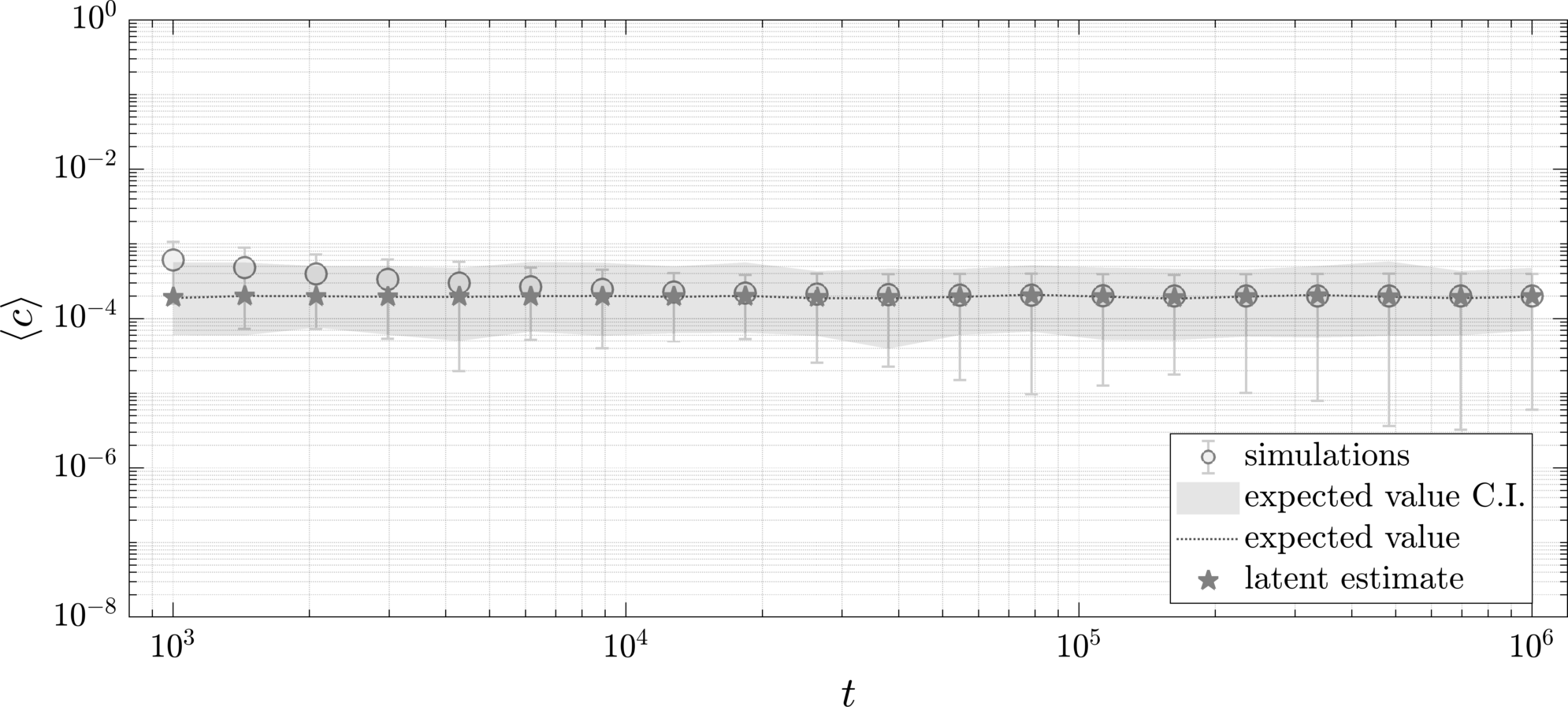}
		\caption{mean in-clustering coefficient for different simulation times}
		\label{fig_clustersimultime}
	\end{subfigure}
	\caption{Two points and three point network correlations using three different approaches, in the case of uncorrelated graph  as reported in Remark \ref{coro2}. The overall average strength of nearest neighbors $\langle k_{nn}(\kappa)\rangle_{\kappa}$ is replicated  for each $t$ so to obtain a mean global value $\langle k_{nn} \rangle$ over a ensemble of $S=50$ replications as in (a). Similarly one can obtain the global mean in-clustering coefficient $\langle c \rangle$ in (b).}\label{fig_correlationtime}
\end{figure}

%	\clearpage
	
\section{Results}	\label{sec_results}
In the present section a network analysis of the main graph measures and topology will be conveyed in the particular case study of New York metropolitan area by using Safegraph Mobility Dataset \cite{safegraph}	for the year 2019. 
\subsection{Data}	
Origin-destination (OD) data  represent movement flows through geographic space, from an origin (O) to a destination (D). OD datasets represent information on trips between two geographic areas  often represented by the geographical centroids of the areas. Typically encoded with a square symmetric matrix, OD flow data contain numerical data on the aggregate quantity of individuals travelling from one geographic area to another  over a specific time period. Mostly used in
transportation planning, OD flows are an invaluable source of data for understanding spatial and temporal patterns of urban mobility and dynamics \cite{martin2018origin,rodrigue2020geography,batty2013new,bettencourt2021introduction}.
Visit flows can be in practice estimated in various ways in real world data. In particular, mobile phone location data are provided by SafeGraph trough dynamic population Origin-Destination flow matrices with hourly temporal resolution and   aggregated by census block groups (CBG) in the USA as discussed in \cite{kang2020multiscale}.  In the daily CBG to CBG visitor flows metric, each row contains an origin CBG and a destination CBG, as well as the number of mobile phone-based visitor flows from the origin CBG to the destination CBG. Every day, the number of unique mobile
phone users who live in the origin CBG and visits to the destination CBG are recorded.
As regarding with visit production  model, the population in each block is the key information to obtain from data in order to define the variable $y$ and its respective distribution. However the population data is susceptible to the way data are collected and sampled by the provider. In fact the demographic sampling  depends on many factors as the geographical boundaries which define a block, a tract, or any administrative tessellation. Moreover in the statistical sampling methodology  the individual measurements in each block go through a few transformations and aggregations  which impacts the final measurement \cite{scheaffer2011elementary}.

From SafeGraph data it is possible to build a matrix of trip flows between locations in a day for arbitrary large region $R$ of the US. In particular, a county is considered with a tessellation at the resolution of census block group level, then it is possible to reconstruct the adjacency matrix of directed trips from an origin location towards a destination as stops  of devices as described in \cite{safegraph}.  Fig.\ref{fig_cityNY}  plots a sequence of visitation counts in different census block group areas during different time windows of a day for New York city.

\begin{figure}[!ht]
	\centering
	\begin{subfigure}[l]{0.32\textwidth}
		\centering
		\includegraphics[width=0.95\linewidth,trim={0 0 7cm 0},clip]{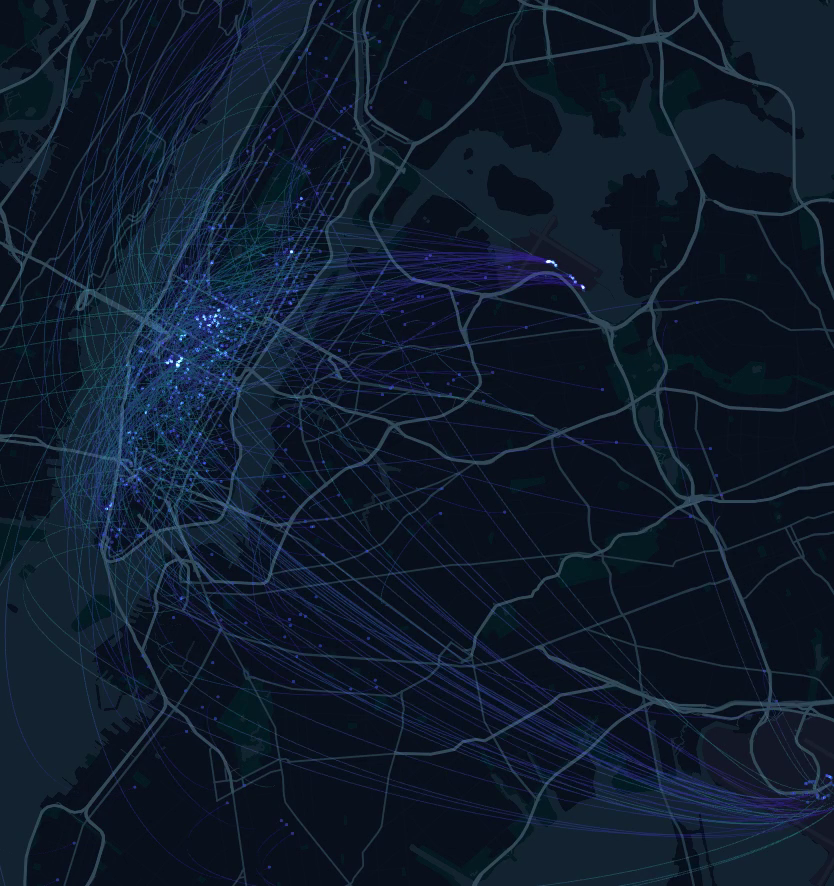}
		\caption{5am-6am}
	\end{subfigure}
	\begin{subfigure}[l]{0.32\textwidth}
		\centering
		\includegraphics[width=0.95\linewidth,trim={0 0 7cm 0},clip]{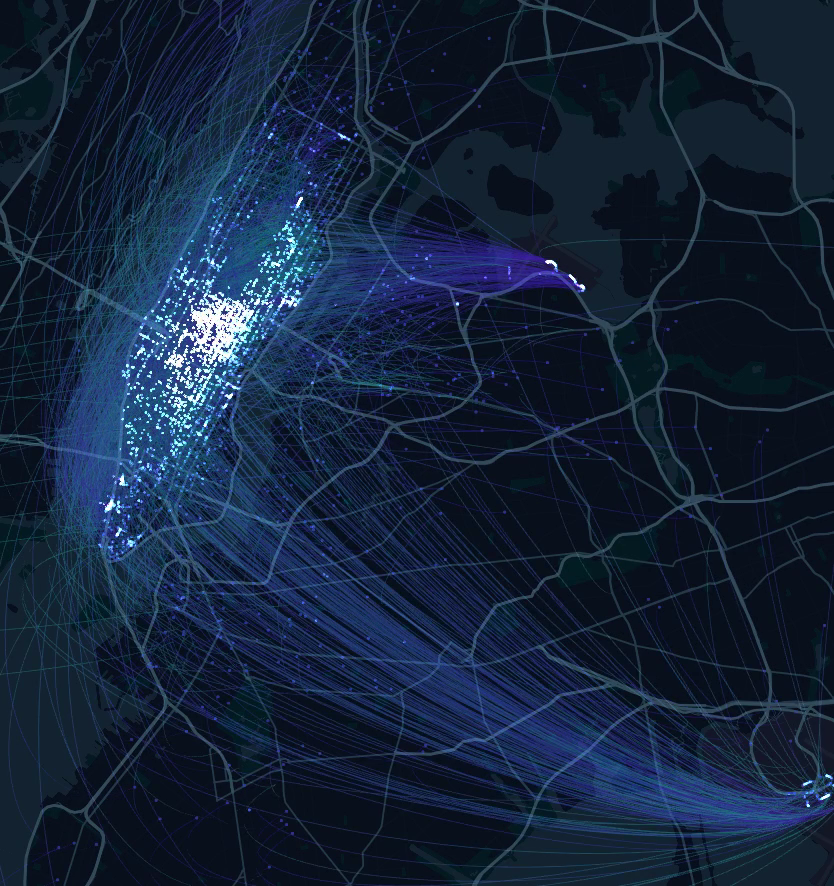}
		\caption{5am-9am}
	\end{subfigure}
	\begin{subfigure}[l]{0.32\textwidth}
		\centering
		\includegraphics[width=0.95\linewidth,trim={0 0 7cm 0},clip]{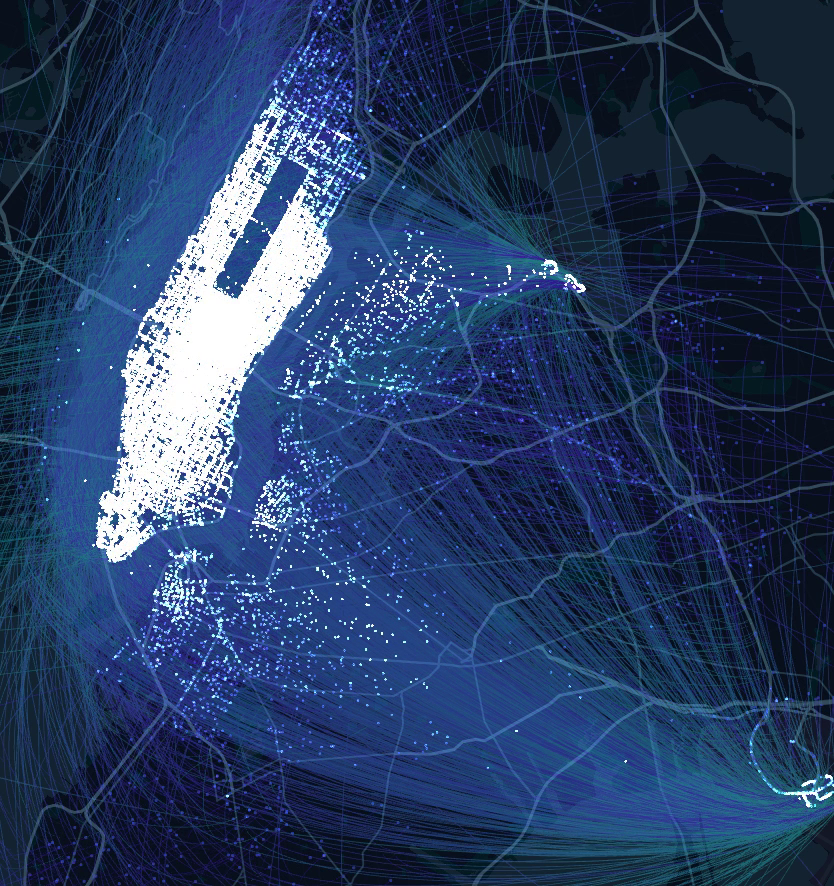}
		\caption{5am-12pm}
	\end{subfigure}
	\caption{Number of cumulative visits  in New York city \cite{keplerNY} for different time windows.}
	\label{fig_cityNY}
\end{figure}

\begin{table}[ht]
	\centering
	\medskip
	${\huge \boldsymbol{A}} = \left[
	\begin{array}{l|cccc|c}
	&     & \boldsymbol{O}    &    &     & {\footnotesize \rotatebox{360}{$W_O$ }}\\ \hline
	& a_{11} & a_{12} & \dots  & a_{1N} & w_1 \\
	\rotatebox{90}{$\boldsymbol{D}$}  & \vdots & \vdots & \ddots & \vdots & \vdots\\
	& a_{N1} & a_{N2} & \dots  & a_{NN}& w_N\\
	\hline
	{\footnotesize \rotatebox{90}{$W_D$ }}	 & w^{1} & w^{2} & \dots  & w^{N}&    v_W
	\end{array}
	\right]= \left(
	\begin{array}{c|c}
	\raisebox{-15pt}{\huge\mbox{ { $\boldsymbol{A}_0$ }  } }  & w_1 \\[-4ex]
	&\vdots  \\[-0.5ex]
	& w_N \\ [0.5ex] \hline \\[-2ex]
	w^1 \cdots w^{N}	 &  v_w
	\end{array}
	\right)$
	\caption{Data matrix format. The vector $\boldsymbol{D}$ represents the array of locations as destination units in the tessellation region. Similarly  $\boldsymbol{O}$ represents the same array locations but as origins of trips inside the tessellation region.  The array $W_D$ is the set of destinations located  outside the region and $W_O$ is the set of all the origin locations outside the region. So $\boldsymbol{A}_0$ is the open O-D table, and $\boldsymbol{A}$ is the close O-D table.  }\label{tab_SFmatrix}
\end{table}	

The data has been re-organized as shown in Table \ref{tab_SFmatrix} where the $(N+1)\times (N+1)$ matrix $\boldsymbol{A}$ indicates the global origin-destination table visiting flow between the $N$ blocks of the region $R$ plus one external node which represents the resto of the world available in the data outseide the region of interes $R$. In particular, $a_{ij}$ is the number of visits registered as stop in the destination $j$ in $R$ that originated in the location $i$ in $R$. For locations outside the selected region, the entry $w_i$ counts the visits in the destination $i$ of the region $R$ that originated from a location outside the region $R$. The in-degree for the destination location $i$ is the sum along the columns of the row $i$ from the matrix $\boldsymbol{A}$ from which the empirical visit distribution is evaluated. In addition, the  trip-visit distribution, aka the in-strength distribution, is evaluated by associating a weight to each visit. As a typical choice, the weight is taken to be the distance between the origin block and the destination one in kilometers as estimate of the distance from home travelled by devices. Such information is recovered by the census bureau geographical data using Census Block Group geometries with longitude and latitude coordinates of the block centroid \cite{censustiger,safegraph}, calculated as the haversine distance between the visitor's home geohash-7 and the destination location geohash-7 for each visit. A more detailed  estimate would be the effective distance traveled by each visitor in the trip between its main location to the selected destination. Such information is not reported in SafeGraph at the moment neither in other data sources consistent with the data structure in the study. However,  the radial movement approximation is motivated by the fact that travelers typically seek the shortest route \cite{universalschlapfer2021,batty2013new}.

\subsection{Mobility network analysis}
Let us start with the estimate of the distribution of visits among different destinations, which namely represent the in-strength distribution so that the in-strength of destinations is  $\kappa_i=\sum_j C_{ij}A_{ij}$ where $A_{ij}$ is the entry of the adjacency matrix indicating the number of arrivals in destination $i$ of a trip originated in the location $j$. Whereas $C_{ij}$ is the visit "size" of the traveler who departed from origin $j$ and has arrived to destination $i$. Such value is taken from  weight matrix $C$ that represents, in this particular case, the distances between origin-destination pairs. Such cost matrices are directly recovered from census data included in the dataset used.  In Fig.\ref{fig_NYstat}(a) the empirical complementary cumulative distribution function is plotted for the case of New York metropolitan area in a typical day of November 2019. The inspection of in-strength distribution shows that the visit distribution has a scale-free asymptotic behavior as $P(\kappa)\sim \kappa^{-\mu}$ with power law coefficient of $\mu\approx 1.8$.
Let us now discuss the degree correlations such as assortativity and clustering just in the case of the origin-destination matrix.
As already discussed, the average out-strength of neighbors of destinations of in-strength $\kappa$ measures the tendency of having a directed trips from an origin location to a destination, defined as in \cite{vanni2021incremental,architectPNAS}, and from here the assortativity spectrum can be built. As plotted in Fig.\ref{fig_NYstat}(b), the average nearest neighbor in-strength function  is flat, and this shows a neural assortativity behavior with a mean value of $\langle k_{nn}(\kappa) \rangle_{\kappa} \approx 15.4 $. Such estimate is in agreement with the analytical prediction of the  expected average in-strength of the nearest neighbor for uncorrelated networks  $ \mathbb{E}[k^{(u)}_{nn} ] = \langle \kappa_{out} ^2\rangle /\langle \kappa \rangle=15.7$ as proposed in Remark \ref{coro2} . Under the same conditions, the average in-clustering coefficient can be computed accordingly to the definition \cite{clemente2018directed,clementecode} adapted to the in-clustering coefficient defined in the present work. 
 The clustering spectrum for the data is shown in Fig.\ref{fig_NYstat}(c) where the clustering spectrum is flat as for uncorrelated-graphs, and the global clustering coefficient is given by $\langle c_{in}(\kappa)\rangle \approx 0.02 $ consistently with the analytical prediction of $\mathbb{E}[c^{(u)}]$ for uncorrelated networks as reported in Remark \ref{coro2}. This shows that the O-D SafeGraph mobility network is consistent with the hypothesis of uncorrelated graph with scale-free visit distribution at a macroscopic scale as also discussed in \cite{bettencourt2021introduction}. The absence of degree-correlations allows to consider the origin-destination conditional probability to be $\chi (y|x)=\chi(y)$. This means that a destination receives a visit from a randomly chosen origin location, without any particular choice of the origin but the number of resident populations in it. Consequently, the correlations between origins and destinations are entirely due to trip costs represented in the weight matrix is a well-mixed locations at level of attractiveness property  as specified in the model assumptions.
 
 The scale free behavior of the visit distribution together with the neutral tendency of degree correlations, allows us to narrow the type of kernel function to be considered in the model.  Despite that, we do not have enough information to uniquely determine the attraction rate and the latent variable distribution. We will face such issue in the nex section when we will discuss possible proxies of attractiveness latent-variable. % The scale-free behavior, in fact, does not show up only in the degree and strength distribution, but additional evidences  can be observed in the power law behavior of the betweenness distribution $P(b)$ which indicates the probability that any given destination node is traversed by  $b$ shortest paths between origin-destination pairs. 
% The role of trip-costs in terms of distance can be investigated trough the latent-variable model where the impact plays directly on the destination attractiveness and not in the correlation between origins and destinations. Such behavior emerges at the mesoscopic scale of location blocks rather then at microscopic level of individuals, who, in the present work, are considered more like the links between regions via  a visitation multi-purpose dynamics. 
\begin{figure}[!h]
	\begin{minipage}{.5\linewidth}
		\centering
		\includegraphics[width=.82\linewidth]{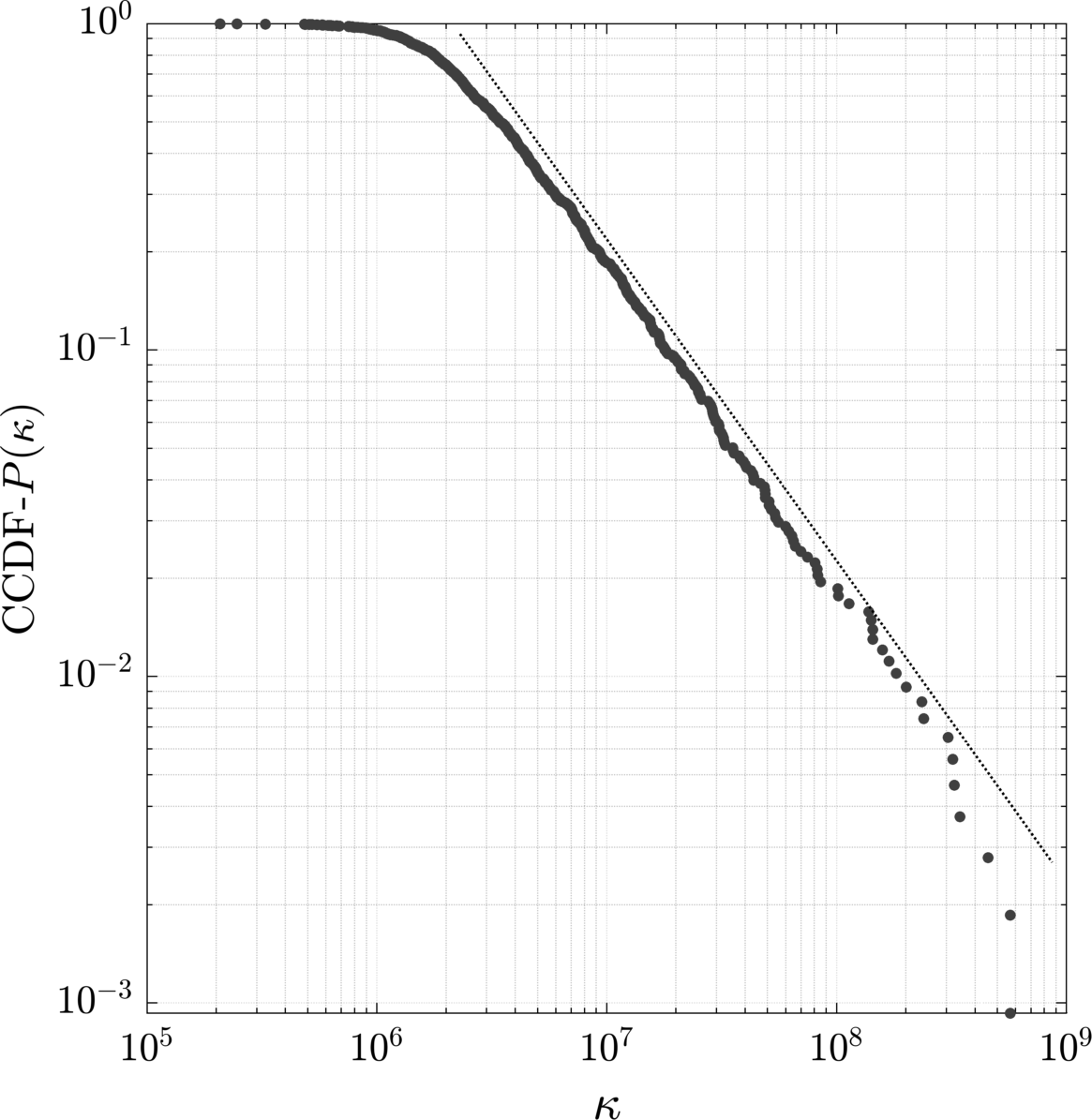}
		
		\small (a)
		
	\end{minipage}%
	\begin{minipage}{.5\linewidth}
		\makebox[.85\linewidth]{\includegraphics[width=1\linewidth]{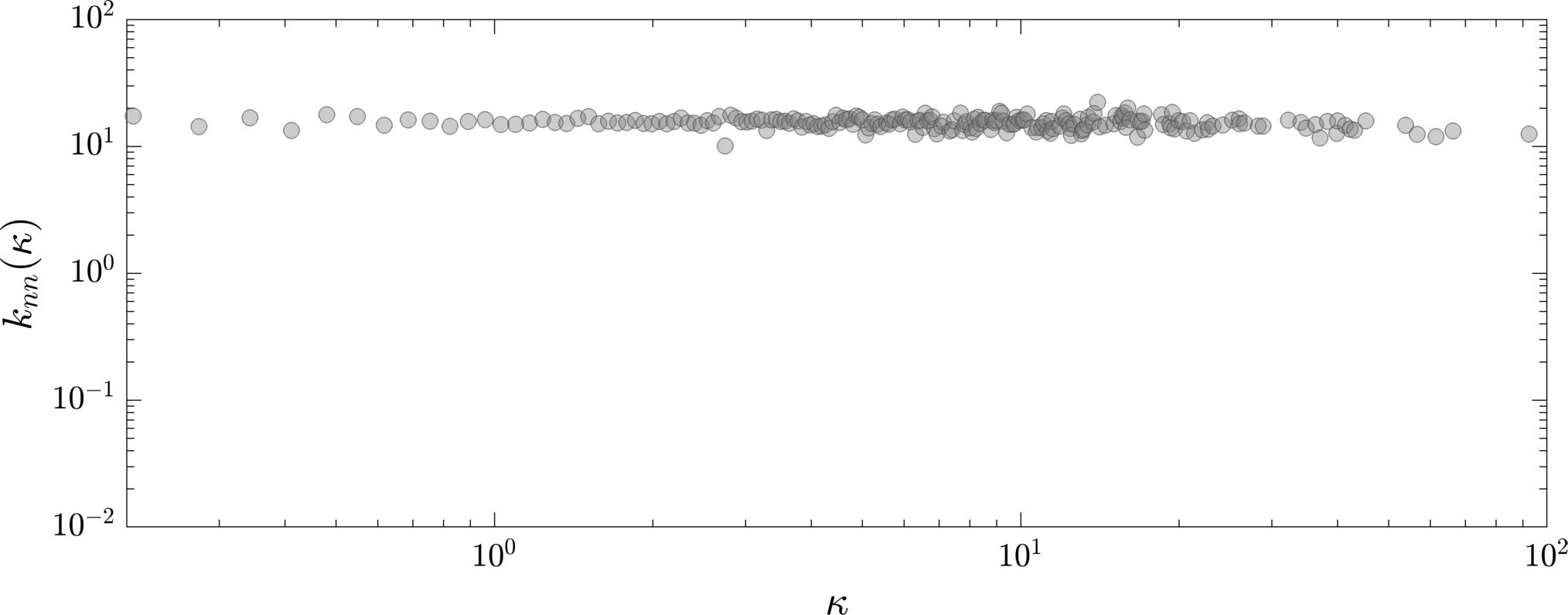}}\\
		\makebox[.85\linewidth]{\small (b)}%
		
		\medskip
		
		\makebox[.85\linewidth]{\includegraphics[width=1\linewidth]{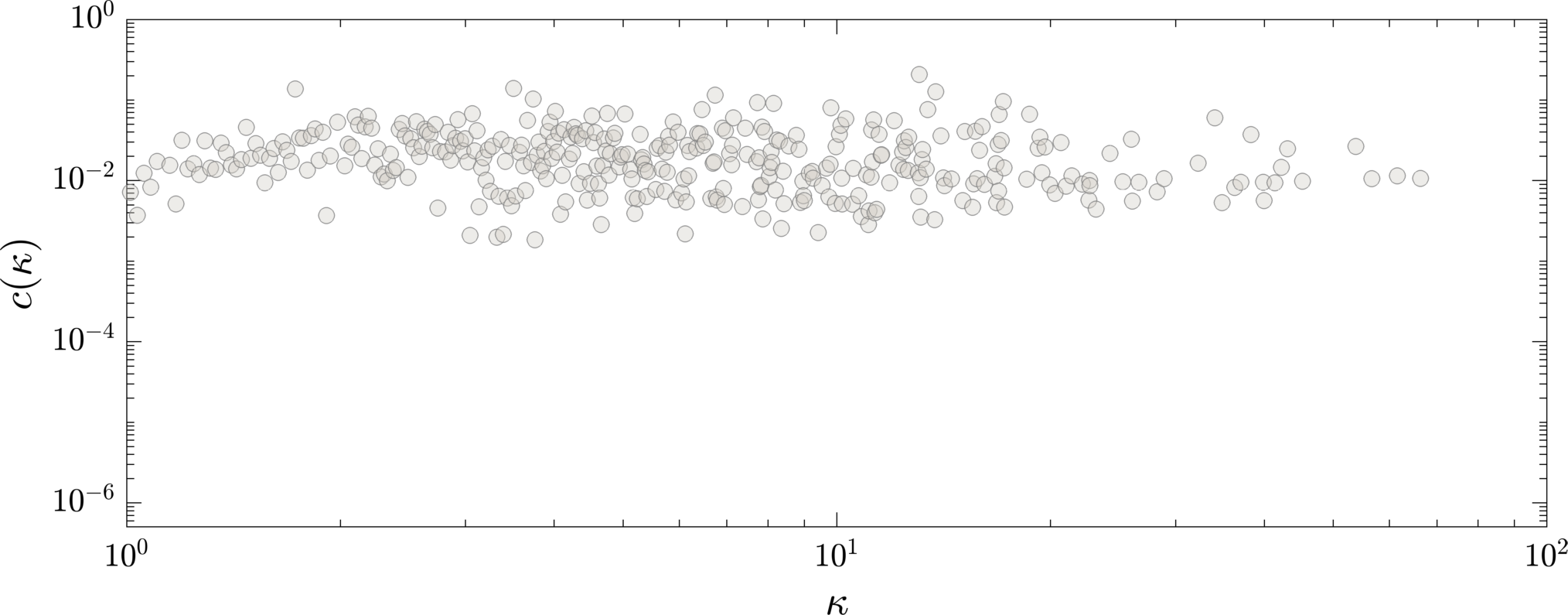}}\\	
		\makebox[.85\linewidth]{\small (c)}%
	\end{minipage}
	\caption{Visit distribution and correlations of the Safegraph Neighbor Patterns mobility network for New York city in November 2019. The complementary cumulative in-strength distribution function is plotted in (a) of the trip-visit probability density function with a clear scale-free behavior. As for the trip-visit correlations, the network shows a neutral assortativity (b) and a neutral clustering (c) computed on the normalized version of the OD adjacency matrix.}
	\label{fig_NYstat}
\end{figure}

\subsection{Network scaling topology}
A very important analysis of the mobility network, is to study how the trip costs affect the topology of network.
The number of trip arrivals at the $i$-th destination can be written as the in-degree $k_i=\sum_j A_{ij}$, meanwhile the visit strength of the $i$-th destination can be written as $\kappa_i=\sum_j C_{ij}A_{ij}$ where $C_{ij}$ is the entry of the Origin-Destination distance matrix $C$. The  in-degree and in-strength distribution have been plotted in Fig.\ref{fig_NYdegree} and Fig.\ref{fig_NYstrength} respectively, with a clear scale-free asymptotic behavior but with different power law coefficients. Such evidence suggests a very interesting aspect of visiting patterns where  trip cost weights have a significant effect of  on the mobility network structure. In particular, the relation between strength and degree of a location node can be written  the average strength of destinations with degree $k$ changes as:
\begin{equation}\label{eq_degreestrength}
\kappa (k)\sim k^{1+\delta}
\end{equation}
where the exponent $\delta$ represents the rescaling factor and $\delta=0$ occurs in the absence of correlations between the weight of links and the degree of nodes \cite{architectPNAS} so that  the strength of a node is simply proportional to its degree and the two quantities provide therefore the same information on the system. The action of some correlation in the weight can bring cases where $\delta\neq 0$. 
In such situation such relation induces a change in the scaling of the degree distribution (i.e. visit distribution) $P(k)\sim k^{-\mu_0}$ and the strength distribution (i.e. trip-visit distribution) $P(\kappa)\sim \kappa ^{-\mu}$ according to the relation:
\begin{equation}\label{eq_changescale}
\mu=\frac{\mu_0+\delta}{1+\delta}
\end{equation}
and in the case of $\delta=-1$, $P(\kappa)$ is a Delta distribution of constant strength. In the case of the data under study, there is a clear linear relation between strength and degree as in Fig.\ref{fig_NYdegreevs}, consequently the strength distribution shows a scale free coefficient $P(\kappa)\sim \kappa^{-\mu}$ different from the one in the degree distribution $P(k)\sim k^{-\mu_0}$ consistently to the trasformation in eq.\eqref{eq_changescale}.
\begin{figure}[!ht]
	\centering
	\begin{subfigure}[l]{0.4\textwidth}
		\centering
		\includegraphics[width=0.9\linewidth]{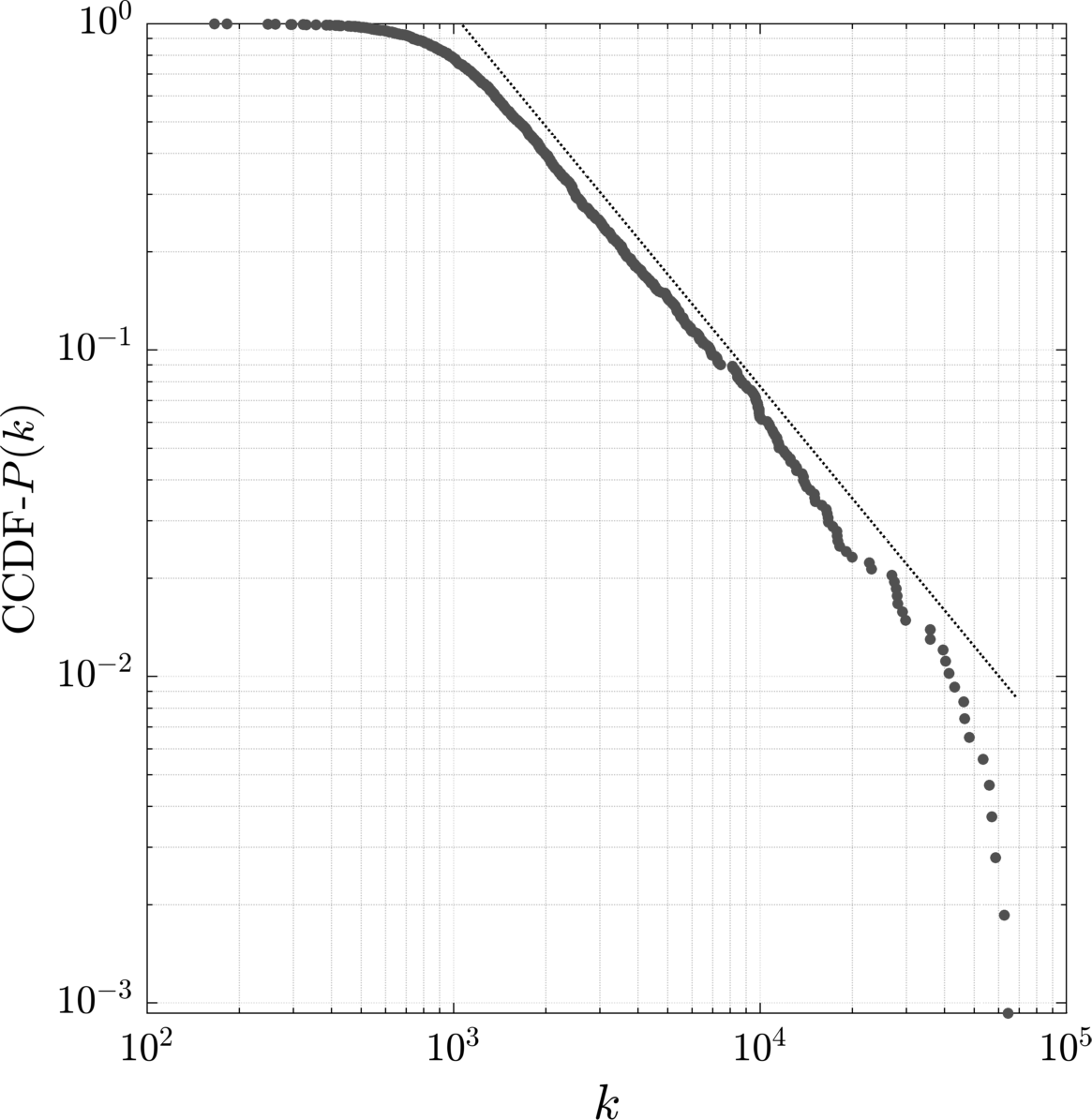}
		\caption{Trip arrival distributon (in-degree)}
		\label{fig_NYdegree}
	\end{subfigure}
	%	\\
	%	\vspace{1cm}
	\begin{subfigure}[r]{0.4\textwidth}
		\centering
		\includegraphics[width=0.9\linewidth]{img/NYtripvisit_Ps.png}
		\caption{Visit distribution (in-strength)}
		\label{fig_NYstrength}
	\end{subfigure}
	\\
	\vspace{1cm}
	\begin{subfigure}[r]{0.75\textwidth}
		\centering
		\includegraphics[width=1\linewidth]{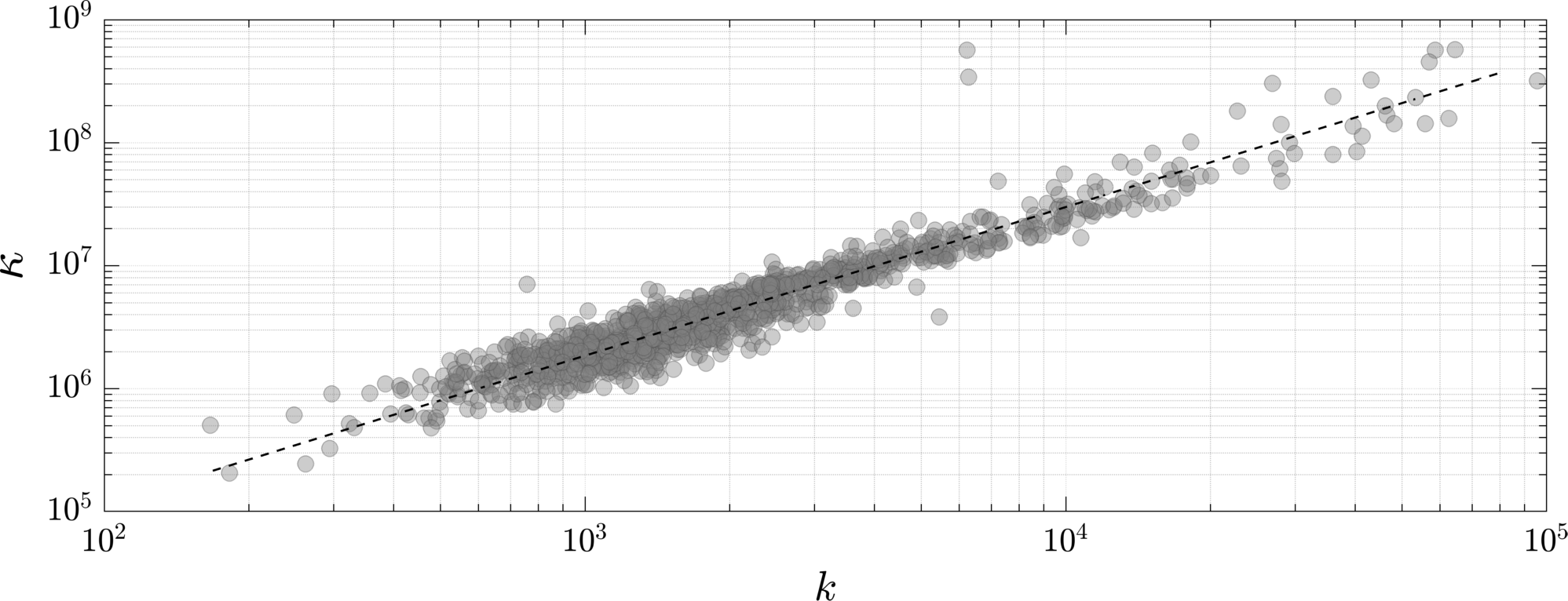}
		\caption{Degree vs Strength}
		\label{fig_NYdegreevs}
	\end{subfigure}	
	\caption{Visit patterns for in-strength defined in terms of distance traveled in trips. The linear fitting of degree-strength scatter plot provide a significant scaling coefficient as $\kappa\sim k^{1+0.53}$, so that the visit distribution asymptotic behavior is $P(k)\sim k^{-2.2}$ meanwhile the trip-visit distribution by distance weights is  $P(\kappa)\sim \kappa^{-1.8}$. }
\end{figure}\label{fig_strengthvsdegree}

The value of the rescaling exponent $\delta$ has been estimated through a regression analysis which also allows to perform the significance level for a linear relation between degree and distance-strength. As regarding with the data used in this study, the weight of a links correspond to the physical distance between the origin and the destination of the trip which defines the trip-visit distribution as indicated in Fig.\ref{fig_strengthvsdegree}. As reported in the caption, there is a neat change of the slope in the scale-free distribution, the visit-distribution based on the degree fig.\ref{fig_NYdegree} shows a slower power law coefficient respect to trip-visit distribution based on node strengths Fig.\ref{fig_NYstrength}. This is due to the linear relation between the  degree and its weighted version trough the origin-destination distances. Let us notice that the weights have an impact on the degree which is significant in determining a change in the scaling relation of the distribution. Such result of $\delta>0$ suggests that the strength of nodes grows faster than their degree, in other words, the trip distances associated to highly visited locations have higher values than those expected if the trip distances were assigned at random. Such tendency denotes a strong correlation between the trip weight and the topological properties in the mobility network, where the higher the number of visits in a location, the more traffic the location can handle. As a conclusion, a general scheme arises where the number of visitors to any destinations decreases as the inverse power law of the product of their visiting frequency and travel distance, as already suggested by \cite{universalschlapfer2021}.
On the contrary, random weights would have produced a $\delta$ close to zero, which occurs in the case when link weights are independent from the network topology, so that the strength distribution would carry no information than the degree distribution. For a more detailed estimate of the rescaling exponent $\delta$, a regression analysis is reported in table Tab.\ref{tab_regression}, for different types of trip weights. 
%This is mainly due to the fact that the survey dataset provides only aggregated information which are not so granular in the origin-destination links for the reasons explained above. However, travel time estimates represent a very good starting point for assessing the different effects that travel times can have on the mobility network structure. 
\begin{table}[ht] 
	\centering 
	\begin{threeparttable}
		{\renewcommand{\arraystretch}{1.15}%
			\begin{tabular}{c||ccc}
				\toprule 
				& {Distance}  & Travel time & Income\\
				\midrule    
				\\	[-2.5ex]
				\multirow{2}{*}{\textsl{$\delta$}}& 0.531$^{*}$ & 0.033$^{*}$ & 0.017$^{*}$ \\ 
				& {\small $[ 0.467,  0.595]$}  & {\small $[ 0.025  ,  0.042]$}  & {\small $[ 0.010  , 0.024]$}\\ [1ex]
				\midrule[0.015cm]
				\multirow{2}{*}{\textsl{$\alpha_0$}}& 0.842$^{*}$ & 0.842$^{*}$ &  0.842$^{*}$ \\ 
				& {\small $[0.790,0.894]$}  & {\small  $[0.790,0.894]$ } & {\small  $[0.790,0.894]$ }\\ [2ex]
		\multirow{2}{*}{\textsl{$\alpha$}}& 1.322$^{*}$  & 0.874$^{*}$ & 0.858$^{*}$ \\ 
	& {\small $[1.223,1.421]$}   & {\small  $[0.822,0.932]$ } & {\small  $[0.805,0.912]$ }\\ 
				\bottomrule 
		\end{tabular} }
		\begin{tablenotes} 
			\scriptsize 
			\item \leavevmode\kern-\scriptspace\kern-\labelsep Note for the linear fit the p-value of the linear regression:$^{*}$p$<$0.01.  The confidence intervals have been calculated at a significance level of $99\%$.
		\end{tablenotes}
	\end{threeparttable} 
	\caption{In-degree vs in-strength regression analysis for different trip-weights. The slope of  linear fit reported is the coefficient  $1+\delta$.} 
	\label{tab_regression}
\end{table} 
As a additional study to  the linear regression statistics, performing a residual analysis  makes it possible to test the assumption of a linear regression model such as the errors are independent and normally distributed, as shown in the Supporting Materials (SM3).   

%\begin{table}[H] 
%	\centering 
%	\caption{Degree vs Strength regression analysis} 
%	\label{tab_regression}
%	\begin{threeparttable}
%		{\renewcommand{\arraystretch}{1.15}%
%			\begin{tabular}{llll}
%				\toprule 
%				& \textbf{Distance}  &  \textbf{Income}  &  \textbf{Age} \\
%				\midrule    
%				\\	[-2.5ex]
%				\multirow{2}{*}{Slope}& 1.549$^{***}$ & 2.374$^{***}$ & 0.992$^{***}$ \\ 
%				& {\small $(0.024)$} &{\small  $(0.014)$ } &{\small  $(0.006)$ } \\ [2ex]
%				\multirow{1}{*}{{\small CI99}}& {\small $[ 1.502,  1.597]$} &  {\small $[2.338,  2.410]$} & {\small $[0.957, 1.007]$} \\ 
%				\bottomrule 
%		\end{tabular} }
%		\begin{tablenotes} 
%			\scriptsize 
%			\item \leavevmode\kern-\scriptspace\kern-\labelsep Note for the linear fit:$^{*}$p$<$0.10,$^{**}$p$<$0.05,$^{***}$p$<$0.01. The confidence intervals have been calculated at a significance level of $99\%$.
%		\end{tablenotes}
%	\end{threeparttable} 
%\end{table} 

\section{Economic applications: from land-use to travel demand}\label{sec_empirical}
In this section, the investigation on possible properties of latent variables will be useful to select which kernel function is suitable to meet the trip mobility network characteristics. Later, the latent variable will be proved to be a crucial key that justify the income elasticity  of a multi-purpose travel demand for any transportation mode. %In particular, it will be shown  that the travel demand is driven by a latent variable structure with the attraction rate of the form $\nu_x\sim x^{\alpha_0}$ and attractiveness probability as $\rho(x)\sim x^{-\eta}$
\subsection{The latent variable interpretation}
The model dynamic is driven by the presence of latent-variables that  are related to intrinsic properties of the areas but they cannot be not directly given (as for deterministic measurements) but they can only be inferred through statistical indicators (probabilistic measurement or proxy). 
At this point, it is possible to formulate the observed mobility network in terms of the latent-variable model so that the scale-free distribution shown in the real-world trip mobility network can be stated in terms of the latent variable statistical attributes.  In particular,  the attractiveness variable $x$ and the productiveness variable $y$ have been considered as hidden variables in the visitation  model. The first summarizes the notion of the ability of a destination to attract visitors, the latter describes the ability of an origin location to produce and generate travelers with specific characteristics.
However, latent variables can be interpreted more like proxies or indexes rather than proper measurements of observed phenomena. Under such perspective, there are several studies which  try to estimate the number of visits by a combination of urban features, such as job opportunities, retail shops, business activities, infrastructural capacity, geographical positions etc...
The standard approach to categorize urban areas classifies regions by their physical features and land use which refers to the way in which land is utilized, developed, and transformed for different purposes such as residential, commercial, industrial, and agricultural purposes.  Urban morphology is seen as the result of dynamic interactions between multiple factors, such as transportation efficiency, population size, and local land use. In particular, regional movement patterns, and consequently travelers' distribution, can be explained from land use, since purposes of people’s trips are strongly correlated with the land use of the trip’s origin and destination \cite{lee2015relating,shi2020prediction,akiba2022correlation,rodrigue2020geography}.  
For example, the latent-variable framework can provide an interesting interpretation of the effect of trip distances on the visit distribution, which can be formalized in terms of the attractiveness latent variable model. 
It is out of scope of this work to detect the best combination of factors which define the attractiveness of locations, however, how supported by some studies \cite{greenwald2006relationship,wagner2007urban,acheampong2015land}, let us take the non-residential land-use of census blocks to be the primary cause of travel demand and so it could be consider a proxy of destination attractiveness $x$ in a multi-purpose travel model. In such perspective, the data from the New York Open Data \cite{pluto} has been used where the land use zones is reported as the square feet occupied by building, parks and areas with a given use of destination (except residential), see the Supplementary Materials for more details on data used and correspondent interpretations. Land-use can be  seen as a "ceteris paribus" candidate for attractiveness of locations and in Fig.\ref{fig_landuse}, the plot shows the probability density function of square feet of land-use lots and it shows a scale-free behavior with a power law exponent of $2$ as plotted in Fig.\ref{fig_landuse}. 
\begin{figure}[!ht]
	\centering
	\begin{subfigure}[l]{0.95\textwidth}
		\centering
		\includegraphics[width=0.75\linewidth]{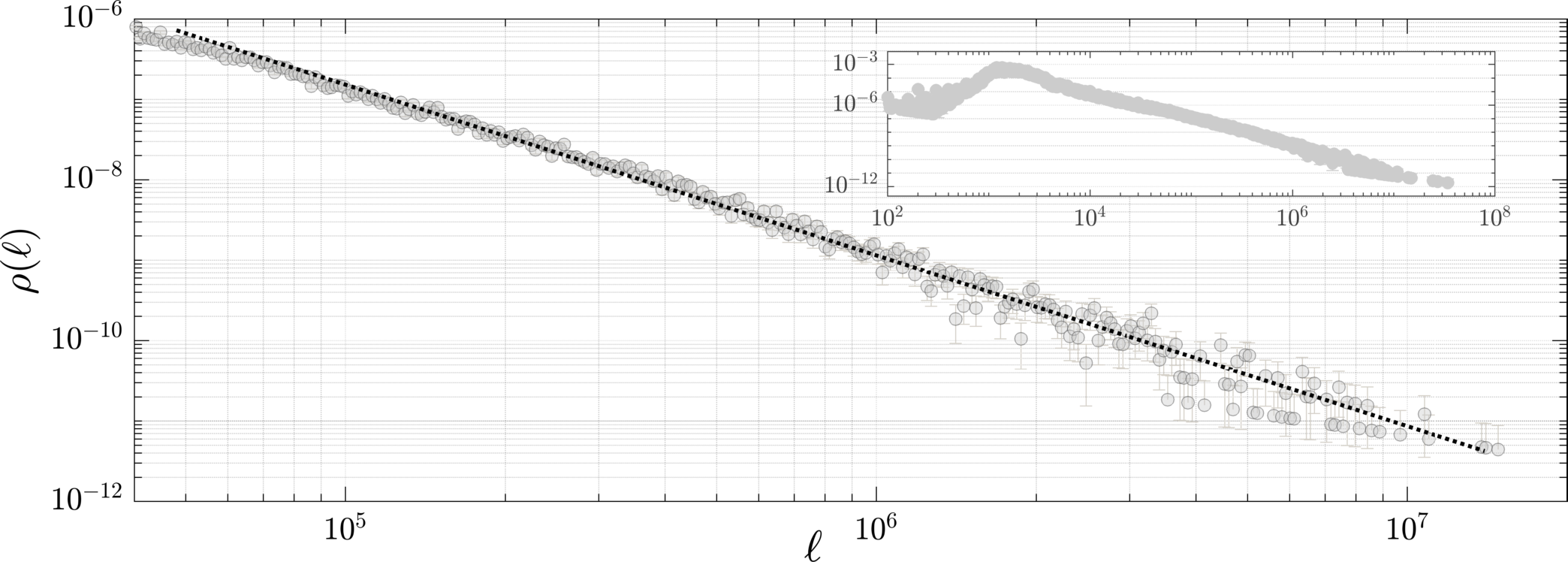}
		\caption{Tax lot square feet for non-residential land use}	\label{fig_landuse}
		%\label{fig_clustersimultime}
	\end{subfigure}
\\
	\vspace{20pt}
		\centering
	\begin{subfigure}[r]{0.95\textwidth}
		\centering
		\includegraphics[width=0.75\linewidth]{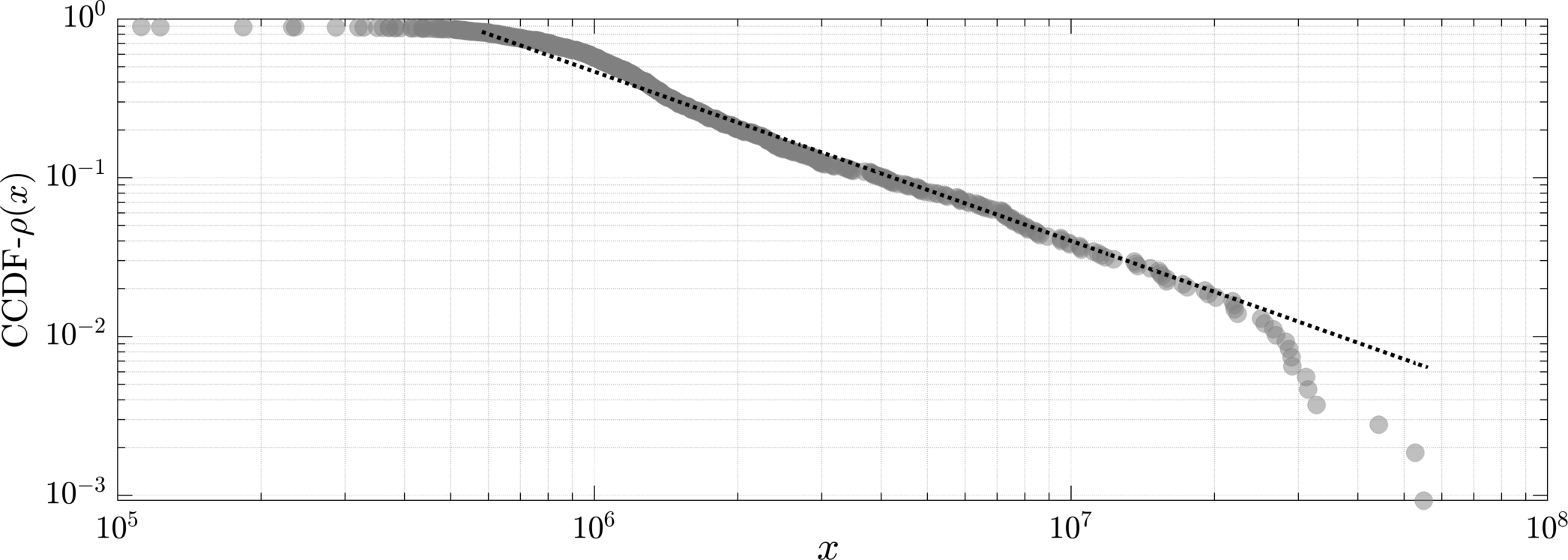}
		\caption{Tax lot square feet aggregated in census block groups}\label{fig_landusedensity}
	\end{subfigure}%     
	\caption{Log-binning procedure for the probability density function of square feet for land-use zones as hypothetical index for destination attractiveness. It shows an inverse power law fat-tail distribution so that the asymptotic behavior the probability density function can be written as $\rho(x)\sim x^{-2}$. In the inset the whole distribution is plotted where the data is fully represented even at a low spatial scale.  In (b) the complementary cumulative density function of the land-use square feet aggregated for census block from tax lot data.}
\end{figure}
The same analysis is performed after aggregating tax lots into census block groups,  the land-use areas for non-residential purposes keep the same asymptotic fat-tail distribution with an inverse power law probability density function   $\rho(x)\sim x^{-\eta}$ with $\eta \approx 2$ as plotted in Fig.\ref{fig_landusedensity}.

\begin{figure}[!ht]
	\centering
	\begin{subfigure}[c]{0.95\textwidth}
		\centering
		\includegraphics[width=0.76\linewidth]{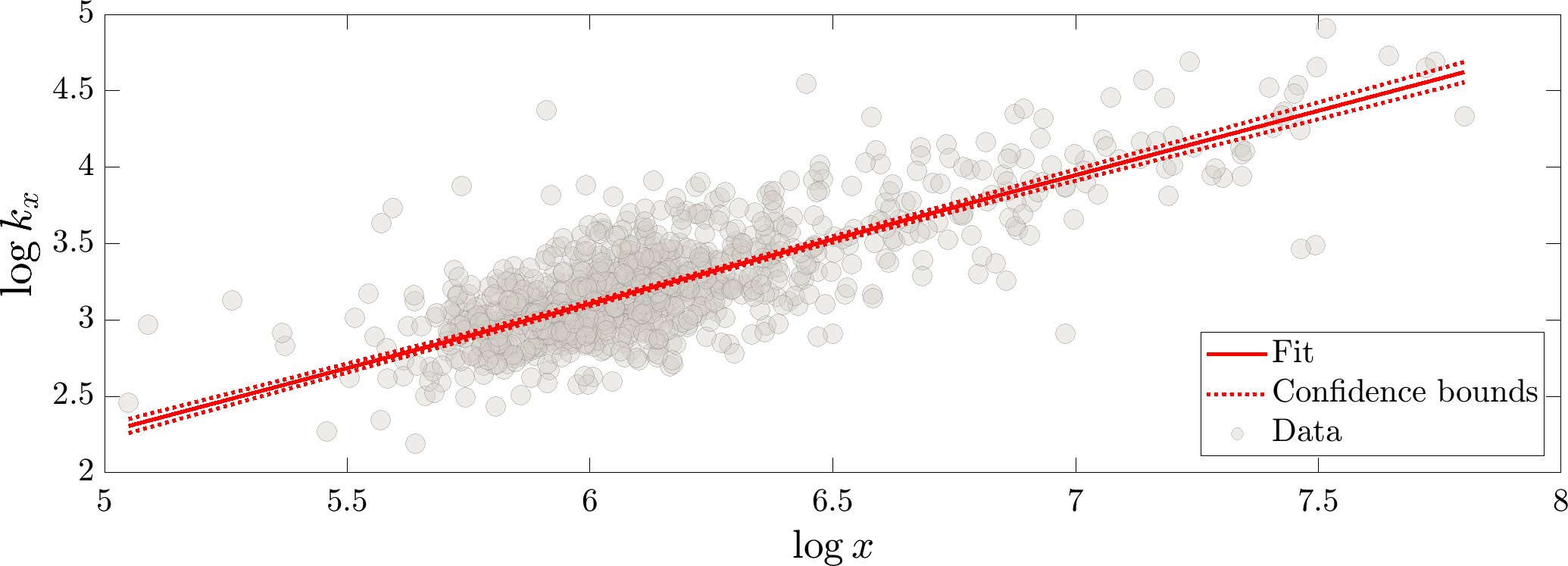}
		\caption{Scatter log-log plot of the destination with land-use $x$ versus the number of arrivals  $k_x $ to reach that destination. The slope is $\alpha_0\approx 0.84$. }\label{fig_landusedistance}
	\end{subfigure}%    
	\\
	\vspace{20pt}
	\centering
	\begin{subfigure}[c]{0.95\textwidth}
		\centering
		\includegraphics[width=0.75\linewidth]{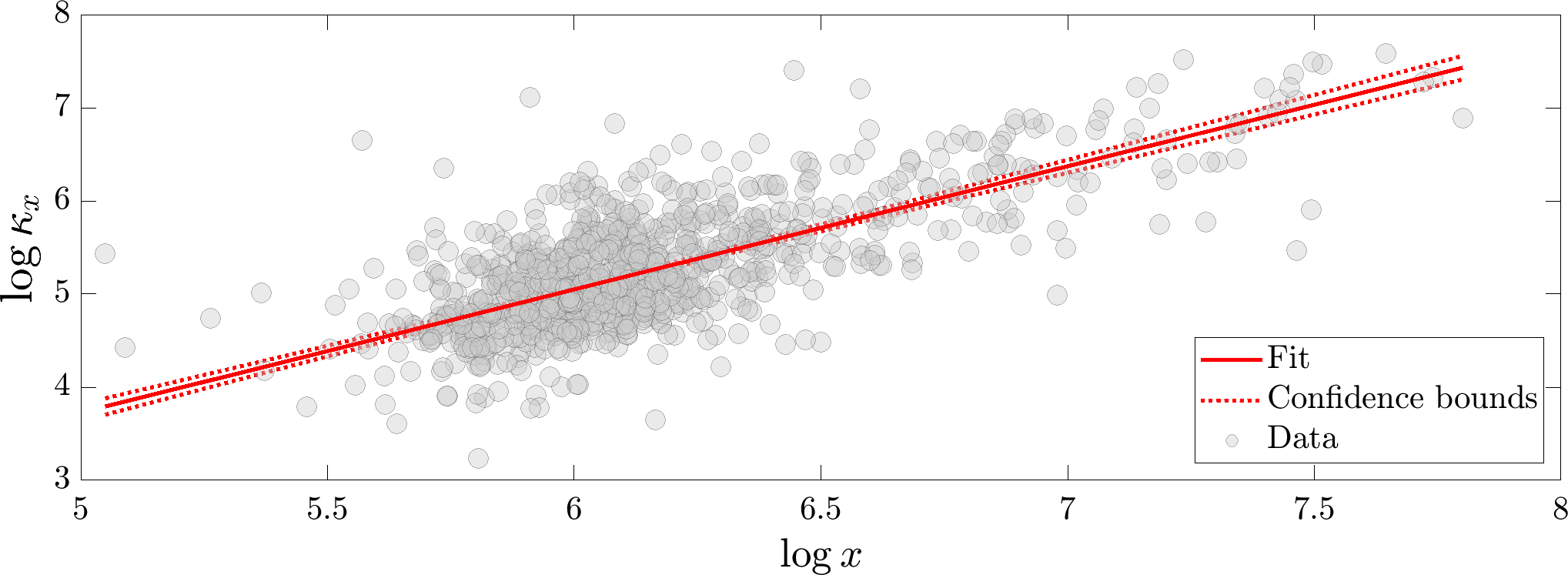}
		\caption{Scatter log-log plot of the destination with land-use $x$ versus the number of visits $\kappa_x$ (i.e. number of arrivals weighted by distance traveled)  to reach that destination. The slope is $\alpha\approx 1.32$. }\label{fig_landusetrip}
	\end{subfigure}%   
	%\hspace{2\textwidth}
	%\vspace{20pt}
	\caption{Regression fit between the land-use of census block group and (a)  arrivals and (b) visits. }
	\label{fig_landusevisit}
\end{figure}\label{fig_landuse_analyze}
 At this point it is possible to investigate the relation between mobility and land-use data (as attractiveness indicator). First, it is possible to analyze the relation between the number of arrivals $k$ in each destination with the non-residential land-use of that area $x$. Similarly the relation of land-use versus the visits $\kappa$ as weighted arrivals is analyzed as well. So,   let us investigate  a log-log linear regression analysis of such variables. The linear fit analysis of the relation $k_x\sim x^{\alpha_0}$ reveals an estimated value $\alpha_0\approx 0.842$ as shown in Fig.\ref{fig_landusedistance} and linear fit analysis of the relation $\kappa_x\sim x^{\alpha}$ reveals an estimated value $\alpha\approx 1.322$ as shown in Fig.\ref{fig_landusetrip}. In table \ref{tab_regression} such estimations for $\alpha_0$ and $\alpha$ are reported for different types of trip weights. By knowing that $\kappa_x=\mathbb{E}[\kappa|x]\propto \langle \mathcal{r}\rangle_x  \nu_x\sim x^{\theta+\alpha_0}$,  it is possible estimate the scaling of the trip size goes as  $\langle \mathcal{r}\rangle_x \sim x^{\theta}$, as confirmed by a direct  linear fit analysis where $\langle \mathcal{r}\rangle_x \propto x^{\theta}$ with $\theta \approx 0.48$. In such circumstances, by using the latent variable framework one can recover the visit distribution as $P(\kappa) \sim  \kappa^{-(1+\frac{\eta-1}{\alpha_0+\theta})}$, which is the same scale-free distribution directly observed during the analysis of the origin-destination network.
It is worth noticing the relation between the scaling exponents $\theta$ and $\delta$. It can be easily verified that the strength distribution in eq\eqref{eq_changescale} has $\mu= {\frac{\mu_0+\delta}{1+\delta} =1+\frac{\eta-1}{\alpha}}$, where $\mu_0=1+\frac{\eta-1}{\alpha_0}$ and $\alpha=\alpha_0 + \theta$. Solving, we find that, the following relation holds:
\begin{equation}
\delta=\frac{\alpha}{\alpha_0}-1=\frac{\theta}{\alpha_0}
\end{equation}
which can be checked by comparing the values reported in the Table \ref{tab_regression} under their relative error margins.
 In conclusion, the attraction rate of a location is higher than another destination in the sense that the travelers are so motivated to travel a longer distance to get there. Such effect reveals the action of human travel demand on mobility, as formulated in the present paper, in terms of latent variable network model. The scaling relation between land use versus both travel distances and visits are power law like so that the attraction rate must be of a power law type and the distribution of land use that is the attractiveness latent variable has a pareto- type distribution. So no other attraction rate is compatible with the observation. Morever since degree correlation are absent, it is plausible that $\chi(y|x)=\chi(y)$ so that the kernel function $\mathcal{K}(x,y)$ is multiplicative separable function. %In conclusion, the land-use as latent variable for attractiveness of locations is compatible, yet more general, with most of the macroscopic models of human mobility such as gravity model and radiation model.

%It is possible to slightly modify the land-use zones multiplying it  by an occupancy factor which determines the density of individuals who are able or are allowed to stay in a given zone. This factor can be assumed as a constant over different types of areas, but in different years it can be regulated by policies or social needs, as in the case of regulations in pandemic times, or during very intense economic crises, where density of individual can show important changes.

\subsection{Income elasticity}
The visitation model approach can provide economical interpretations of some empirical  urban scaling evidences \cite{bettencourt2021introduction}, where scaling laws are also present in economical values of each location respect to its attractiveness. A crucial attention will be focused on the {income elasticity of visit demand}. 
%	The net social income profit of travel demand can be written as the settlement 'benefit' of receiving the visits  over the cost of being in the settlement \cite{bettencourt2013origins,lobo2013urban}. The numerator $I_x$ is given by the number of visits multiplied by a factor $ g_0$ which converts a visit into a potential economic output, such as  monetary or labor. These numbers attribute a value to each interaction, which can vary by type and over time. On the opposite side, the denominator is given by the cost faced by travelers to reach the settlement and it can be assumed that commuting costs are proportional to distance, so $c(d)=c_0 d$ where $c_0$ is the cost of transportation per unit distance traveled\footnote{a lower $c_0$ correspond to a  better transportation system}. so when $\alpha_0 >\delta$, the locations  have increasing returns respect to their own attractiveness. In conclusion, empirical trend insights for a urban economic efficiency can be recovered  by the scaling relation between visits and social income profit as $E\sim  k^{1-\frac{\delta}{\alpha_0}}$. As long as the attraction scaling coefficient $\alpha_0$  is larger than the visiting cost scaling coefficient $\delta$ the urban region shows a virtuous economic trend.
%	\begin{figure}[!ht]
%	\centering
%	%		\begin{subfigure}[l]{0.35\textwidth}
%	\centering
%	\includegraphics[width=0.5\linewidth]{img/Income_trend.png}
%	\caption{}
%	\label{fig_income}
%	%	\end{subfigure}
%\end{figure}
	Let us, first, define the "benefit" that travelers received by visiting a location, and in particular, the variable  $I_x$  indicates the income level associated to the visitors who have traveled towards a given location with attractiveness $x$ for a job purposes. In this way the benefit can be determined trough the strength-by-income variable $\kappa_i$ which  converts a visit into potential economic output through a conversion factor $i_0$ that is the traveler's income per unit of visit time \footnote{The factor $i_0$ is an average conversion which attributes a benefit-value to each trip, but this value depends on the  distribution among possible trip purposes (as business, consumption, leisure, ...) for each traveler by their origin locations. If origins and destinations are, again, randomly independent $i_0$ can be considered a mean-field constant correction factor. Further it can change over the observation time during which data have been collected. Considering $i_0$ to be location independent it can be considered a  proportionality constant.}.
	On the other side, let us define the "cost" $Q_x$ faced by travelers to reach a location as the strength-by-distance variable $\kappa_q$  taking into consideration that  commuting costs are proportional to distance or time traveled, and the proportionality conversion factor $c_0$ is the cost of transportation per unit of quantity traveled\footnote{Quantity can be measures in distance or travel time. A lower $c_0$ correspond to a  better transportation system. Let us observed that in the case of cities the transportation costs and distance are well approximated by a linear relation. However for larger areas, empirical studies \cite{felbermayr2022trade,carra2016modelling}  show that transport cost is an increasing and concave function of distance so that $c(\mathcal{r})=c_0 \mathcal{r}^{v}$ with $0<v<1$,  since travelers switch to faster transport	modes for longer trips.}. 	Finally, the relation between the two variables can be written in terms of the attractiveness variable as:
\begin{align}\label{eq_returns}
\frac{\text{travelers' income }(I_x)}{\text{travel quantity } (Q_x)}=\frac{i_0\mathbb{E}[\kappa_i|x]}{c_0\mathbb{E}[\kappa_{q}|x]} \propto\frac{i_0\,x^{\alpha_0 + \theta_I}}{c_0\,x^{\alpha_0 + \theta_Q}}
%e^{(0)}_x=&\frac{\text{traveler income }(I_x)}{\text{traveled distance } (Q_x)}=\frac{i_0\mathbb{E}[i|x]}{c_0\mathbb{E}[{\mathcal{r}}|x]}= \frac{i_0\,x^{\delta'}}{c_0\,x^{\delta}} \sim x^{\delta'-\delta}
\end{align}
where $\theta_Q$ is the scaling exponent derived from regression slope by income in Table \ref{tab_regression} and the exponent $\theta$ from regression slope by amount of travel given by distance traveled in the same analysis. The relation between three variables is represented graphically in Fig.\ref{fig_incomescalingv}. 
%So, since $1.54=\delta' >\delta=0.54$, the locations  offer increasing returns respect to their own attractiveness, i.e. visitors in more attractive locations have a net-income higher than the visitors who travel to less attractive locations.
%Another scaling behavior, can be assessed by using the  relation between visits and the net income profit as $E_k\sim  k^{1+\frac{\delta'-\delta}{\alpha_0}}$ directly from the mobility network and avoiding the use of the hidden-variable information.
%Those two results are confirmed by the fits plotted in Fig.\ref{fig_incomescaling} and Fig.\ref{fig_incomescalingv} respectively.  
At this point  it is possible to write the income elasticity of travel demand  which has the meaning of how sensitive the demand for traveling a certain distance is to changes in income levels, the direct relation between the income reward and the quantity of travel  demanded for visiting locations can be derived from eq.\eqref{eq_returns} as:
\begin{equation}\label{eq_inela}
Q_x\sim I_x^{\varepsilon} \qquad \text{ , with }  \;\varepsilon= \frac{\partial Q_x}{ Q_x}\frac{I_x}{\partial I_x}=\frac{\alpha _Q}{\alpha_I}=\frac{1+\delta_Q}{1+\delta_I} \quad , \forall x \in \Omega_{x}
\end{equation} 
where the exponent $\varepsilon$ is the income elasticity as estimated in Fig.\ref{fig_incomescaling2}, and in the case under study $\varepsilon <1 $,  indicating that distance traveled is a necessity good, i.e., as income increases, people spend proportionally less on traveling when the  income levels of travelers increase for any transportation mode\footnote{For any mode and any purpose condition, undistinguished transportation mode is considered here, so all the possible modes are combined and only the distance necessary  to reach the destination is taken into account. }.  Such prediction of an income elasticity of about $\varepsilon \approx 0.65$ can be compared to results presented in literature reviews of empirical evidences and meta-analysis studies \cite{litman2021understanding, mayeres2000efficiency,goodwin2004elasticities,borjesson2012income}  and also discussed with a theoretical interpretations \cite{carra2016modelling,ghoddusi2021income}, where average income elasticity for aggregated travel demand is estimated to be in a range values  compatible with the elasticity $\varepsilon$ estimated here\footnote{Let us notice that income elasticity shows a large variability in the empirical evidences since distance traveled is not homogeneous across different sources of income, type of jobs and age\cite{le2014does}. Moreover, travel demand is reported in different units (individual or aggregate distance km/day, travel time, fuel consumption) and in different behaviors (commuting vs non-commuting, essential vs non-essential) or travel purposes (business, job, shopping, leisure). Travel can also vary in terms of different traveling modes according to transportation infrastructure. Moreover the estimations reported can even change over the years.}.  It's important to notice that the income elasticity of travel demand can be influenced by many factors, such as the availability of transportation options, the price of transportation, and individual preferences. %Additionally, the income elasticity of travel demand can vary across different modes of transportation and different regions. 
%Moreover, income elasticity of distance traveled can vary depending on the specific context and population being studied, and it is not always straightforward to predict how changes in income will affect travel behavior. %In conclusion,  as long as the attraction scaling coefficient $\alpha_0$  is larger than the visiting cost scaling coefficient $\delta$ the urban region shows a virtuous economic trend.
\begin{figure}[!ht]
%	\centering
%	\begin{subfigure}[c]{0.9\textwidth}
%	\centering
%	\includegraphics[width=0.95\linewidth]{img/netincome_landuse_scaling.png}
%	\caption{}
%	\label{fig_incomescaling}
%	\end{subfigure}
%	\\
%	\vspace{0.75cm}
		\centering
	\begin{subfigure}[c]{0.9\textwidth}
		\centering
		\includegraphics[width=0.6\linewidth]{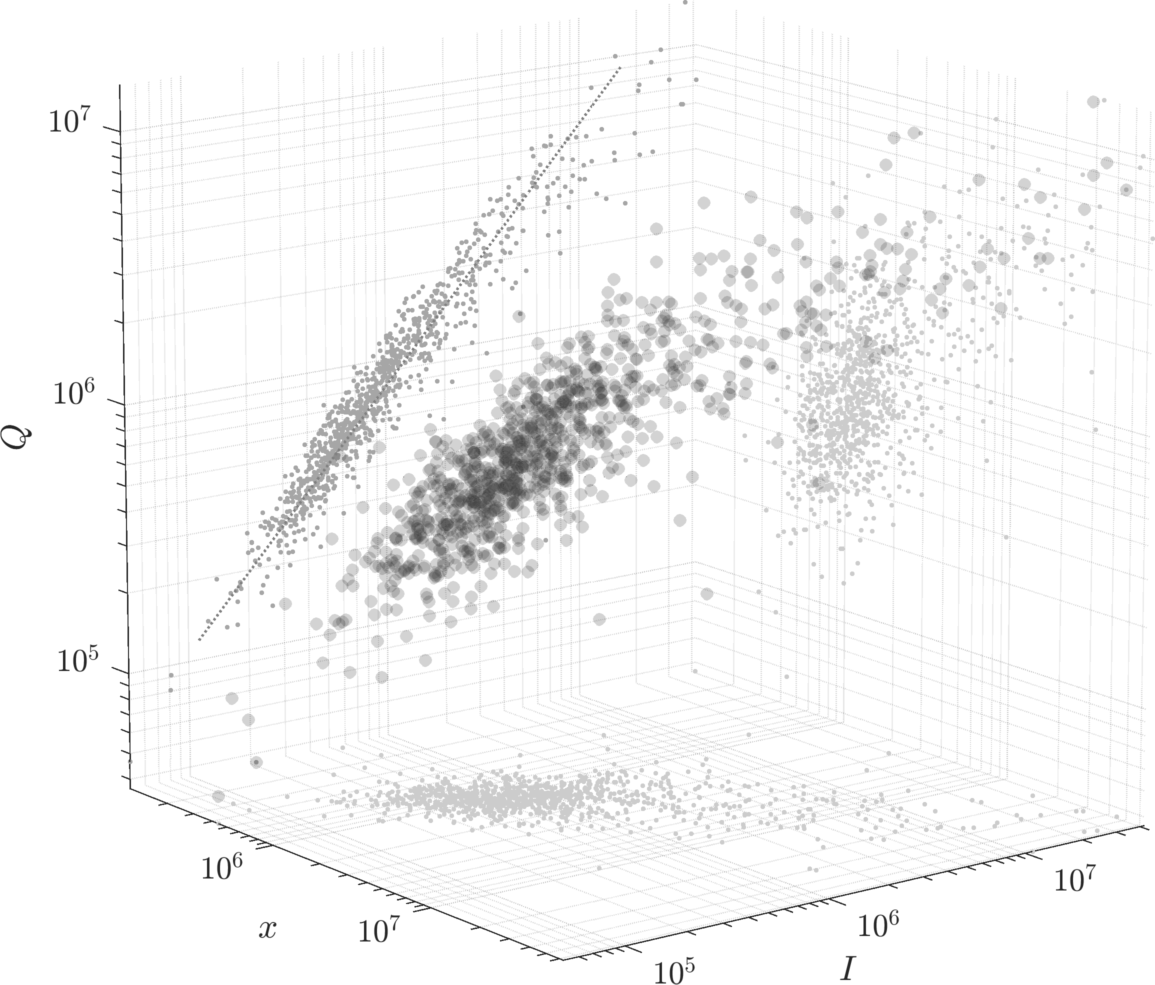}
		\caption{Three dimensional scatter plot for the variable $\{x,I_x,Q_x\}$.}
		\label{fig_incomescalingv}
	\end{subfigure}
	\\
\vspace{0.95cm}
\centering
\begin{subfigure}[c]{0.9\textwidth}
	\centering
	\includegraphics[width=0.94\linewidth]{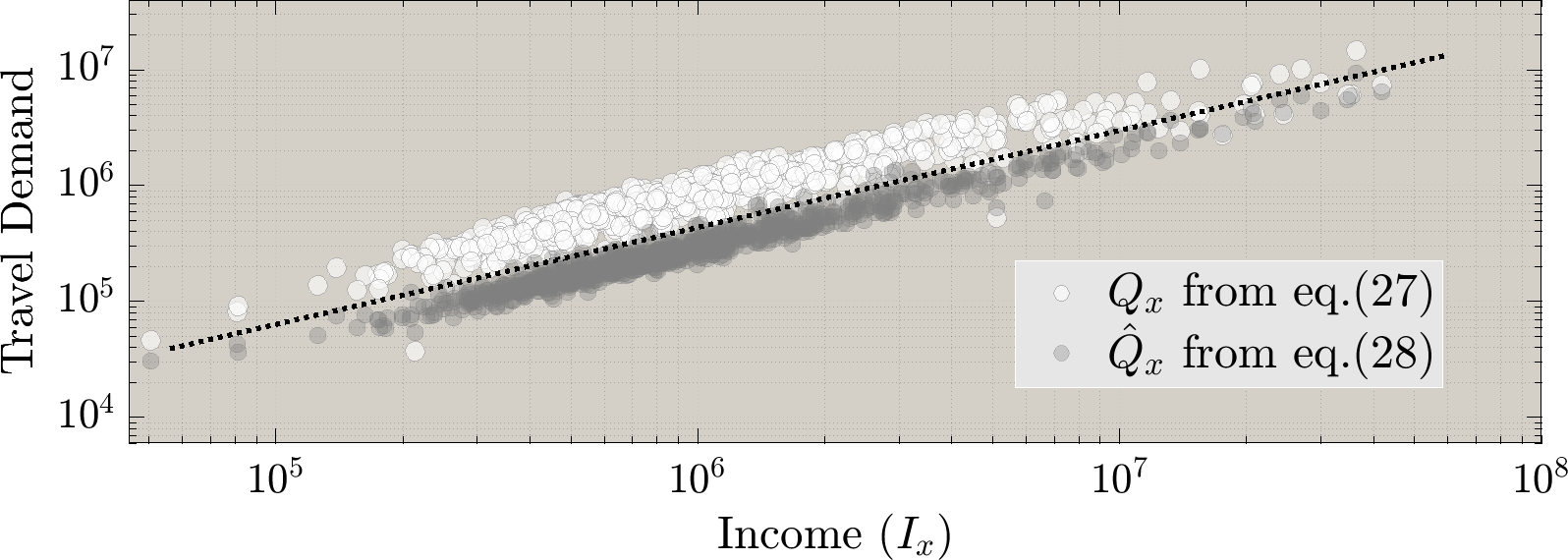}
	\caption{Scatter plot for the evaluation of the income elasticity of travel demand}
	\label{fig_incomescalingD}
\end{subfigure}
\caption{Using the dataset, for New York in 2019, in (a) a 3D scatter plot shows the relation among attractiveness $x$ as land-use, the income level of visitors and travel demand as distance traveled for each visit. The projection planes show three 2D scatter plots for pairs of the previous variables. In particular, the Q-I plane shows  the projected regression analysis   an elasticity of $\varepsilon\simeq 0.76$ consistently with the theoretical prediction from eq.\eqref{eq_inela} and from  eq.\eqref{eq_inela} where the income $I_x\sim k_x^{1+{\delta_I}/{\alpha_0}}$. This indicates that a 10 per cent increase in income leads to about 6.5 percent increase in distance traveled or, conversely, a 10 percent increase in the demand of travel distance requires an increase of the income  by about 15.4 percent.}
\label{fig_incomescaling2}
\end{figure}
In conclusion, the in eq.\eqref{eq_returns}  shed lights on an intrinsic relation  between the income elasticity of travel demand and the attractiveness scaling of urban areas. So urban planning policies, market conditions, and social dynamics emerges but also explains the interplay between attractiveness and economic through a mobility network flow which drives the travel behavior. 
From a network analysis perspective, it would be convenient to compare the  number of trip arrivals at destinations against the travelers' income. Using eq.\eqref{eq_returns} and eq.\eqref{eq_inela}, and visit elasticity of travel demand $\varepsilon'$ can be defined through:
\begin{align}
{Q}_x \sim k_x^{\varepsilon'}\qquad \text{ where }\; \varepsilon'=\left(1+\frac{\theta_I}{\alpha_0}\right)\varepsilon = (1+\delta_I)\varepsilon
\end{align}
where the income scales as $I_x\sim k_x^{1+{\theta_I}/{\alpha_0}}$ respect to the number of trip arrivals. The visit elasticity of travel demand is a measure of how sensitive the number of trips to a certain destination is to changes in key travel attributes, such as fare level, service quality, journey time components, income and car ownership, and price of competing modes.
%By using such expression, network centrality of destinations $\kappa_x$ and $\alpha_0$ together with the income impact of the visitors $I_x$ and $\delta_I$, it is possible to determine the amount of travel (in time or distance or cost) needed according to the income elasticity that acts as a control parameter for the policy regulators. 
%For example, for luxury transportation modes $\varepsilon>1$ the travel time have decreasing returns 
The previous scaling relations of income elasticity, it is possible to range from the microeconomics perspective of transport economics to the macroeconomic outputs such as employment and the economic growth.  For example, during a period of an economic growth (recession) where incomes are rising (falling) the distribution and the magnitudes of attractiveness can change and then modify the mobility pattern which, in its turn, has an impact on the global economic performance as well.  
%Another complementary concept associated to travel behavior is the income elasticity of travel supply which is often used to analyze the relationship between transportation infrastructure investment and economic development. 

\section{Conclusions and perspectives}
In conclusion, the study presents a data-driven model for human mobility network based on an origin-destination structure, which serves as the foundation for understanding mobility visitation flows. The model utilizes latent variables associated with each location, representing attractiveness and productiveness, to capture the intrinsic characteristics of destinations and origins.
The contributions of this research are twofold. Firstly, it provides a theoretical framework that describes and reproduces a visit generation stochastic process through the use of intergral-differential equation for the evolution of the  visit probability density and degree correlations in a trip mobility network.  Consequently, analytical, numerical, and computational solutions are provided for important network characteristics, such as the strength distribution, assortativity, and clustering coefficient. A collateral impact of the visit generation model in travel dynamics is the introduction of a mathematical formalism of compound renewal processes commonly used in financial and actuarial science literature. Such stochastic models in a network perspective are useful for capturing the randomness and dependence between the arrival times and event sizes, providing a flexible framework for modeling and analyzing various mobility dynamics and financial and actuarial risks as well.  
%The study draws upon the theory of continuous-time Markov processes description of visit generation process that utilizes a Kolmogorov-Feller partial differential equation. 
%This analytical approach allows for the characterization and analysis of the visit distribution of the .
The second contribution of this research has to do with the empirical analysis of real world phenomena. By analyzing  origin-destination data, the study reveals the presence of scale-free behaviors in visit frequencies and identifies correlations between visits and trip costs. The research also explores the statistical characteristics that latent variables should have to reproduce the observed patterns in the trip mobility network. % By using land-use as proxy of locations' attractiveness, the mobility graph will be driven by such variable through scaling laws which involves the income elasticity of travel demand.  
Hence, the model permits to disentangle the effect of attractiveness (as land-use), population,  trip costs and economic features of travelers on the visit dynamics in a  mobility network. Clear scaling laws  emerges between latent variables and travel demand. Finally, the model points out the effect of income on travel behavior depends strictly from the latent-variables that, therefore, can be considered as decision variable from a economic policy viewpoint. The possibility that human mobility belongs to the class of scale-free networks has impacted on the economic, engineering and  mathematical  communities in the multidisciplinary field of sustainable urban transportation, smart cities and world trade webs.   % In transport economics,  elasticity of travel demand is important for understanding the impacts of various policy measures on travel behaviour and welfare. 
%For example, higher attractiveness reflects an higher demand of visits in that locations so that the economics mechanism of prices for transportation and other services related to the urban movements as also market prices of residential areas and retail business locations. 
%So, the analysis proposed here can prompt policy makers to invest in infrastructure for alternative travel modes in order to alleviate the energy and time cost burden placed by car traffic. 
% Furthermore, a challenging issue is the use of the network for the identification of movement patterns, such as the most popular destinations or the most common routes taken by people according to the urban structure and how  planning  or policy or geography and environmental and climate have an impact on the transportation network and consequently on the social and economical status of the population.
As future outlooks, from a modeling side, the mathematical formalism discussed in the paper could also be extended to more general mobility graphs by considering more granular interactions on a time interval much shorter than the one used in the data explored here. For example, one can investigate more on the distribution of inter-arrival times of new trips to be non-Poissonian, and the the trip size are events with intensity that does not have finite moments (Levy-jumps).  Another forthcoming study will be the study of emission mobility reductions for sustainable city planning. The approach can help the decision for the optimal allocation of economically attractive urban areas by  transportation mode and land-use polices, and at the same time, minimizing the emission  of pollutants and the impact of mobility trips among origin and destinations.  
From the data-side, moreover, a larger scale panel analysis is required for other regions to increase the robustness and the validity of the model.
Moreover, the study shed lights on possible applications in fields like urban,  transportation and environmental economics.   These findings offer valuable insights to policymakers and urban planners, aiding in the comprehension and prediction of mobility patterns for informed decision-making and sustainable development.

%%%%%%%%
%Imagine the city is shaping its network structure in the way each link-trip has a weight that is the distance of a representative trip, we can find the asymptotic mean field behavior of the distance–frequency distribution of visitation flows of the trip mobility network.

	\section{CRediT authorship contribution statement}
	\textbf{Fabio Vanni}: conceptualization, methodology, formal, numerical and statistical analysis, investigation, data curation, writing – original draft, Revision. 
	
	\section{Declaration of Competing Interest}
The author declares that he has no known competing financial interests or personal relationships that could have appeared to influence the work reported in this paper.

\section{Data availability}
Origin Destination mobility data has been retreived by SafeGraph for academic research program \cite{safegraph}, and they cannot made public available.  Social, demographic and geographical data are public accessible at \cite{Tiger,censuscommute,censustiger,pluto}.  Codes for the main network analysis are povide in the Github repository \cite{Fabiogit}.
	%%%%%%%%%%%%%%%%%%%%%%%%%%%%%%%%%%%%%%%%%%%%%%%%%%%%%%%%%%%%%%%%%%%%%%%%%%%%%%%%%%%%%%%%%%%%%%

%\begin{table}[!ht]
%	\centering
%	\begin{tabular}{c|ccc|c}
%		{\footnotesize Distance}	& \textit{Coefficient} & \textit{Standard Error} &\textit{ t statistic} & \textit{p-value}  \\ 
%		\midrule 
%		(Intercept) & 0.70404 & 0.033489 & 21.0231 & 5.3427e-56\\ 
%		$x$ & 0.11205 & 0.035321 & 3.1722 & 0.0017134  \\
%		\bottomrule 
%	\end{tabular}
%	\caption{ Google social gathering }
%\end{table}

	%\bibliographystyle{elsarticle-num} 
%	\bibliographystyle{plainnat} 
\bibliographystyle{unsrt} 	
	\bibliography{cas-refs}

\begin{thebibliography}{10}

\bibitem{bettencourt2013origins}
Lu{\'\i}s~MA Bettencourt.
\newblock The origins of scaling in cities.
\newblock {\em science}, 340(6139):1438--1441, 2013.

\bibitem{bettencourt2021introduction}
Lu{\'\i}s~MA Bettencourt.
\newblock Introduction to urban science, 2021.

\bibitem{barbosa2018human}
Hugo Barbosa, Marc Barthelemy, Gourab Ghoshal, Charlotte~R James, Maxime
  Lenormand, Thomas Louail, Ronaldo Menezes, Jos{\'e}~J Ramasco, Filippo
  Simini, and Marcello Tomasini.
\newblock Human mobility: Models and applications.
\newblock {\em Physics Reports}, 734:1--74, 2018.

\bibitem{solmaz2019survey}
G{\"u}rkan Solmaz and Damla Turgut.
\newblock A survey of human mobility models.
\newblock {\em IEEE Access}, 7:125711--125731, 2019.

\bibitem{arcaute2022recent}
Elsa Arcaute and Jos{\'e}~J Ramasco.
\newblock Recent advances in urban system science: Models and data.
\newblock {\em Plos one}, 17(8):e0272863, 2022.

\bibitem{van2016random}
Remco Van Der~Hofstad.
\newblock {\em Random graphs and complex networks}, volume~43.
\newblock Cambridge university press, 2016.

\bibitem{bollobas2007phase}
B{\'e}la Bollob{\'a}s, Svante Janson, and Oliver Riordan.
\newblock The phase transition in inhomogeneous random graphs.
\newblock {\em Random Structures \& Algorithms}, 31(1):3--122, 2007.

\bibitem{kim2018review}
Bomin Kim, Kevin~H Lee, Lingzhou Xue, and Xiaoyue Niu.
\newblock A review of dynamic network models with latent variables.
\newblock {\em Statistics surveys}, 12:105, 2018.

\bibitem{rastelli2016properties}
Riccardo Rastelli, Nial Friel, and Adrian~E Raftery.
\newblock Properties of latent variable network models.
\newblock {\em Network Science}, 4(4):407--432, 2016.

\bibitem{hartle2021dynamic}
Harrison Hartle, Fragkiskos Papadopoulos, and Dmitri Krioukov.
\newblock Dynamic hidden-variable network models.
\newblock {\em Physical Review E}, 103(5):052307, 2021.

\bibitem{balogh2019generalised}
S{\'a}muel~G Balogh, P{\'e}ter Pollner, and Gergely Palla.
\newblock Generalised thresholding of hidden variable network models with
  scale-free property.
\newblock {\em Scientific reports}, 9(1):1--10, 2019.

\bibitem{caldarelli2002scale}
Guido Caldarelli, Andrea Capocci, Paolo De~Los~Rios, and Miguel~A Munoz.
\newblock Scale-free networks from varying vertex intrinsic fitness.
\newblock {\em Physical review letters}, 89(25):258702, 2002.

\bibitem{vanni2021incremental}
Fabio Vanni.
\newblock Incremental formation of scale-free fitness networks.
\newblock {\em arXiv preprint arXiv:2106.02168}, 2021.

\bibitem{boguna2003class}
Mari{\'a}n Bogun{\'a} and Romualdo Pastor-Satorras.
\newblock Class of correlated random networks with hidden variables.
\newblock {\em Physical Review E}, 68(3):036112, 2003.

\bibitem{hoppe2013mutual}
K~Hoppe and GJ~Rodgers.
\newblock Mutual selection in time-varying networks.
\newblock {\em Physical Review E}, 88(4):042804, 2013.

\bibitem{feller1991introduction}
William Feller.
\newblock {\em An introduction to probability theory and its applications,
  Volume 2}, volume~81.
\newblock John Wiley \& Sons, 1991.

\bibitem{denisov2019exact}
SI~Denisov and Yu~S Bystrik.
\newblock Exact stationary solutions of the kolmogorov--feller equation in a
  bounded domain.
\newblock {\em Communications in Nonlinear Science and Numerical Simulation},
  74:248--259, 2019.

\bibitem{jang2021review}
Jiwook Jang and Rosy Oh.
\newblock A review on poisson, cox, hawkes, shot-noise poisson and dynamic
  contagion process and their compound processes.
\newblock {\em Annals of Actuarial Science}, 15(3):623--644, 2021.

\bibitem{privault2013stochastic}
Nicolas Privault.
\newblock {\em Stochastic finance: An introduction with market examples}.
\newblock CRC Press, 2013.

\bibitem{van2007effects}
Veronique Van~Acker, Frank Witlox, and Bert Van~Wee.
\newblock The effects of the land use system on travel behavior: a structural
  equation modeling approach.
\newblock {\em Transportation planning and technology}, 30(4):331--353, 2007.

\bibitem{greenwald2006relationship}
Michael~J Greenwald.
\newblock The relationship between land use and intrazonal trip making
  behaviors: Evidence and implications.
\newblock {\em Transportation Research Part D: Transport and Environment},
  11(6):432--446, 2006.

\bibitem{lee2015relating}
Minjin Lee and Petter Holme.
\newblock Relating land use and human intra-city mobility.
\newblock {\em PloS one}, 10(10):e0140152, 2015.

\bibitem{geurs2004accessibility}
Karst~T Geurs and Bert Van~Wee.
\newblock Accessibility evaluation of land-use and transport strategies: review
  and research directions.
\newblock {\em Journal of Transport geography}, 12(2):127--140, 2004.

\bibitem{acheampong2015land}
Ransford~A Acheampong and Elisabete~A Silva.
\newblock Land use--transport interaction modeling: A review of the literature
  and future research directions.
\newblock {\em Journal of Transport and Land use}, 8(3):11--38, 2015.

\bibitem{song2010modelling}
Chaoming Song, Tal Koren, Pu~Wang, and Albert-L{\'a}szl{\'o} Barab{\'a}si.
\newblock Modelling the scaling properties of human mobility.
\newblock {\em Nature physics}, 6(10):818--823, 2010.

\bibitem{Fabioetal2013}
D.~R. Fabio, D.~Fabio, and P.~Carlo.
\newblock Profiling core-periphery network structure by random walkers.
\newblock {\em Sci. Rep.}, 3:1467, 2013.

\bibitem{universalschlapfer2021}
Markus Schl{\"a}pfer, Lei Dong, Kevin O’Keeffe, Paolo Santi, Michael Szell,
  Hadrien Salat, Samuel Anklesaria, Mohammad Vazifeh, Carlo Ratti, and
  Geoffrey~B West.
\newblock The universal visitation law of human mobility.
\newblock {\em Nature}, 593(7860):522--527, 2021.

\bibitem{nie2021understanding}
Wei-Peng Nie, Zhi-Dan Zhao, Shi-Min Cai, and Tao Zhou.
\newblock Understanding the urban mobility community by taxi travel trajectory.
\newblock {\em Communications in Nonlinear Science and Numerical Simulation},
  101:105863, 2021.

\bibitem{le2012urban}
Florent Le~N{\'e}chet.
\newblock Urban spatial structure, daily mobility and energy consumption: a
  study of 34 european cities.
\newblock {\em Cybergeo: European Journal of Geography}, 2012.

\bibitem{wang2020review}
Anqi Wang, Anshu Zhang, Edwin~HW Chan, Wenzhong Shi, Xiaolin Zhou, and Zhewei
  Liu.
\newblock A review of human mobility research based on big data and its
  implication for smart city development.
\newblock {\em ISPRS International Journal of Geo-Information}, 10(1):13, 2020.

\bibitem{bassolas2019hierarchical}
Aleix Bassolas, Hugo Barbosa-Filho, Brian Dickinson, Xerxes Dotiwalla, Paul
  Eastham, Riccardo Gallotti, Gourab Ghoshal, Bryant Gipson, Surendra~A
  Hazarie, Henry Kautz, et~al.
\newblock Hierarchical organization of urban mobility and its connection with
  city livability.
\newblock {\em Nature communications}, 10(1):4817, 2019.

\bibitem{akiba2022correlation}
Yuri Akiba, Masaya Yamasaki, Hiroyuki Shima, and Motohiro Sato.
\newblock Correlation analysis between land use and urban street patterns.
\newblock {\em Available at SSRN 4111251}, 2022.

\bibitem{fouquet2012trends}
Roger Fouquet.
\newblock Trends in income and price elasticities of transport demand
  (1850--2010).
\newblock {\em Energy Policy}, 50:62--71, 2012.

\bibitem{le2014does}
Scott Le~Vine, Bingqing~Emily Chen, and John Polak.
\newblock Does the income elasticity of road traffic depend on the source of
  income?
\newblock {\em Transportation Research Part A: Policy and Practice}, 67:15--29,
  2014.

\bibitem{litman2021understanding}
Todd~Alexander Litman.
\newblock Understanding transport demands and elasticities-how prices and other
  factors affect travel behavior.
\newblock 2021.

\bibitem{ghoddusi2021income}
Hamed Ghoddusi, Alexander Rodivilov, and Mandira Roy.
\newblock Income elasticity of demand versus consumption: Implications for
  energy policy analysis.
\newblock {\em Energy Economics}, 95:105009, 2021.

\bibitem{kang2020multiscale}
Yuhao Kang, Song Gao, Yunlei Liang, Mingxiao Li, Jinmeng Rao, and Jake Kruse.
\newblock Multiscale dynamic human mobility flow dataset in the us during the
  covid-19 epidemic.
\newblock {\em Scientific data}, 7(1):1--13, 2020.

\bibitem{Tiger}
{{Kepler.gl}{}}.
\newblock Tiger/line shapefiles, 2021.
\newblock data retrieved by Esri ,
  \url{https://www.arcgis.com/apps/mapviewer/index.html?url=https://services5.arcgis.com/GfwWNkhOj9bNBqoJ/ArcGIS/rest/services/NYC_Census_Blocks_for_2010_US_census_Water_Included/FeatureServer/0}.

\bibitem{janson2011random}
Svante Janson, Andrzej Rucinski, and Tomasz Luczak.
\newblock {\em Random graphs}.
\newblock John Wiley \& Sons, 2011.

\bibitem{bollobas1998random}
B{\'e}la Bollob{\'a}s.
\newblock Random graphs.
\newblock In {\em Modern graph theory}, pages 215--252. Springer, 1998.

\bibitem{ben2020semi}
Omri Ben-Eliezer, Dan Hefetz, Gal Kronenberg, Olaf Parczyk, Clara Shikhelman,
  and Milo{\v{s}} Stojakovi{\'c}.
\newblock Semi-random graph process.
\newblock {\em Random Structures \& Algorithms}, 56(3):648--675, 2020.

\bibitem{hanneke2010discrete}
Steve Hanneke, Wenjie Fu, and Eric~P Xing.
\newblock Discrete temporal models of social networks.
\newblock {\em Electronic journal of statistics}, 4:585--605, 2010.

\bibitem{zhang2017random}
Xiao Zhang, Cristopher Moore, and Mark~EJ Newman.
\newblock Random graph models for dynamic networks.
\newblock {\em The European Physical Journal B}, 90(10):1--14, 2017.

\bibitem{barrat2008dynamical}
Alain Barrat, Marc Barthelemy, and Alessandro Vespignani.
\newblock {\em Dynamical processes on complex networks}.
\newblock Cambridge university press, 2008.

\bibitem{berberan2007computation}
M{\'a}rio~N Berberan-Santos.
\newblock Computation of one-sided probability density functions from their
  cumulants.
\newblock {\em Journal of Mathematical Chemistry}, 41(1):71--77, 2007.

\bibitem{everitt2013finite}
Brian Everitt.
\newblock {\em Finite mixture distributions}.
\newblock Springer Science \& Business Media, 2013.

\bibitem{sundt2009recursions}
Bj{\o}rn Sundt and Raluca Vernic.
\newblock {\em Recursions for convolutions and compound distributions with
  insurance applications}.
\newblock Springer Science \& Business Media, 2009.

\bibitem{tulsyan2016particle}
Aditya Tulsyan, R~Bhushan Gopaluni, and Swanand~R Khare.
\newblock Particle filtering without tears: A primer for beginners.
\newblock {\em Computers \& Chemical Engineering}, 95:130--145, 2016.

\bibitem{handschin1969monte}
Johannes~Edmund Handschin and David~Q Mayne.
\newblock Monte carlo techniques to estimate the conditional expectation in
  multi-stage non-linear filtering.
\newblock {\em International journal of control}, 9(5):547--559, 1969.

\bibitem{privault2022introduction}
Nicolas Privault.
\newblock {\em Introduction to Stochastic Finance with Market Examples}.
\newblock CRC Press, 2022.

\bibitem{tankov2003financial}
Peter Tankov.
\newblock {\em Financial modelling with jump processes}.
\newblock Chapman and HallCRC, 2003.

\bibitem{armenter2014balls}
Roc Armenter and Mikl{\'o}s Koren.
\newblock A balls-and-bins model of trade.
\newblock {\em American Economic Review}, 104(7):2127--2151, 2014.

\bibitem{riccaboni2013global}
Massimo Riccaboni, Alessandro Rossi, and Stefano Schiavo.
\newblock Global networks of trade and bits.
\newblock {\em Journal of Economic Interaction and Coordination}, 8:33--56,
  2013.

\bibitem{casiraghi2021configuration}
Giona Casiraghi and Vahan Nanumyan.
\newblock Configuration models as an urn problem.
\newblock {\em Scientific Reports}, 11(1):13416, 2021.

\bibitem{chung2002connected}
Fan Chung and Linyuan Lu.
\newblock Connected components in random graphs with given expected degree
  sequences.
\newblock {\em Annals of combinatorics}, 6(2):125--145, 2002.

\bibitem{chung2002average}
Fan Chung and Linyuan Lu.
\newblock The average distances in random graphs with given expected degrees.
\newblock {\em Proceedings of the National Academy of Sciences},
  99(25):15879--15882, 2002.

\bibitem{servedio2004vertex}
Vito~DP Servedio, Guido Caldarelli, and Paolo Butt{\`a}.
\newblock Vertex intrinsic fitness: How to produce arbitrary scale-free
  networks.
\newblock {\em Physical Review E}, 70(5):056126, 2004.

\bibitem{masuda2004analysis}
Naoki Masuda, Hiroyoshi Miwa, and Norio Konno.
\newblock Analysis of scale-free networks based on a threshold graph with
  intrinsic vertex weights.
\newblock {\em Physical Review E}, 70(3):036124, 2004.

\bibitem{fujihara2010universal}
A~Fujihara, M~Uchida, and H~Miwa.
\newblock Universal power laws in the threshold network model: A theoretical
  analysis based on extreme value theory.
\newblock {\em Physica A: Statistical Mechanics and its Applications},
  389(5):1124--1130, 2010.

\bibitem{di2022score}
Domenico Di~Gangi, Giacomo Bormetti, and Fabrizio Lillo.
\newblock Score driven generalized fitness model
  for$\backslash$$\backslash$sparse and weighted temporal networks.
\newblock {\em arXiv preprint arXiv:2202.09854}, 2022.

\bibitem{latora2017complex}
Vito Latora, Vincenzo Nicosia, and Giovanni Russo.
\newblock {\em Complex networks: principles, methods and applications}.
\newblock Cambridge University Press, 2017.

\bibitem{serrano2007correlations}
M~Angeles Serrano, Marian Bogun{\'a}, Romualdo Pastor-Satorras, and Alessandro
  Vespignani.
\newblock Correlations in complex networks.
\newblock {\em Large scale structure and dynamics of complex networks: From
  information technology to finance and natural sciences}, pages 35--66, 2007.

\bibitem{murase2019sampling}
Yohsuke Murase, Hang-Hyun Jo, J{\'a}nos T{\"o}r{\"o}k, J{\'a}nos Kert{\'e}sz,
  and Kimmo Kaski.
\newblock Sampling networks by nodal attributes.
\newblock {\em Physical Review E}, 99(5):052304, 2019.

\bibitem{fardet2021weighted}
Tanguy Fardet and Anna Levina.
\newblock Weighted directed clustering: Interpretations and requirements for
  heterogeneous, inferred, and measured networks.
\newblock {\em Physical Review Research}, 3(4):043124, 2021.

\bibitem{fagiolo2007clustering}
Giorgio Fagiolo.
\newblock Clustering in complex directed networks.
\newblock {\em Physical Review E}, 76(2):026107, 2007.

\bibitem{boguna2002epidemic}
Mari{\'a}n Bogun{\'a} and Romualdo Pastor-Satorras.
\newblock Epidemic spreading in correlated complex networks.
\newblock {\em Physical Review E}, 66(4):047104, 2002.

\bibitem{jiang2014topological}
Bin Jiang, Yingying Duan, Feng Lu, Tinghong Yang, and Jing Zhao.
\newblock Topological structure of urban street networks from the perspective
  of degree correlations.
\newblock {\em Environment and Planning B: Planning and Design},
  41(5):813--828, 2014.

\bibitem{van2017local}
Remco Van Der~Hofstad, AJEM Janssen, Johan~SH Van~Leeuwaarden, and Clara
  Stegehuis.
\newblock Local clustering in scale-free networks with hidden variables.
\newblock {\em Physical Review E}, 95(2):022307, 2017.

\bibitem{clemente2018directed}
Gian~Paolo Clemente and Rosanna Grassi.
\newblock Directed clustering in weighted networks: A new perspective.
\newblock {\em Chaos, Solitons \& Fractals}, 107:26--38, 2018.

\bibitem{stegehuis2017clustering}
Clara Stegehuis, Remco van~der Hofstad, AJEM Janssen, and Johan~SH van
  Leeuwaarden.
\newblock Clustering spectrum of scale-free networks.
\newblock {\em Physical Review E}, 96(4):042309, 2017.

\bibitem{boguna2005generalized}
Mari{\'a}n Bogun{\'a} and M~{\'A}ngeles Serrano.
\newblock Generalized percolation in random directed networks.
\newblock {\em Physical Review E}, 72(1):016106, 2005.

\bibitem{serrano2006correlations}
M~{\'A}ngeles Serrano, Mari{\'a}n Bogu{\~n}{\'a}, and Romualdo Pastor-Satorras.
\newblock Correlations in weighted networks.
\newblock {\em Physical Review E}, 74(5):055101, 2006.

\bibitem{dorogovtsev2004clustering}
SN~Dorogovtsev.
\newblock Clustering of correlated networks.
\newblock {\em Physical Review E}, 69(2):027104, 2004.

\bibitem{kaiser2008mean}
Marcus Kaiser.
\newblock Mean clustering coefficients: the role of isolated nodes and leafs on
  clustering measures for small-world networks.
\newblock {\em New Journal of Physics}, 10(8):083042, 2008.

\bibitem{safegraph}
Safegraph.
\newblock \url{https://www.safegraph.com/}.
\newblock Free sample of neighborhood patterns data, retrieved on 2019.

\bibitem{martin2018origin}
David Martin, Christopher Gale, Samantha Cockings, and Andrew Harfoot.
\newblock Origin-destination geodemographics for analysis of travel to work
  flows.
\newblock {\em Computers, Environment and Urban Systems}, 67:68--79, 2018.

\bibitem{rodrigue2020geography}
Jean-Paul Rodrigue.
\newblock {\em The geography of transport systems}.
\newblock Routledge, 2020.

\bibitem{batty2013new}
Michael Batty.
\newblock {\em The new science of cities}.
\newblock MIT press, 2013.

\bibitem{scheaffer2011elementary}
Richard~L Scheaffer, William Mendenhall~III, R~Lyman Ott, and Kenneth~G Gerow.
\newblock {\em Elementary survey sampling}.
\newblock Cengage Learning, 2011.

\bibitem{keplerNY}
{{Kepler.gl}{}}.
\newblock New york city drop-off car rides, 2016.
\newblock data retrieved from Kepler.gl data,
  \url{https://github.com/uber-web/kepler.gl-data}.

\bibitem{censustiger}
U.s. census bureau tiger.
\newblock \url{https://www.census.gov/geo/maps-data/data/tiger-line.html}.
\newblock TIGER/Line Files and Shapefiles geographical data.

\bibitem{architectPNAS}
A.~Barrat, M.~Barthélemy, R.~Pastor-Satorras, and A.~Vespignani.
\newblock The architecture of complex weighted networks.
\newblock {\em Proceedings of the National Academy of Sciences},
  101(11):3747--3752, 2004.

\bibitem{clementecode}
Gian~Paolo Clemente and Rosanna Grassi.
\newblock Directedclustering: A package for computing clustering coefficient in
  weighted and directed networks, 01 2018.

\bibitem{shi2020prediction}
Shuyang Shi, Lin Wang, Shuangdie Xu, and Xiaofan Wang.
\newblock Prediction of intra-urban human mobility by integrating regional
  functions and trip intentions.
\newblock {\em IEEE Transactions on Knowledge and Data Engineering},
  34(10):4972--4981, 2020.

\bibitem{wagner2007urban}
Peter Wagner and Michael Wegener.
\newblock Urban land use, transport and environment models: Experiences with an
  integrated microscopic approach.
\newblock {\em DisP-the planning review}, 43(170):45--56, 2007.

\bibitem{pluto}
Newyork open data.
\newblock
  \url{https://data.cityofnewyork.us/City-Government/Primary-Land-Use-Tax-Lot-Output-PLUTO-/64uk-42ks}.
\newblock Primary Land Use Tax Lot Output (PLUTO).

\bibitem{felbermayr2022trade}
Gabriel~J Felbermayr and Alexander Tarasov.
\newblock Trade and the spatial distribution of transport infrastructure.
\newblock {\em Journal of Urban Economics}, 130:103473, 2022.

\bibitem{carra2016modelling}
Giulia Carra, Ismir Mulalic, Mogens Fosgerau, and Marc Barthelemy.
\newblock Modelling the relation between income and commuting distance.
\newblock {\em Journal of the Royal Society Interface}, 13(119):20160306, 2016.

\bibitem{mayeres2000efficiency}
Inge Mayeres.
\newblock The efficiency effects of transport policies in the presence of
  externalities and distortionary taxes.
\newblock {\em Journal of Transport Economics and Policy}, pages 233--259,
  2000.

\bibitem{goodwin2004elasticities}
Phil Goodwin, Joyce Dargay, and Mark Hanly.
\newblock Elasticities of road traffic and fuel consumption with respect to
  price and income: a review.
\newblock {\em Transport reviews}, 24(3):275--292, 2004.

\bibitem{borjesson2012income}
Maria B{\"o}rjesson, Mogens Fosgerau, and Staffan Algers.
\newblock On the income elasticity of the value of travel time.
\newblock {\em Transportation Research Part A: Policy and Practice},
  46(2):368--377, 2012.

\bibitem{censuscommute}
U.s. census bureau survey.
\newblock
  \url{https://data.census.gov/table?q=travel+time&tid=ACSDT1Y2019.B08303}.
\newblock American Community Survey Data at tlevel of ceneus block groups,
  B08303 TRAVEL TIME TO WORK, ACS 1-year estimates.

\bibitem{Fabiogit}
Vanni Fabio.
\newblock Project title.
\newblock \url{https://github.com/fxfabio/tripmobnet}, 2023.

\bibitem{willmot1986mixed}
Gord Willmot.
\newblock Mixed compound poisson distributions.
\newblock {\em ASTIN Bulletin: The Journal of the IAA}, 16(S1):S59--S79, 1986.

\bibitem{van2013critical}
Remco van~der Hofstad.
\newblock Critical behavior in inhomogeneous random graphs.
\newblock {\em Random Structures \& Algorithms}, 42(4):480--508, 2013.

\bibitem{scalas2012class}
Enrico Scalas.
\newblock A class of ctrws: compound fractional poisson processes.
\newblock In {\em Fractional dynamics: recent advances}, pages 353--374. World
  Scientific, 2012.

\end{thebibliography}
	
	
	
	
	
	% \end{thebibliography}

	%%% APPENDICI %%%%%%%%%%%
\renewcommand{\appendixname}{}
\appendix

\section*{APPENDIX}
In the following appendix, the visit generation process of the trip mobility network will be reframed in terms of a mixture of many compound poisson counting processes, as follows:
\begin{proposition}
	The probability distribution of the visit generation process can be restate as a mixture of compound counting poisson processes. 
	\begin{itemize}
		\item[$\circ$] Let us define the evolution of the trip arrival process looking at the generation of new visits as a counting stochastic process $\{\mathcal{C}_x(t),t\geq0\}$, defined on the probability space $(\mathbb{R},\mathcal{B}(\mathbb{R}),\mu_x)$, for all the locations of attractiveness $x$ where $\mathcal{C}_x(t)$ represents the number of visits which have arrived between time $0$ and $t$ at the destination nodes of attractiveness $x$, and  let the attraction rate $\nu_x=\mathrm{Prob}\{\mathcal{C}_x(t+\tau)=k | \mathcal{C}_x(t)=k-1 \}$ be the probability (per unit of infinitesimal time $\tau$) that a destination of attractiveness $x$  increase its visits from $k-1$ to $k$. Then, let $\mathcal{Z}_k$ be a sequence of independent, identically distributed square-integrable random variables, which represent jumps as trip-costs of the trip arrivals and they are distributed as a common random variable $\mathcal{r}_x$  with cumulative distribution function $F_{\mathcal{Z}}(\mathcal{r}_x)$. 
	\end{itemize} 
	Then, the conditional  visit probability density function $p_x(k,t)$ can be derived by a continuous-time stochastic process with jumps. In particular:
	\begin{itemize} 
		\item[1] 	The compound Poisson process $\{Y_x(t):t>0\}$ is  continuous-time stochastic process, adapted to a filtration $\mathcal{F}_t$, with random jumps of intensity $\nu_x$,  defined by:
		\begin{align}
		Y_x(t)=\sum_{k=1}^{\mathcal{C}_x(t)}\mathcal{Z}_k
		\end{align}
		Consequently, it turns out that the probability density function  can be explicitly written for a compound
		Poisson process as $p(\kappa,t|x)=f_{Y_x(t)}(\kappa,t|x)$ and the cumulative distribution function is given by:
		\begin{align}
		F_{Y_x(t)}(\kappa,t|x)=e^{-\nu_x t}\sum_{n=0}^{\infty}\frac{(\nu_{x} t)^n}{n!}F^{*n}_{{Y}_{n,x}}(\kappa|x)
		\end{align}
		where $F^{*n}_{{Y}_{n,x}}$ is the n-fold convolution of $F_{\mathcal{Z}}(\mathcal{r})$.
		Moreover, the characteristic function and the first two central moments  of $Y_x(t)$ are:
		\begin{align*}
		\mathbb{E}[e^{isY_x(t)}]&=\exp \{\nu_{x}t( \mathbb{E}[e^{is\mathcal{Z}}-1] )\} & \\
		\mathbb{E}[Y_x(t)]&=	\mathbb{E}[\mathcal{C}_x(t)]\,	\mathbb{E}[\mathcal{Z}]=\mathbb{E}_t[\kappa|x]=\nu_{x}\langle  \mathcal{r} \rangle_x t\\
		\mathbb{V}ar[Y_x(t)]&=\mathbb{V}ar[\mathcal{C}_x(t)]\,	\mathbb{E}[\mathcal{Z}^2] = \nu_{x}\langle  \mathcal{r}^2 \rangle_x t
		\end{align*}
		
		\item[2] 	The visit distribution is the mixture of stochastic poisson distributions $p_x(\kappa,t)$ \cite{willmot1986mixed} i.e. the conditional visit probability density function where the intensity of the (simple) counting process, $\nu_x$, is a random variable with a probability density function $\rho(x)$.
		The number of locations with visit strength $\kappa$, which we denote by $N_{\kappa}$ satisfies, see  \cite{bollobas2007phase,van2013critical}
		\begin{equation}
		N_{\kappa}/n \xrightarrow{\mathbb{P}} f_{\kappa}\equiv \mathbb{E}_x\left[ F'_{Y_x(t)}(\kappa,t|x) \right]
		\end{equation}
		where $ \xrightarrow{\mathbb{P}}$ denotes convergence in probability. We recognize the limiting distribution $\{f_k \}$  as a so-called mixed Poisson distribution with mixing distribution $\mu_x(x)$ .
	\end{itemize}
\end{proposition}
\begin{proof}
The proof is recovered from \cite{privault2022introduction,scalas2012class,tankov2003financial}. In particular, a compound Poisson counting process is equivalent to a Levy process and its realizations are piecewise constant cadlag functions.  As a consequence of the above results, the compound Poisson process enjoys 	 all the properties of Levy processes, including the Markov property. Compound Poisson processes are commonly used in insurance to model the number of claims and the size of each claim. Insurance companies can use this model to estimate their expected losses and set their premiums accordingly.
\end{proof}
	Moreover, in finance, Compound Poisson processes can be used to model stock prices, where the arrival times of price changes follow a Poisson process and the size of each change is a random variable.  Finally, in risk management, Compound Poisson processes can be used to model risk in various industries, such as energy trading, where the arrival times of price changes and the size of each change are both uncertain.
	
\end{document}